\documentclass[11pt]{article}

\usepackage[top=50mm, bottom=50mm, left=50mm, right=50mm]{geometry}
\usepackage{lineno}
\usepackage[authoryear,round]{natbib}
\usepackage{amssymb}
\usepackage{algorithm, algpseudocode}
\usepackage{algcompatible}
\usepackage{amsmath}
\usepackage{amsthm}
\usepackage{epsfig}
\usepackage{booktabs}
\usepackage{graphicx}
\usepackage{graphics}
\usepackage{float}
\usepackage{multirow}
\usepackage{color}
\usepackage{caption}
\usepackage{subcaption}
\usepackage{fullpage}
\usepackage[normalem]{ulem} 
\usepackage{makeidx}
\usepackage{xspace}
\usepackage{wrapfig}
\usepackage{tikz}
\usepackage{hyperref}
\hypersetup{
colorlinks=true,
citecolor=blue,
linkcolor=blue,
filecolor=blue,      
urlcolor=blue,
}
\usepackage{setspace}
\usepackage{bm}
\usetikzlibrary{arrows.meta}
\usepackage{enumitem}

\newtheorem{theorem}{Theorem}
\newtheorem{proposition}{Proposition}
\newtheorem{assumption}{Assumption}
\newtheorem{definition}{Definition}
\newtheorem{corollary}{Corollary}

\newtheorem{lemma}{Lemma}

\newtheorem{example}{Example}
\newtheorem{remark}{Remark}

\usepackage{dsfont}

\usepackage[affil-it]{authblk}

\newenvironment{keywords}{
  \vspace{1em}
  \noindent\textbf{Keywords:}
  \begin{itshape}
}{
  \end{itshape}
}

\usepackage{comment}
\usepackage{pagecolor}
\usepackage{mdframed}

\usepackage{xcolor}

\begin{document}

\title{\LARGE Design-Based Weighted Regression Estimators for Average and Conditional Spillover Effects}
\author[1]{Fei Fang} \author[1]{Laura Forastiere}
\affil[1]{Biostatistics, Yale School of Public Health, Yale University}
\date{\today}
\maketitle
\vspace{-1cm}

\setcounter{page}{0}
\thispagestyle{empty}

\linespread{1.5}\selectfont

\begin{abstract}
When individuals engage in social or physical interactions, a unit’s outcome may depend on the treatments received by others. In such interference environments, we provide a unified framework characterizing a broad class of spillover estimands as weighted averages of unit-to-unit spillover effects, with estimand-specific weights. We then develop design-based weighted least squares (WLS) estimators for
both average and conditional spillover effects. We introduce three nonparametric estimators under the dyadic, sender, and receiver perspectives, which distribute the estimand weights differently across the outcome vector, design matrix, and weight matrix. For the average-type estimands, we show that all three estimators are equivalent to the H\'{a}jek estimator.  For conditional spillover effects, we establish conditions under which the estimands are consistent for the target conditional spillover effects. We further derive concentration inequalities, a central limit theorem, and conservative variance estimators in an asymptotic regime where both the number of clusters and cluster sizes grow.
\end{abstract}


\section{Introduction}\label{s1}
When evaluating the effect of a policy or intervention, many causal inference methodologies rely on the Stable Unit Treatment Value Assumption (SUTVA; \citealp{rubin1980randomization}), which rules out \emph{interference} between units, that is, it assumes that a unit's outcome is only affected by its own treatment. 
When individuals or entities can interact with or observe one another, however, interference naturally arises. 
Such phenomena are pervasive in economics \citep{cai2015social, egger2022general, angelucci2009indirect}, 
social science \citep{airoldi2024induction, paluck2016changing}, 
business \citep{wager2021experimenting, ni2025interplay}, 
political science \citep{nickerson2008voting, bhatti2017voter}, 
and public health \citep{glass2006targeted, aiello2016design}. 
For instance, \citet{egger2022general} examine how cash transfers to some households affect the consumption and living standards of other households within the same village in rural Kenya; 
\citet{cai2015social} study how providing weather‐insurance information sessions to a subset of rice farmers influences their peers’ insurance‐purchasing behavior in rural China; 
and \citet{wager2021experimenting} investigate digital service platforms and marketplaces, where frequent interactions among customers, providers, and the platform itself give rise to  treatment spillovers. 


Such spillover effects often propagate through network connections or within clusters, motivating interference assumptions such as \emph{neighborhood interference}, in which a unit’s potential outcome is affected by the treatments of its neighbors \citep{ForastiereJASA2021, ogburn2022causal, weinstein2023causal}, and \emph{partial interference}, in which outcomes depend on the treatments of units within the same cluster \citep[e.g.,]{tchetgen2012causal, park2022efficient, qu2021semiparametric, dean2025effective}. 
Partial interference is often a conservative assumption in the presence of clustered data, even when interference is assumed to occur on a network but connections are not measured or the extent of interference is unknown.


When interference is present, a complete investigation of the causal impact of an intervention must account not only for the \emph{direct treatment effect}—the effect of a unit’s own treatment—but also for \emph{spillover effects}, which arise from the treatments received by others. Estimating spillover effects is crucial for several reasons. First, it enables accurate policy and program evaluation: ignoring interference leads to biased causal effect estimates and misleading assessments of intervention effectiveness \citep{benjamin2017spillover, sussman2017elements, ForastiereJASA2021}. Second, when interventions are subject to budget or allocation constraints, knowledge of spillover effects allows policymakers to reduce costs, while maintaining or even enhancing aggregate welfare, through targeted deployment \citep{czaller2022allocating, kitagawa2023should}. Third, accounting for spillover effects facilitates advances in policy improvement, for example by exploiting welfare gradients that depend on spillover effects in sequential or networked settings \citep{viviano2019policy, hu2022average, li2023experimenting, hu2025optimal}. Finally, quantifying heterogeneous spillover effects can reveal influential units or “key players’’ within networks, thereby guiding strategies for information diffusion, influence maximization, and targeted interventions \citep{he2025identifying, bargaglistoffi2023heterogeneoustreatmentspillovereffects, ji2025within}.

There are two main approaches to defining spillover effect estimands. The first approach, common under network interference, uses exposure mappings \citep{aronow2017estimating,leung2022causal,savje2024causal}--that is, functional forms describing how treatments of others influence one's potential outcomes. Then spillover effects are defined by contrasting potential outcomes under two values of the exposure mappings. 
The second approach, common under partial interference, defines spillover effect estimands through changes in the hypothetical treatment allocation applied to the whole cluster \citep{hudgens2008toward, tchetgen2012causal, papadogeorgou2019causal}.
In addition to these two approaches, a third option is to define the global average treatment effect as the contrast between outcomes when all units are treated and when all units are not treated. This estimand naturally coincides with the standard average treatment effect in the absence of interference, and it is especially appropriate in switchback experiments, where all units are either treated or not treated during each time period \citep{hu2022switchback, bojinov2023design}.

In this paper, we focus on spillover effects, defined as the effect of a unit's treatment on the average outcomes of a subset of units, or as the average effect on a unit's outcomes of the treatment received by another unit in a specific subset. 
A similar definition was first introduced by \cite{hu2022average} and then by \cite{lee2023finding}, whose causal estimand of interest is the average effect of a unit's treatment on the sum of outcomes in the whole sample.
Here, we rely on the partial interference assumption and define dyadic average potential outcomes by setting a unit's treatment status while assigning treatment to the rest of the cluster under a given hypothetical treatment assignment. In this way, we do not rely on a prespecified exposure-mapping function and are able to assess spillover effects under different treatment allocations.
Furthermore, we define our unit-level spillover effects conditional on the characteristics of the treated unit, which we refer to as the \textit{sender}.
Such estimands facilitate the identification of influential individuals within the population and, in turn, can inform the design of policy interventions.

The study of conditional or heterogeneous treatment effects under interference has recently gained significant attention. \citet{bong2024heterogeneous} propose a nonparametric kernel-smoothing estimator based on empirical pseudo-outcomes to estimate unit-level outcome expectations and, further, node-specific spillover effects. \citet{dean2025effective} introduce estimands that explicitly exploit heterogeneous interference across covariate profiles, allowing for more efficient individualized treatment decisions under partial interference. \citet{viviano2019policy} develop a welfare-optimization framework that incorporates heterogeneity in treatment effects arising from neighbors’ treatments and provide theoretical guarantees for the resulting policy’s regret. \citet{bargaglistoffi2023heterogeneoustreatmentspillovereffects} propose a network causal tree method to detect and estimate heterogeneous treatment and spillover effects under clustered network interference. \citet{qu2021semiparametric} construct generalized augmented inverse probability weighting estimators for heterogeneous direct and spillover effects driven by observed characteristics under partial interference. Our approach contributes to this literature by estimating heterogeneous spillover effects through simple regression-based methods that incorporate interaction terms between treatments and covariates in the design or weighting matrices. The consistency of our estimators relies on a flexible yet parametric model for dyadic average potential outcomes, which accommodates rich forms of treatment heterogeneity while preserving statistical efficiency in inference.

From an inferential perspective, we develop regression-based estimators for average and conditional spillover effects under partial interference in randomized experiments with design-based uncertainty. For simplicity, we assume a clustered structure and rely on the partial interference assumption. Nevertheless, our estimators can readily be adapted to alternative interference structures.
The design-based framework for randomized experiments with interference has received growing attention in the literature \citep{aronow2017estimating, wang2024designbasedinferencespatialexperiments, gao2025causalinferencenetworkexperiments}, where the only source of randomness arises from the treatment assignment mechanism. This framework is particularly appealing when the observed population coincides with the population of interest—for instance, when all villages within a state are included in the experiment—and it requires fewer assumptions on potential outcomes than the super-population perspective, such as weaker or no distributional restrictions on the error terms.

The motivation for adopting regression-based estimators is twofold. First, they are computationally scalable and therefore suitable for large-scale experimental data. Second, they offer an intuitive representation of spillover effects: one can view the (aggregated) outcomes of interest as being regressed on the (aggregated) treatments that influence them. This relationship can be formulated in several alternative ways. In the \emph{dyadic perspective}, a unit’s outcome is regressed on the treatment of another unit that affects it. In the \emph{effect-sender perspective}, an aggregated outcome of other units is regressed on a given unit’s treatment. In the \emph{effect-receiver perspective}, a unit’s outcome is regressed on a constant design matrix with a weight matrix encoding the aggregated treatments from the units that influence it. Together, these perspectives provide a unified and flexible regression-based approach to estimating spillover effects under interference.

The dyadic perspective is inspired by the dyadic regression estimators, which are used to analyze how dyadic characteristics affect dyadic outcomes involving pairs of agents. For example, dyadic outcomes and characteristics of interest may be the voting behavior among members of parliaments and their seating arrangements \citep{harmon2019peer}, or bilateral trade flows such as exports and imports among partner countries and their participation in WTO/GATT \citep{anderson2003gravity}. Because dyadic observations are not independent—two dyads may share a common node—the asymptotic theory for estimated coefficients in dyadic regression, including consistency, central limit theorems, and variance estimation, must account for this induced dependence \citep{aronow2015cluster, tabord2019inference, graham2020dyadic}. \citet{canen2024inference} further extend this framework by allowing for dependence between dyads that are indirectly connected through network paths. \citet{minhas2019inferential} employ additive and multiplicative effects models for dyadic outcomes to account for several forms of dyadic dependence: first-order dependence (two dyads sharing a common node), second-order dependence (reciprocal dyads), and third-order dependence (a dyad whose nodes appear in other dyads that share a common node). These dependence structures overlap with, but are not identical to, those addressed in \citet{canen2024inference}.

In our setting, we employ the dyadic regression framework in a different way and for a different purpose: the dyad here represents an outcome–treatment pair, entering the outcome vector and the design matrix, respectively. Nonetheless, our cluster-robust variance estimator under partial interference is conceptually related to that of \citet{canen2024inference}, since the dependence among units within a cluster can be viewed as a fully connected network, analogous to the dependence induced by connected network paths in their framework.

The regression estimator from the effect-receiver perspective corresponds to the approach commonly used to account for interference in regression settings. In this framework, one typically regresses a unit’s outcome on its own treatment and on some summary measure of the treatments received by others, such as the fraction or average treatment among neighbors (e.g., \citealp{soetevent2006empirics,davezies2009identification,cai2015social, bramoulle2020peer}). This approach generally relies on a linearity assumption linking others’ treatments to a unit’s potential outcomes. In contrast, our paper seeks to avoid such parametric restrictions by encoding the treatment information directly into the weight matrix of the estimator from the effect-receiver perspective.

The regression estimators from the effect-sender perspective are motivated by spillover estimands that capture the impact of a unit’s treatment on its neighbors’ outcomes \citep{fang2025inwardoutwardspillovereffects}. Both the estimand and the corresponding regression formulation from this perspective have received limited attention in the existing literature. A related idea appears in \citet{zigler2021bipartite}, who introduce the \( P \)-indexed average causal effect in the context of bipartite graphs, where the set of treated units does not overlap with the set of outcome units affected by those treatments. Similarly, \citet{wang2024designbasedinferencespatialexperiments} in a spatial setting and \citet{wang2021causal} in longitudinal and spatial contexts consider analogous estimands and regression formulations consistent with the effect-sender perspective. However, neither study investigates regression methods for conditional spillover effects within this framework.

There are three closely related studies on regression-based estimators under design-based uncertainty, both without and with interference. \citet{abadie2020sampling} provides a foundational theoretical analysis of regression-based estimators under the SUTVA assumption. Specifically, they derive the explicit form of the estimand that a regression estimator targets when potential outcomes have heterogeneous coefficients, and they establish both central limit theorems (CLTs) and conservative variance estimators in that setting. Our work extends this line of research beyond the SUTVA framework by incorporating partial interference. In addition, we analyze regression-based estimators by characterizing their properties for estimating conditional spillover effects.
\citet{sakamoto2025design} generalize \citet{abadie2020sampling} to settings with network interference, analyzing the behavior of regression estimators under both network sampling and design-based uncertainty. \citet{gao2025causalinferencenetworkexperiments} develop regression-based estimators for contrasts across exposure mappings under approximate neighborhood interference \citep{leung2022causal} and provide improved covariance estimation procedures. Although our work shares with \citet{sakamoto2025design} and \citet{gao2025causalinferencenetworkexperiments} the broader objective of conducting inference for regression-based estimators under interference, it differs from these contributions in several important respects:
(i) as opposed to \citet{sakamoto2025design}, who specify exposure mappings
and incorporate these mappings directly as regressors, we adopt a partial interference framework that does not impose a functional form for the dyadic average potential outcomes when estimating average spillover effects and assumes only a flexible parametric structure when analyzing conditional spillover effects, thereby allowing greater flexibility in modeling potential outcomes; 
(ii) we emphasize a unified representation of spillover estimands, develop regression-based estimators from multiple perspectives, and derive the conditions required for their consistency; (iii) our estimands are designed primarily for policy evaluation—focusing on the effects of hypothetical treatment assignments—rather than on contrasts across exposure levels.

Our contribution is fourfold. First, we introduce a general framework for representing spillover estimands, encompassing both average and conditional types. The estimands of interest are constructed as weighted averages of unit-to-unit spillover effects, i.e., the effect on a unit's outcome of altering another unit’s treatment status from control to treated (Definition \ref{pair_spillover}), where the weights define the subset of interest for the outcome units or the treatment units and may depend on unit-level covariates or on the underlying network structure.
By varying the estimand weights, this framework flexibly generates a broad class of estimands. Under this unified formulation, the corresponding regression-based estimators can be constructed uniformly across estimands by substituting the appropriate estimand weights.

Second, we develop three estimators corresponding to distinct perspectives: the \emph{dyadic}, \emph{effect-sender}, and \emph{effect-receiver} perspectives. When different estimands are considered, certain perspectives naturally align with specific estimands—for example, the average outward spillover effect aligns more closely with the effect-sender perspective, whereas the average inward spillover effect aligns with the effect-receiver perspective. Nevertheless, all three estimators are applicable to any estimand within the framework. We show that these estimators are equivalent and coincide with the H\'{a}jek estimator, a nonparametric and consistent estimator of the average spillover effect (ASE). 

Third, for the conditional spillover effect (CSE), we work with a parametric yet flexible specification of the dyadic-average potential outcomes (Definition \ref{struc_APO}). We characterize the relationships among the three CSE estimators and introduce intermediate quantities that bridge them to the target estimand. We delineate the conditions under which the dyadic, sender, and receiver estimators are consistent for the CSE, and we discuss the extent to which these conditions can be satisfied in practice. We also derive asymptotically conservative cluster-robust variance estimators for inference on the CSE.

Fourth, we establish consistency and asymptotic normality (CLT) for all three estimators—for both the ASE and the CSE—within a unified framework under partial interference, where both the number of clusters and the cluster sizes grow, a setting that contrasts with much of the partial-interference literature, which typically treats cluster sizes as fixed. This asymptotic regime can be readily extended to accommodate other forms of interference.

The remainder of the paper is organized as follows. Section \ref{sec:setup} introduces the setup and notation. Section \ref{sec:Estimand} defines a general class of estimands and illustrates several examples obtained by varying the estimand weights. 
Section \ref{sec:est_ase} develops three regression-based formulations of the ASE estimators, discusses their relationships, and establishes their consistency and asymptotic normality. We also derive asymptotically conservative cluster-robust variance estimators for the ASE estimators. Section \ref{sec:est_CSE} extends these formulations to incorporate conditioning covariates for the CSE, introduces the three estimators, and derives conditions for establishing their consistency and asymptotic normality. Section \ref{sec:Simulation Results} evaluates the performance of the proposed estimators for both ASE and CSE through simulation studies. Finally, Section \ref{sec:Real Data Application} applies the estimators to the data from \citet{cai2015social} to examine the average and conditional spillover effects of intensive information sessions on weather-insurance uptake in rural China.

\section{Setup}
\label{sec:setup}
We adopt a design-based framework in which randomness arises solely from the treatment assignment, whose mechanism is known, while the network structure and potential outcomes are treated as fixed. This framework is common in causal inference under randomized experiments \citep[e.g.,][]{imbens2015causal,abadie2020sampling} and is also commonly employed in interference settings \citep[e.g.,][]{aronow2017estimating,leung2022causal}. 
We consider \(K\) clustered networks, with the \(k\)-th cluster containing \(n_k\) units, indexed by \(i = 1, \dots, n_k\). The set of all units in cluster \(k\) is defined as \(\mathcal{N}_k = \{ik : i = 1, \dots, n_k\}\), and the set of all units across clusters is defined as \(\mathcal{N} = \cup_{k=1}^K \mathcal{N}_k\), with \(N = |\mathcal{N}|\). Within each cluster \(k\), the \(n_k\) units form a directed network denoted by \(\mathcal{G}_k = (\mathcal{N}_k, E_k)\), where \(E_k\) represents the set of directed edges among units in \(\mathcal{N}_k\). The overall network encompassing all clusters is denoted by \(\mathcal{G}\). In the population \(\mathcal{N}\), the experimenter assigns a treatment vector \(\mathbf{Z} := (Z_{11}, \dots, Z_{n_K K})\), with \(Z_{ik} \in \{0,1\}\) for each unit \(i = 1, \dots, n_k\) in cluster \(k = 1, \dots, K\). Let $\mathbf{Z}_k$ and $\mathbf{Z}_{-k}$ be the treatment subvectors in cluster $k$ and in the population excluding cluster $k$, respectively, with $\mathbf{z}_k$ and $\mathbf{z}_{-k}$ denoting the corresponding realizations. We consider the assignment mechanisms to be based on a known parameter (or vector of parameters) \(\beta\). For instance, in a Bernoulli experiment where the treatment is assigned independently and with constant probability, \(\beta\) simply represents this probability of treatment (type B parametrization in \cite{tchetgen2012causal}). The assignment mechanism $\beta$ may depend on covariates and is assumed to be (conditionally) independent across clusters. We denote by $\mathbb{P}_{\beta}(\mathbf{Z}_k)$ the probability of observing the treatment vector $\mathbf{Z}_k$ in cluster $k$ under the realized assignment mechanism $\beta$.



The potential outcome for unit \( i \) in cluster \( k \) is denoted by \( Y_{ik}(\mathbf{Z} = \mathbf{z}) \), or simply \( Y_{ik}(\mathbf{z}) \), where \( \mathbf{z} \) denotes a specific realization of the treatment vector. Throughout, we assume partial interference, which restricts the dependence of potential outcomes to the treatment vector within the same cluster, as formalized below \footnote{We focus on partial interference primarily for simplicity of exposition and to facilitate clear comparisons across different estimator formulations. The framework, however, is readily extensible to more general and heterogeneous interference structures, such as neighborhood interference \citep{JMLR:v23:18-711} or other interference assumptions based on the network \citep{leung2022causal}.
}.

\begin{assumption}[Partial interference]
\label{part_intf} 
For any $\mathbf{z}_{-k}, \mathbf{z}'_{-k} \in \{0,1\}^{\sum_{h=1}^K n_h - n_k}$, the potential outcome satisfies $Y_{ik}(\mathbf{z}_k, \mathbf{z}_{-k}) = Y_{ik}(\mathbf{z}_k, \mathbf{z}'_{-k})$ for $i = 1, \ldots, n_k$ and $k = 1, \ldots, K$. 
\end{assumption}

Under Assumption \ref{part_intf}, the potential outcome for unit $ik$ can be expressed as $Y_{ik}(\mathbf{z}_k)$ or $Y_{ik}(z_{jk}, \mathbf{z}_{-jk})$, where $\mathbf{z}_{-jk}$ denotes the treatment vector in cluster $k$ excluding unit $jk$. Let \( \mathbf{Z}_{-jk} \) denote the corresponding random treatment vector. 

\section{Estimands}
\label{sec:Estimand}
In this section, we introduce a general representation of spillover effects that allows for flexible weighting schemes, thereby inducing estimands of specific interest. This formulation enables the construction of a unified inference framework applicable to a variety of estimands. Our causal estimands are defined as weighted averages of the spillover effect from the treatment of unit $jk$ on the outcome of unit $ik$. Throughout the paper, unless otherwise specified, we refer to $jk$ as the \textit{effect sender} and to $ik$ as the \textit{effect receiver}.

We begin by defining the dyadic average potential outcome for unit $ik$ when the treatment of another unit $jk$ is kept fixed, while the treatment of the rest of the cluster, including that of unit $ik$, is assigned under a hypothetical treatment assignment governed by a known parameter \( \alpha \). This may or may not follow the same parametrization or take the same values as \( \beta \) and, as with the realized assignment mechanism, it is assumed to be (conditionally) independent across clusters. We further introduce two assumptions regarding the hypothetical and realized treatment assignments. These assumptions are necessary to ensure the consistency of the estimators of the average and conditional spillover effects, and to prevent the variances of these estimators from diverging as cluster sizes increase.

\begin{assumption}[Overlap between hypothetical and realized assignments]
\label{unif_bound_weight}
For all $\mathbf{z}_k \in \{0,1\}^{n_k}$ such that $\mathbb{P}_{\alpha}(\mathbf{Z}_k = \mathbf{z}_k) > 0$, it holds that $\mathbb{P}_{\beta}(\mathbf{Z}_k = \mathbf{z}_k) > 0$, for all $k \in\{ 1, \dots, K\}$.
\end{assumption}
Assumption \ref{unif_bound_weight} guarantees that the realized treatment assignment covers the entire support of the hypothetical assignment. For instance, suppose that the realized treatment assignment follows an i.i.d. Bernoulli distribution with parameter $\beta = 0.5$, and that the hypothetical treatment assignment corresponds to a completely randomized design in which, within each cluster $k \in \{1, \dots, K\}$, the number of treated units is given by the rounded value of $\frac{1}{2} n_k$. Under such circumstances, Assumption \ref{unif_bound_weight} is satisfied.
\begin{assumption}[Positivity of realized assignments]
\label{pos_Ratio}
For any $\mathbf{z}_k \in \{0,1\}^{n_k}$ such that $\mathbb{P}_{\beta}(\mathbf{Z}_k=\mathbf{z}_k) > 0$, there exists a constant $c > 0$ such that $\mathbb{P}_{\beta}(\mathbf{Z}_k=\mathbf{z}_k) \geq c$ for all $k \in \{1, \dots, K\}$.
\end{assumption}
Assumption \ref{pos_Ratio} imposes a uniform lower bound on the feasible treatment assignment across clusters. For instance, if the realized assignment probability within a cluster is independent of $n_k$ and equal across clusters, then Assumption \ref{pos_Ratio} holds.

Assumptions \ref{unif_bound_weight} and \ref{pos_Ratio} are both required for the identification and estimation of ASE and CSE, as discussed in detail in Section \ref{sec:estimators_ACE}.


Given a pair of units \( ik \) and \( jk \), we define the dyadic average potential outcome of unit $ik$ when unit $jk$'s treatment is fixed and the remaining units in cluster $k$ (including unit $ik$) are assigned treatments according to a hypothetical assignment. 
\begin{definition}[Dyadic average potential outcome]
\label{APO}
Under Assumption \ref{part_intf}, the dyadic average potential outcome of unit $ik$, when unit $jk$'s treatment is fixed at $z_{jk}$ and the remaining units in cluster $k$ (including unit $ik$) are assigned treatments according to a hypothetical assignment parameterized by $\alpha$, is defined as follows:
\[
\bar{Y}_{ik}(Z_{jk}=z_{jk},\alpha) := \mathbb{E}_{\mathbf{Z}_{-jk}|\alpha}\left[ Y_{ik}(Z_{jk}=z_{jk}, \mathbf{Z}_{-jk}) \right].
\]
\end{definition}
We now take the dyadic average potential outcome under treatment assignment $\alpha$ as the basic building block for our estimands and define the pairwise spillover effects as follows.

\begin{definition}[Pairwise spillover effect]
\label{pair_spillover}
Let $i, j \in \{1, \dots, n_k\}$ and $k \in \{1, \dots, K\}$. Under a treatment assignment with parameter $\alpha$, the spillover effect from unit $jk$ to unit $ik$ is defined as
$$\tau_{ik,jk}(\alpha)=\bar{Y}_{ik}(Z_{jk}=1,\alpha)-\bar{Y}_{ik}(Z_{jk}=0,\alpha).$$
\end{definition}
The dyadic average potential outcomes in Definition~\ref{APO} 
are fundamentally different from the average potential outcomes commonly defined under partial interference \citep[e.g.,][]{halloran2018estimating}. Whereas \citet{halloran2018estimating} focus on fixing each unit’s \emph{own} treatment status, we focus on fixing the treatment status of \emph{another} unit (not the unit itself). In \citet{halloran2018estimating}, the spillover effect is defined as the contrast between potential outcomes under two hypothetical treatment allocations. In contrast, under a given hypothetical treatment allocation, we consider the spillover effect generated by changing another unit’s treatment status from $1$ to $0$.

We now define our general estimand for the average spillover effect, which depends on the weights assigned to the pairwise spillover effects.
For clarity of notation, throughout the paper we write $
\sum_{jk \neq ik} := \sum_{\substack{jk=1 \\ jk \neq ik}}^{n_k}$, that is, the summation over all units $jk$ in cluster $k$, excluding unit $ik$.

\begin{definition}[General estimand for average spillover effect (ASE)]
  \label{gen_estimand}
Let $S_{ik,jk} \geq 0$ denote the estimand weight assigned to each pair $(ik,jk)$, 
with $i,j \in \{1,\ldots,n_k\}$, $i \neq j$, and $k \in \{1,\ldots,K\}$. We set $S_{ik,jk} = 0$ when $ik = jk$. Assume further that the weights satisfy $\sum_{k=1}^K \sum_{i=1}^{n_k} \sum_{j \neq i}^{n_k} S_{ik,jk} = S_N$, where $S_N > 0$ is a constant, possibly depending on $N$. The corresponding average spillover effect is then defined as
\begin{equation}
\label{estimand_ASE}
        \begin{split}
            \tau(\alpha)= \sum_{k=1}^K \sum_{ik=1}^{n_k} \sum_{jk\neq ik} S_{ik,jk} \ \tau_{ik,jk}(\alpha)
        \end{split}
    \end{equation}
where $\tau_{ik,jk}(\alpha)$ is given in Definition \ref{pair_spillover}. 
\end{definition}
The estimand weight \(S_{ik,jk}\) can depend on both the network structure and covariates, and it can be specific to each effect receiver \(ik\) and effect sender \(jk\). 

We next provide three examples illustrating how varying $S_{ik,jk}$ leads to different estimands.
For each unit $i \in \{1, \dots, n_k\}$ in cluster $k$, define the set of out-neighbors as $\mathcal{N}^{\mathrm{out}}_{ik} = \{\, jk \in \mathcal{N}_k : e_{ik,jk} \in E_k \,\}$ and the set of in-neighbors as $\mathcal{N}^{\mathrm{in}}_{ik} = \{\, jk \in \mathcal{N}_k : e_{jk,ik} \in E_k \,\}$, with cardinalities $|\mathcal{N}^{\mathrm{out}}_{ik}|$ and $|\mathcal{N}^{\mathrm{in}}_{ik}|$, respectively. We further let $\mathcal{N}_k^{\mathrm{out}} := \{\, jk \in \{1, \dots, n_k\} : |\mathcal{N}^{\mathrm{out}}_{jk}| > 0 \,\}$ and $N^{\mathrm{out}} := \sum_{k=1}^K \sum_{j=1}^{n_k} \mathbf{1}\{\, |\mathcal{N}^{\mathrm{out}}_{jk}| > 0 \}$. Similarly, let $\mathcal{N}^{\mathrm{in}}_k = \{\, ik \in \mathcal{N}_k : |\mathcal{N}^{\mathrm{in}}_{ik}| > 0 \,\}$ and $N^{\mathrm{in}} := \sum_{k=1}^K \sum_{j=1}^{n_k} \mathbf{1}\{\, |\mathcal{N}^{\mathrm{in}}_{jk}| > 0 \}$.

\begin{example}[Average outward spillover effect \citep{fang2025inwardoutwardspillovereffects}]
\label{Examp_Ave_outward}
For $i,j \in \{1,\dots,n_k\}$ and $k \in \{1,\dots,K\}$, let the estimand weight be $S_{ik,jk} = (N^{\text{out}} \lvert \mathcal{N}^{\text{out}}_{jk} \rvert)^{-1} \mathbf{1}\{ik \in \mathcal{N}^{\text{out}}_{jk}\} \mathbf{1}\{jk \in \mathcal{N}^{\text{out}}_k\}$
for $jk$ such that $\lvert \mathcal{N}^{\mathrm{out}}_{jk} \rvert > 0$, and $S_{ik,jk} = 0$ otherwise. Then the average outward spillover effect is defined as
 \begin{equation*}
        \begin{split}
\tau(\alpha) = \frac{1}{N^{\mathrm{out}}} \sum_{k=1}^K \sum_{jk \in \mathcal{N}^{\mathrm{out}}_k} \frac{1}{|\mathcal{N}^{\mathrm{out}}_{jk}|} \sum_{ik \in \mathcal{N}^{\mathrm{out}}_{jk}} \tau_{ik,jk}(\alpha),
\end{split}
\end{equation*}
which measures the average spillover effect of changing a unit’s treatment status on the outcomes of its out-neighbors. When the receiver $ik$ is taken instead from among the in-neighbors of $jk$, i.e., $\mathcal{N}^{\mathrm{in}}_{jk} = \{\, ik \in \mathcal{N}_k : e_{ik,jk} \in E_k \,\}$, and consequently $|\mathcal{N}^{\mathrm{out}}_{jk}|$ is replaced with $|\mathcal{N}^{\mathrm{in}}_{jk}|$, then $\tau(\alpha)$ measures the average spillover effect of changing a unit’s treatment status on the outcomes of its in-neighbors. We still refer to this effect as the average outward spillover effect, as it is defined from the perspective of the sender.
\end{example}

\begin{example}[Average inward spillover effect \citep{fang2025inwardoutwardspillovereffects}]
\label{Examp_Ave_inward}
For $i,j \in \{1,\dots,n_k\}$ and $k \in \{1,\dots,K\}$, define the weight
$S_{ik,jk} \;=\; \frac{1}{N^{in} \, |\mathcal{N}^{in}_{ik}|}\,
\mathbf{1}\{ik \in \mathcal{N}^{in}_k\}\,\mathbf{1}\{jk \in \mathcal{N}^{in}_{ik}\}$ for $ik$ with $|\mathcal{N}^{in}_{ik}|>0$, and set $S_{ik,jk} = 0$ otherwise. The corresponding estimand is
\[
\tau(\alpha) \;=\; \frac{1}{N^{in}} \sum_{k=1}^K \sum_{ik \in \mathcal{N}^{in}_k} 
\frac{1}{|\mathcal{N}^{in}_{ik}|} \sum_{jk \in \mathcal{N}^{in}_{ik}} \tau_{ik,jk}(\alpha),
\]
which represents the average spillover effect on a given unit's outcome of changing the treatment status of one of the unit’s in-neighbors from treated to control. When the sender $jk$ is instead taken from the set of out-neighbors of the effect receiver $ik$, that is, $
\mathcal{N}^{\mathrm{out}}_{ik}$, and consequently $|\mathcal{N}^{\mathrm{in}}_{ik}|$ is replaced with $|\mathcal{N}^{\mathrm{out}}_{ik}|$, the estimand $\tau(\alpha)$ measures the average spillover effect on a given unit's outcome of changing the treatment status of one of the unit’s out-neighbors. We still refer to this quantity as the average inward spillover effect, as it is defined from the perspective of the receiver.
\end{example}

\begin{example}[Average pairwise spillover effect (\citealp{hu2022average} under Assumption \ref{part_intf})]
\label{Examp_ave_ind_effect}
Let $S_{ik,jk} = 1/N$. Then
$$
\tau(\alpha) = \frac{1}{N} \sum_{k=1}^K \sum_{ik=1}^{n_k} \sum_{\substack{jk \neq ik}} \tau_{ik,jk}(\alpha),
$$
which measures the average spillover effect from all other units in the same cluster on a given unit. This definition is similar to that of \citet{hu2022average}, who instead consider general interference.
\end{example}
We next present the general formulation of conditional spillover effects. Here, the restriction on the covariates is encoded through the choice of $S_{ik,jk}$.  

\begin{definition}[General estimand for conditional spillover effect (CSE)]
\label{def:cse_estimand}
Let \( S_{ik,jk}(x) \ge 0 \) denote the estimand weight assigned to each pair \((ik, jk)\), where the effect sender satisfies \(x_{jk} = x\), with \(i, j \in \{1, \ldots, n_k\}\), \(i \neq j\), and \(k \in \{1, \ldots, K\}\). We set \(S_{ik,jk}(x) = 0\) when \(ik = jk\) or when \(x_{jk} \neq x\). Let \(S_{N}(x)=\sum_{k=1}^K \sum_{i=1}^{n_k} \sum_{j \neq i}^{n_k} S_{ik,jk}(x)\), where \(S_{N}(x) > 0\) is a constant that may depend on \(\mathcal{N}(x) := \{\, jk \in \mathcal{N}_k : X_{jk} = x, \ k = 1, \ldots, K \,\}\). The corresponding conditional spillover effect is then defined as
\begin{equation}
\label{estimand_CSE}
    \tau(\alpha, x)
    = \sum_{k=1}^K \sum_{ik=1}^{n_k} \sum_{jk \neq ik} S_{ik,jk}(x) \, \tau_{ik,jk}(\alpha).
\end{equation}
\end{definition}
The conditional spillover effect measures the average spillover originating from effect senders whose covariate value equals \(x\).
In this paper, we focus on the case in which the conditioning covariate is one-dimensional.

We next illustrate how modifying \(S_{ik,jk}(x)\) induces different conditional spillover estimands. Define $
\mathcal{N}^{\mathrm{out}}_k(x) := \{\, jk \in \mathcal{N}_k : |\mathcal{N}^{\mathrm{out}}_{jk}| > 0, \ X_{jk} = x \,\}$ and $
N^{\mathrm{out}}(x) := \sum_{k=1}^K |\mathcal{N}^{\mathrm{out}}_k(x)|$.
\begin{example}[Conditional outward spillover effect \citep{fang2025inwardoutwardspillovereffects}]
\label{Examp_cond_outward}
Let
$S_{ik,jk}(x)=(N^{out}(x)|\mathcal{N}^{out}_{jk}|)^{-1} \mathbf{1}\{ik \in \mathcal{N}^{out}_{jk} \}\cdot \mathbf{1}\{ X_{jk}=x \} $ for $jk$ such that $|\mathcal{N}^{\mathrm{out}}_{jk}| > 0$, and $S_{ik,jk}(x) = 0$ otherwise for $i, j \in \{1,\cdots,n_k\}$ and $k = \{1, \cdots ,K\}$. Then
$$
\tau(\alpha, x) = \frac{1}{N^{\mathrm{out}}(x)} \sum_{jk \in \mathcal{N}^{\mathrm{out}}_k(x)} \frac{1}{|\mathcal{N}^{\mathrm{out}}_{jk}|} \sum_{ik \in \mathcal{N}^{\mathrm{out}}_{jk}} \tau_{ik,jk}(\alpha),
$$ which measures the average  spillover effect from the treatment of  a unit with covariate value $x$ on the outcomes of its out-neighbors.
\end{example}

\begin{example}[Conditional inward spillover effect \citep{fang2025inwardoutwardspillovereffects}]
\label{Examp_cond_inward}
Let $
\mathcal{N}^{in}_{ik}(x) = \{\, jk \in \mathcal{N}^{in}_{ik} : X_{jk} = x \,\}$, $\mathcal{N}^{in}_k(x) = \{\, ik \in \mathcal{N}_k : |\mathcal{N}^{in}_{ik}(x)| > 0 \,\}$ and $
N^{in}(x) = \sum_{k=1}^K |\mathcal{N}^{in}_k(x)|$. For $i,j = \{1,\dots,n_k\}$ and $k = \{1,\dots,K\}$, define the weight $
S_{ik,jk} \;=\; \frac{1}{N^{in}(x)\, |\mathcal{N}^{in}_{ik}(x)|}\,
\mathbf{1}\{ik \in \mathcal{N}^{in}_k(x)\}\,\mathbf{1}\{jk \in \mathcal{N}^{in}_{ik}(x)\}$ for $ik$ with $|\mathcal{N}^{in}_{ik}(x)|>0$, and set $S_{ik,jk} = 0$ otherwise. The corresponding estimand is
\[
\tau(\alpha,x) \;=\; \frac{1}{N^{in}(x)} \sum_{ik \in \mathcal{N}^{in}_k(x)} 
\frac{1}{|\mathcal{N}^{in}_{ik}(x)|} \sum_{jk \in \mathcal{N}^{in}_{ik}(x)} \tau_{ik,jk}(\alpha),
\]
which represents the average spillover effect on a unit's outcome of changing the treatment status of one of the unit’s in-neighbors  with covariate value $x$.
\end{example}

\begin{example}[Conditional pairwise spillover effect]
\label{Examp_cond_ind_effect}
Let $S_{ik,jk} = \frac{1}{|\mathcal{N}(x)|}\, \mathbf{1}\{x_{jk}=x\}$. Then
\[
\tau(\alpha,x) 
= \frac{1}{|\mathcal{N}(x)|} \sum_{k=1}^K \sum_{\substack{jk \in \mathcal{N}_k(x)}} \sum_{ik\neq jk} \tau_{ik,jk}(\alpha),
\]
where $
\mathcal{N}{(x)} := \{\, jk \in \mathcal{N}_k : X_{jk} = x , k \in \{1,\cdots, K\} \,\}$. The estimand represents the average spillover effect from the treatment  of a unit with covariate value $x$ to all other units in the same cluster.  
\end{example}





\section{Estimators for average spillover effect}
\label{sec:est_ase}
In this section, we present three formulations of weighted least squares (WLS) estimators for the average spillover effect, offering distinct yet intuitive perspectives: (i) the dyadic formulation $\hat{\tau}_D(\alpha)$, which treats a dyad $(ik, jk)$ as the unit of analysis; (ii) the effect-receiver formulation $\hat{\tau}_R(\alpha)$, which focuses on the effect receiver; and (iii) the effect-sender formulation $\hat{\tau}_S(\alpha)$, which focuses on the effect sender. Intuitively, the dyadic formulation is suited for a simple average of pairwise spillover effects, such as the average pairwise spillover effect in Example~\ref{Examp_ave_ind_effect}. The receiver formulation is inspired by estimands focused on the effect receiver capturing  spillover effects from aggregated treatments (e.g., neighbors’ treatments) that influence an individual’s outcome, as in the average inward spillover effect in Example~\ref{Examp_Ave_inward}. Conversely, the sender formulation is inspired by estimands focused on the effect sender capturing spillover effects from an individual’s treatment on aggregated outcomes (e.g., neighbors’ outcomes) influenced by that treatment, as in the average outward spillover effect (Example~\ref{Examp_Ave_outward}). However, we can show that all three estimators can be constructed to estimate the same average spillover effect, even though each naturally corresponds to a particular type of estimand sharing the same perspective. 

For each WLS estimator under a specific formulation, it corresponds a specific choice of the design matrix $V_{A}$, the diagonal weight matrix $B_{A}$, and the outcome vector $Y_{A}$ in the WLS expression $
\left(V_{A}^\top B_{A} V_{A}\right)^{-1} \left(V_{A}^\top B_{A} Y_{A}\right)$ where $A \in \{D,R,S\}$. 
We show the equivalence of these formulations to the H\'{a}jek estimator, a nonparametric and consistent estimator of ASE.
In other settings, researchers have already shown the equivalence between weighted least squares and the H\'{a}jek estimator \citep{aronow2017estimating, wang2024designbasedinferencespatialexperiments, gao2025causalinferencenetworkexperiments}.

We then establish the consistency and asymptotic normality of our estimators for the target estimand in Definition \ref{gen_estimand}, with detailed proofs given in Appendix \ref{Appendix_est_ase}.
Since our estimators are identical to the H\'{a}jek estimator, they inherit its consistency. Nonetheless, we provide a consistency proof tailored to the regression framework, following the approach in \citet{abadie2020sampling}, rather than relying solely on arguments specific to the H\'{a}jek estimator, as in the references above. This regression-based proof extends directly to the estimators for the conditional spillover effect. This motivates the inclusion of the proof in Appendix \ref{Appendix_est_ase:equivalent_est}.

\subsection{Three WLS estimators for ASE: dyadic, effect-receiver, and effect-sender formulations}
\label{sec:estimators_ACE}
We now define the estimator weight as $W_{jk}(\mathbf{Z}_k)=\frac{\mathbb{P}_{\alpha}(\mathbf{Z}_{-jk})}{\mathbb{P}_{\beta}(\mathbf{Z}_k)}$, that is, the ratio between the probability of observing $\mathbf{Z}_{-jk}$ under the hypothetical assignment mechanism $\alpha$ and the probability of observing the realized treatment vector $\mathbf{Z}_k$ under $\beta$. The weight $W_{jk}(\mathbf{Z}_k)$ is random, depending on the realization of $\mathbf{Z}_k$. Assumptions \ref{unif_bound_weight} and \ref{pos_Ratio} ensure that these weights are well behaved, enabling consistent estimation of spillover effects under the hypothetical assignment. Together, these conditions imply that $W_{jk}(\mathbf{Z}_k)$ is bounded, i.e., $0 \leq W_{jk}(\mathbf{Z}_k) \leq c^{-1}$ for all $\mathbf{Z}_k$ and $k$. In addition, Assumption~\ref{pos_Ratio} guarantees that the ratio $W_{jk}(\mathbf{Z}_k)$ remains uniformly bounded and does not diverge with $n_k$\footnote{See \citet{fang2025inwardoutwardspillovereffects} for further discussion and implications of these assumptions for estimator performance.}.
 

We first consider the dyadic formulation, which views estimation as regressing a unit’s outcome on another unit’s treatment, based on dyads induced by the assumed interference structure. Our use of dyads and the underlying source of dependence among outcomes differ from the standard dyadic regression literature \citep{aronow2015cluster, tabord2019inference, canen2024inference}. In the latter, dyadic regression is used to study how dyadic characteristics affect dyadic outcomes, and dependence among outcomes arises because two dyadic outcomes share a common unit or are connected through indirect links, rather than through shared treatments. In contrast, we construct dyads such that one unit’s outcome is regressed on another unit’s treatment, with dyads induced by the assumed interference structure. Consequently, in our framework, dependence arises because multiple outcomes depend on the same unit’s treatment. Nevertheless, the resulting forms of the robust variance estimators are similar in both settings.

Based on Assumption \ref{part_intf}, consider all dyads $(ik, jk)$ where $jk$ and $ik$ belong to the same cluster and $jk \neq ik$. We then regress $Y_{ik}$ on $Z_{jk}$ for all such dyads, using weights $B_{ik,jk} := S_{ik,jk} \, W_{jk}(\mathbf{Z}_k)$, where $S_{ik,jk}$ is the weight in the estimand \eqref{estimand_ASE}, and $W_{jk}(\mathbf{Z}_k)$ is the estimator weight.

Let $\mathbf{a}_n$ denote the $n$-vector with all entries equal to $a$. For unit $ik$ in cluster $k$, where $k \in \{1,\ldots,K\}$, let $A_{k,-ik} := (A_{1k},\ldots,A_{(i-1)k},A_{(i+1)k},\ldots,A_{n_k k})$ denote the vector for cluster $k$ excluding unit $ik$. $\mathbf{B}_{ik,-hk} := \big(B_{ik,1k}, \dots, B_{ik,(h-1)k}, B_{ik,(h+1)k}, \dots, B_{ik,n_k k}\big)^\top$ denotes the $(n_k-1)$-vector of weights for the dyads $(ik,jk)$ excluding the pair $(ik,hk)$. The WLS estimator for $\tau(\alpha)$ from the dyadic perspective can then be expressed as follows.
\begin{definition}[Dyadic estimator for ASE]
    \label{Est_ASE_dyad} Let 
\[
\renewcommand{\arraystretch}{0.7}
\begin{array}{ccc}
Y_D=\begin{bmatrix}
Y_{11}\mathbf{1}_{n_1-1} \\
\vdots \\
Y_{n_1 1}\mathbf{1}_{n_1-1}\\
\vdots\\
\vdots\\
Y_{1 K}\mathbf{1}_{n_K-1}\\
\vdots\\
Y_{n_K K}\mathbf{1}_{n_K-1}\\
\end{bmatrix} \quad &
\mathrm{diag}(B_D)=\begin{bmatrix}
\mathbf{B}_{11,-11} \\
\vdots\\
\mathbf{B}_{n_{1}1,-n_11} \\
\vdots\\
\vdots\\
\mathbf{B}_{1K,-1K} \\
\vdots\\
\mathbf{B}_{n_KK,-n_KK} \\
\end{bmatrix} \quad &
V_D=\begin{bmatrix}
\mathbf{1}_{1,-11} & \mathbf{Z}_{1,-11} \\
\vdots & \vdots \\
\mathbf{1}_{1,-n_11} & \mathbf{Z}_{1,-n_11} \\
\vdots & \vdots \\
\vdots & \vdots \\
\mathbf{1}_{K,-1K} & \mathbf{Z}_{K,-1K} \\
\vdots & \vdots \\
\mathbf{1}_{K,-n_KK} & \mathbf{Z}_{K,-n_KK} \\
\end{bmatrix}.
\end{array}
\]
where \( B_D \) is a 
\( \sum_{k=1}^K n_k (n_k - 1) \times \sum_{k=1}^K n_k (n_k - 1) \) 
diagonal matrix.
Then the dyadic estimator of the average spillover effect in Definition \ref{gen_estimand}, with estimand weights $S_{ik,jk}$, is given by the second component of the WLS coefficient vector
$$ \hat{\tau}_D(\alpha) = S_{N} \left[ \left(V_D^\top B_D V_D\right)^{-1} \left(V_D^\top B_D Y_D\right) \right]_{2},$$ 
where $[\cdot]_a$ denotes the $a$-th component of a vector.
\end{definition}
Here, the WLS estimator is rescaled by $S_{N}$, which may depend on the
population size $N$ and is defined in Definition~\ref{gen_estimand}. This term
is incorporated into each formulation of the estimator to ensure that the
scaling of the WLS estimator matches that of the estimand. When
$S_{N}$ is of order $1$, as in the average outward and inward spillover
effects (Examples~\ref{Examp_Ave_outward} and \ref{Examp_Ave_inward}), this
scaling does not affect the rate of convergence or the order of the
asymptotic variance. In contrast, when $S_{N}$ varies with $N$,
as in the case of the average pairwise spillover effect
(Example~\ref{Examp_ave_ind_effect}), it affects both the convergence rate and
the asymptotic variance.

To construct the estimators from the effect-receiver and effect-sender perspectives, we further decompose the estimand weights as follows. For each pair \((ik, jk)\), with \(i, j \in \{1, \ldots, n_k\}\) and \(k \in \{1, \ldots, K\}\), the weight \(S_{ik,jk}\) can be expressed as the product of a marginal and a conditional weight. Specifically, there exist marginal weights \(S_{jk}\) and \(S_{ik}\), and conditional weights \(S_{ik\mid jk}\) and \(S_{jk\mid ik}\), such that $S_{ik,jk} = S_{ik\mid jk} S_{jk} = S_{jk\mid ik} S_{ik}$. For instance, for the average outward spillover effect (Example~\ref{Examp_Ave_outward}), we have $S_{ik\mid jk} = |\mathcal{N}^{\mathrm{out}}_{jk}|^{-1} \mathbf{1}\{\, ik \in \mathcal{N}^{\mathrm{out}}_{jk} \,\}$ and $S_{jk} = (N^{\mathrm{out}})^{-1} \mathbf{1}\{\, jk \in \mathcal{N}^{\mathrm{out}}_k \,\}$. Equivalently, $S_{jk\mid ik} = |\mathcal{N}^{\mathrm{out}}_{jk}|^{-1} 
\mathbf{1}\{\, ik \in \mathcal{N}^{\mathrm{out}}_{jk} \,\}
\mathbf{1}\{\, jk \in \mathcal{N}^{\mathrm{out}}_k \,\}$ and $S_{ik} = (N^{\mathrm{out}})^{-1}$. The conditional weight \(S_{jk\mid ik}\) is used to aggregate the treatments of effect senders that influence the effect receiver \(ik\), whereas \(S_{ik\mid jk}\) is used to aggregate the outcomes of effect receivers that are influenced by a given effect sender \(jk\).

We now consider the second formulation, which adopts the effect-receiver perspective. This formulation is inspired by the common perspective in causal inference that focuses on a unit's outcome to assess how it is affected by others' treatments (see Example \ref{Examp_Ave_inward}). Here, the idea is to regress an effect receiver’s outcome on the aggregated treatments of a subset of senders, where the aggregation is determined by the estimand weights and incorporated directly into the weight matrix. In contrast to specifications that include a parametric summary of others’ treatments as a regressor, this construction avoids imposing functional-form restrictions on the relationship between the outcome and others' treatments, while still yielding coefficient estimates that are consistent for the ASE. Specifically, for each $Y_{ik}$, we construct aggregated weights
$B^{z}_{ik} = \sum_{jk \neq ik} B_{ik,jk} \, \mathbf{1}\{ Z_{jk} = z \}, \ z \in \{0,1\}$,
and then form the contrast between treated senders and control senders. Let $\mathbf{Y} = (Y_{11}, \dots, Y_{n_K K})^\top$. The estimator from the effect-receiver perspective is then defined as follows.
\begin{definition}[Receiver estimator for ASE]
\label{est_ASE_R}
Let $\mathbf{B}^z = (B^z_{11}, B^z_{21}, \dots, B^z_{n_K K})^\top$ for $z \in \{0,1\}$. Define
\[
\renewcommand{\arraystretch}{0.7}
\begin{array}{ccc}
Y_{R}=\begin{bmatrix}
\mathbf{Y} \\
\mathbf{Y} 
\end{bmatrix} \quad &
\mathrm{diag}(B_R)=\begin{bmatrix}
\mathbf{B}^1 \\
\mathbf{B}^0 
\end{bmatrix} \quad &
V_R=\begin{bmatrix}
\mathbf{1}_{N} & \mathbf{1}_{N} \\
\mathbf{1}_{N} & \mathbf{0}_{N}\\
\end{bmatrix}.
\end{array}
\]
where \( B_R \) is a \( 2N \times 2N \) diagonal matrix.
The estimator from the effect-receiver perspective is given by the second component of the weighted least-squares coefficient vector as follows:
$$
\hat{\tau}_R(\alpha) = S_{N} \left[ \left( V_R^\top B_R V_R \right)^{-1} \left( V_R^\top B_R Y_R \right) \right]_2.
$$ 
\end{definition}
Here, $S_N$ plays the same role as in the definition of $\hat{\tau}_{D}(\alpha)$. Although the formulation may appear less intuitive than the dyadic and effect-sender perspectives in Definitions \ref{Est_ASE_dyad} and \ref{est_ASE_S}, the specific construction of $B_R$ and $V_R$ is motivated by two considerations. First, aggregated weights cannot be placed directly in the design matrix, as this would generate interaction terms across different effect receivers in $\hat{\tau}_{R}(\alpha)$. Second, treatment indicators are incorporated into the weights rather than allocating $S_{ik,jk}$ and the treatment indicators separately to the weight matrix $B_R$ and the design matrix $V_R$. The reason is that, under such a separation, establishing the consistency of 
$\hat{\tau}_{R}(\alpha)$ for the ASE, requires imposing restrictive homogeneity 
conditions on the weights $S_{jk\mid ik}$ for $jk \neq ik$. 

We now turn to the third formulation, which adopts the effect-sender perspective. This formulation is inspired by estimands focused on the effect sender, as in Example \ref{Examp_Ave_outward}. The idea is to regress aggregated outcomes for the subset of receivers, defined by the estimand weights, on the treatment of an effect sender; that is, for each $jk$ with treatment $Z_{jk}$, we consider the aggregated outcome $\sum_{ik \neq jk} \frac{S_{ik\mid jk}}{\tilde{S}_{jk}} Y_{ik}$, where $\tilde{S}_{jk} = \sum_{ik \neq jk} S_{ik\mid jk}$, and assign weights $\tilde{S}_{jk} S_{jk}$. The resulting estimator is given below.
\begin{definition}[Sender estimator for ASE]
\label{est_ASE_S}
Let \begin{small}\[
\renewcommand{\arraystretch}{0.7}
\begin{array}{ccc}
Y_{S}=\begin{bmatrix}
 \sum_{ik \neq 11} \frac{S_{ik|11}}{\tilde{S}_{11}} Y_{i1}  \\
\vdots \\
\sum_{ik \neq n_11} \frac{S_{ik|n_11}}{\tilde{S}_{n_11}} Y_{i1}\\
\vdots\\
\vdots\\
\sum_{ik \neq 1K} \frac{S_{ik|1K}}{\tilde{S}_{1K}} Y_{iK}  \\
\vdots \\
\sum_{ik \neq n_K K} \frac{S_{ik|n_{K}K}}{\tilde{S}_{n_{K}K} } Y_{iK}\\
\end{bmatrix}  &
\mathrm{diag}(B_{S})=\begin{bmatrix}
\tilde{S}_{11} S_{11} W_{11}(\mathbf{Z}_1)\\
\vdots\\
\tilde{S}_{n_11} S_{n_11} W_{n_11}(\mathbf{Z}_{n_1})\\
\vdots\\
\vdots\\
\tilde{S}_{1K} S_{1K} W_{1K}(\mathbf{Z}_{K})\\
\vdots\\
\tilde{S}_{n_KK} S_{n_KK} W_{n_KK}(\mathbf{Z}_{K}) \\
\end{bmatrix}  &
V_{S}=\begin{bmatrix}
1 & Z_{11} \\
\vdots & \vdots \\
1 & Z_{n_1 1} \\
\vdots & \vdots \\
\vdots & \vdots \\
1 & Z_{1K} \\
\vdots & \vdots \\
1 & Z_{n_K K} \\
\end{bmatrix}.
\end{array}
\] \end{small}
where \( B_S \) is an \( N \times N \) diagonal matrix.
The estimator from the effect-sender perspective is the second component of the weighted least-squares coefficient vector
$$
\hat{\tau}_{S}(\alpha) = S_{N} \left[ \left( V_{S}^\top B_{S} V_{S} \right)^{-1} \left( V_{S}^\top B_{S} Y_{S} \right) \right]_2.
$$ 
\end{definition}
$S_N$ plays the same role as in the definition of $\hat{\tau}_{D}(\alpha)$. This formulation is more intuitive than the effect-receiver formulation in Definition~\ref{est_ASE_R}, since the treatment appears explicitly in the design matrix $V_{S}$ rather than being absorbed into the weight matrix $B_{S}$.

Although these three formulations appear different and are motivated by different perspectives, they all estimate the same estimand in \eqref{estimand_ASE}.
Furthermore, we can show that they are all equivalent to each other and to the H\'{a}jek estimator $\hat{\tau}_{hj}(\alpha)$ for the ASE, a nonparametric and consistent estimator\footnote{This estimator is similar to the one in Proposition 3 in \citet{wang2024designbasedinferencespatialexperiments}.}, as established in the following theorem.
\begin{theorem}
\label{equiv_hj_D_R_S}
Under Assumption \ref{part_intf},
$$
\hat{\tau}_{D}(\alpha) = \hat{\tau}_{R}(\alpha) = \hat{\tau}_{S}(\alpha) = \hat{\tau}_{hj}(\alpha),
$$
where
\begin{small}
\begin{equation}
\label{est_ASE_hj}
\hat{\tau}_{hj}(\alpha)=S_{N}\left[\frac{\sum_{k=1}^K \sum_{ik=1}^{n_k} \sum_{jk \neq ik} S_{ik,jk} W_{jk}(\mathbf{Z}_k) Z_{jk} Y_{ik}}{\sum_{k=1}^K \sum_{ik=1}^{n_k} \sum_{jk \neq ik} S_{ik,jk} W_{jk}(\mathbf{Z}_k) Z_{jk} } - \frac{\sum_{k=1}^K \sum_{ik=1}^{n_k} \sum_{jk \neq ik} S_{ik,jk}  W_{jk}(\mathbf{Z}_k) (1-Z_{jk}) Y_{ik} }{\sum_{k=1}^K \sum_{ik=1}^{n_k} \sum_{jk \neq ik} S_{ik,jk}  W_{jk}(\mathbf{Z}_k) (1-Z_{jk}) }\right]
\end{equation} 
\end{small}
\end{theorem}
The proof is in Appendix~\ref{Appendix_est_ase:equivalent_est}. Theorem \ref{equiv_hj_D_R_S} further implies that the three estimators share the same variance.
\begin{remark}
The equivalence of $\hat{\tau}_{D}(\alpha)$, $\hat{\tau}_{R}(\alpha)$, and $\hat{\tau}_{S}(\alpha)$ to $\hat{\tau}_{hj}(\alpha)$ can be seen by interpreting \eqref{est_ASE_hj} under different weighting schemes and outcomes. Taking the weights as $S_{ik,jk} W_{jk}(\mathbf{Z}_k)$ and the outcome as $Y_{ik}$ yields the dyadic formulation. Using weights $\sum_{jk \neq ik} S_{ik,jk} W_{jk}(\mathbf{Z}_k)\mathbf{1}\{Z_{jk}=z\}$ for $z \in \{0,1\}$, with the outcome $Y_{ik}$, gives the effect-receiver formulation. Finally, setting the weights as $\sum_{jk=1}^{n_k} S_{jk} \sum_{ik \neq jk} \frac{S_{ik|jk}}{\tilde{S}_{jk}}$ and the outcome as $\sum_{ik \neq jk} \frac{S_{ik|jk}}{\tilde{S}_{jk}} Y_{ik}$ leads to the effect-sender formulation.
\end{remark}

\begin{remark}
For conditional spillover effects with few categories, the equivalence in Theorem \ref{equiv_hj_D_R_S} extends directly: one simply adjusts the weights $S_{ik,jk}$ to $S_{ik,jk}(x)$ and applies the same estimators to units with covariate value $x$.
\end{remark}

\subsection{Inference for estimators of ASE}

To establish the consistency and asymptotic normality of the proposed estimators, we consider an asymptotic regime in which both the number of clusters $K$ and the cluster sizes $n_k$ grow to infinity. When cluster sizes increase with $K$, the within-cluster aggregates that enter the estimators may also scale with $n_k$, so standard arguments based on bounded cluster-level moments no longer apply directly. Instead, we extend the concentration results for dependence graphs in \citet{viviano2024policydesignexperimentsunknown} to establish consistency for our estimators of both the average (Section~\ref{sec:est_ase}) and conditional spillover effects (Section~\ref{sec:est_CSE}), and we adapt a central limit theorem for network data to derive their asymptotic distributions \citep{ogburn2022causal}.

We postulate that the potential outcomes are bounded, which serves as a condition for establishing consistency and asymptotic normality of the estimators, as follows.
\begin{assumption}[Bounded potential outcomes]
    \label{unif_bound_pot_out}
  For each unit $i \in \{1, \ldots, n_k\}$ and $k \in \{1, \ldots, K\}$, there exists a constant $C \geq 0$ such that, for any $\mathbf{z}_{k}\in \{0,1\}^{n_k}$, the potential outcome satisfies $|Y_{ik}(\mathbf{z}_{k})|\leq C$.
\end{assumption}

We now establish the consistency of the ASE estimators in Definitions \ref{Est_ASE_dyad}, \ref{est_ASE_R}, and \ref{est_ASE_S} within a unified framework. 
\begin{proposition}[Consistency of $\hat{\tau}_{A}(\alpha)$]
\label{consist_ASE}
Let $\bar{n}_k = \max_{k} n_k$. Suppose Assumptions \ref{part_intf}--\ref{unif_bound_pot_out} hold. Then, with probability at least $1 - \delta$, 
\begin{equation*}
\left| \hat{\tau}_{A}(\alpha) - \tau(\alpha) \right|
\;\leq\;
\max_{ik,jk} \left(S_{ik,jk}\right) \, S_{N} \, \bar{n}_k^{3/2} \, N^{1/2} 
\log\!\big( 2\bar{n}_k^2 / \delta \big) \quad A\in \{D,R,S\}.
\end{equation*}
Moreover, under the rate condition
\begin{equation}
\label{ase_est_const_rate_cond}
    \begin{split}
        \max_{ik,jk} \left(S_{ik,jk}\right) \, S_{N} \, 
\bar{n}_k^{3/2} \, N^{1/2} 
\log\!\big( 2\bar{n}_k^2N  \big)
\;\overset{N \to \infty}{\longrightarrow}\; 0,
    \end{split}
\end{equation}
then
\begin{equation*}
\left| \hat{\tau}_{A}(\alpha) - \tau(\alpha) \right|
\;\overset{N \to \infty}{\longrightarrow}\; 0,
\qquad A \in \{D, R, S\}.
\end{equation*}
\end{proposition}
Proposition \ref{consist_ASE} provides an upper bound on the rate of convergence of $\hat{\tau}_{A}(\alpha)$ to $\tau(\alpha)$, which is, up to a multiplicative constant, of order $\max_{ik,jk} \left(S_{ik,jk}\right) \, S_{N} \, \bar{n}_k^{3/2} \, N^{1/2} 
\log ( {2\bar{n}_k^2}/{\delta})$. Under the condition that this quantity converges to $0$, the estimator $\hat{\tau}_{A}(\alpha)$ is therefore consistent for $\tau(\alpha)$. It should be noted that this bound is not necessarily tight in settings where the estimand weight $S_{ik,jk}$ is in general not of the same order as $\max_{ik,jk} S_{ik,jk}$ and the lower bound of the cluster size $n_k$ is of the same order as $\bar{n}_k$. Further details are provided in the proof of Proposition \ref{consist_ASE} in Appendix \ref{Appendix:inf_ASE}.

We next introduce an additional assumption on the growth rate of $n_k$ and then present two examples of estimands for which, under a controlled growth rate of $n_k$, the rate condition in Proposition~\ref{consist_ASE} is satisfied, thereby ensuring consistency of the corresponding estimators. 

\begin{assumption}[Controlled growth rate of $n_k$]
\label{ordern_eta_cluster}
For a sequence of clusters indexed by $K$, the cluster sizes satisfy $2 \leq n_k \leq O(K^\eta)$, where $0 \leq \eta \leq \frac{1}{5}$ for $k \in \{1, \dots, K\}$.
\end{assumption}

Assumption \ref{ordern_eta_cluster} allows cluster sizes to either be bounded or grow with \(K\) at a controlled rate. A lower bound of \(n_k \ge 2\) is imposed to ensure that spillover effects are well-defined.

\begin{corollary}
\label{consist_out_ASE}
Consider the average outward spillover effect in Example~\ref{Examp_Ave_outward}, where 
\( S_{ik,jk} = (N^{\mathrm{out}}|\mathcal{N}^{\mathrm{out}}_{jk}|)^{-1} \) for all pairs $ik\neq jk$ and $k\in \{1,\cdots, K\}$, and \( S_{N} = 1 \). Suppose that Assumption~\ref{ordern_eta_cluster} holds and that each cluster contains at least one out-neighbor in this corollary. Then the rate condition~\eqref{ase_est_const_rate_cond} in Proposition~\ref{consist_ASE} is satisfied. Consequently, for any \( A \in \{D,R,S\} \), the estimator \( \hat{\tau}_{A}(\alpha) \) is consistent for the average outward spillover effect.  
\end{corollary}

\begin{corollary}
\label{consist_out_ISE}
Consider the average pairwise spillover effect in Example~\ref{Examp_ave_ind_effect}, where 
\( S_{ik,jk} = N^{-1} \) for all pairs $ik\neq jk$ and $k\in \{1,\cdots, K\}$, and 
\( S_{N} = N^{-1} \sum_{k=1}^K n_k (n_k - 1) \). Under Assumption~\ref{ordern_eta_cluster}, the rate condition~\eqref{ase_est_const_rate_cond} in Proposition~\ref{consist_ASE} holds. Consequently, for \(A \in \{D, R, S\}\), the estimator \( \hat{\tau}_{A}(\alpha) \) is consistent for the average pairwise spillover effect.
\end{corollary}
The proofs of Corollaries \ref{consist_out_ASE} and \ref{consist_out_ISE} are given in Appendix \ref{Appendix:inf_ASE}. When \(n_k = O(\bar{n}_k)\) for all \(k \in \{1, \ldots, K\}\), i.e., when all cluster sizes are of the same order, the estimators can generally converge faster to the estimands. Detailed derivations of these enhanced convergence rates are provided in the proofs of Corollaries \ref{consist_out_ASE} and \ref{consist_out_ISE} in Appendix \ref{Appendix:inf_ASE}. We next establish the asymptotic normality of \(\hat{\tau}_{A}(\alpha)\) for \(A \in \{D, R, S\}\).


\begin{theorem}[Asymptotic Normality of $\hat{\tau}_{A}(\alpha)$]
\label{CLT_ASE_est}
Let $\bar{n}_k$ be defined as in Proposition~\ref{consist_ASE}.  
Suppose Assumptions \ref{part_intf}--\ref{unif_bound_pot_out} hold and that 
\begin{equation}
    \begin{split}
       \label{ase_CLT_rate_cond} \max\left(S_1,S_2\right)\overset{N\rightarrow \infty}{\longrightarrow} 0 
    \end{split}
\end{equation}
where $S_{1}=\frac{\bar{n}_k^{4}}{N}$ and $S_{2}=\max_{ik,jk} S_{ik,jk}^{2} \, (\min_{ik,jk} S_{ik,jk})^{-1} 
\, S_{N}^{-1} K^{1/2} \bar{n}_k^{4} 
\log^{2}(2 \bar{n}_k^{2} N)$. Let $\tilde{N}:=(\sum_{k=1}^K \sum_{i_1k=1}^{n_k} \sum_{j_1k\neq i_1k}\sum_{i_2k =1}^{n_k} \sum_{j_2k\neq i_2k} S_{i_1kj_1k} S_{i_2kj_2k})^{-1}$. Then,
\[
\tilde{N}^{1/2} \big( \tilde{N}\mathrm{var}( \hat{\tau}_{A}(\alpha))\big)^{-1/2}
\big(\hat{\tau}_{A}(\alpha) - \tau(\alpha)\big)
\xrightarrow{d} \mathcal{N}(0,1).
\]
\end{theorem}
The proof is in Appendix~\ref{Appendix:inf_ASE}. In Theorem \ref{CLT_ASE_est}, 
$S_1$ characterizes the within-cluster decay rate of dependence required for the asymptotic normality of our estimators. The second term $S_2$ in~\eqref{ase_CLT_rate_cond} ensures that one component of the decomposition of $\hat{\tau}_{A}(\alpha) - \tau(\alpha)$ is of smaller order than $(\mathrm{var}(\hat{\tau}_{A}(\alpha)))^{1/2}$. The estimator converges to a normal distribution at rate $\tilde{N}^{-1/2}$. Depending on the weights $S_{ik,jk}$, $\tilde{N}$ may exceed $N$, so the convergence rate can be slower than $N^{-1/2}$. Nonetheless, when~\eqref{ase_CLT_rate_cond} holds, the CLT is valid with rate $\tilde{N}^{-1/2}$.

We now turn to the cluster-robust variance estimator for $\operatorname{var}(\hat{\tau}_A(\alpha))$, which is asymptotically conservative under partial interference.

\begin{proposition}[Cluster-Robust Variance Estimator for $\hat{\tau}_A(\alpha)$]
\label{ase_conserve_variance_est}
Let 
\[
\boldsymbol{\hat{\beta}}_{A}(\alpha)
= (V_A^{\top} B_A V_A)^{-1} V_A^{\top} B_A Y_A
\quad \text{and} \quad
\boldsymbol{{\beta}}^{r}_{A}(\alpha)
= [\mathbb{E}(V_A^{\top} B_A V_A)]^{-1} \mathbb{E}(V_A^{\top} B_A Y_A).
\]
Define the residuals $\xi_A = Y_A - V_A \boldsymbol{{\beta}}^{r}_{A}(\alpha)$.  
Then,
\[
\mathrm{var}(\hat{\tau}_{A}(\alpha))
= \Sigma_{A,(2,2)} + o_p(1)
= [\Omega^{-1}_{A} \tilde{\Gamma}_{A} \Omega^{-1}_{A}]_{(2,2)} + o_p(1),
\]
where 
$\Omega_{A} = S_{N}^{-1} \, \mathbb{E}(V_A^{\top} B_A V_A)$ 
and 
$\tilde{\Gamma}_{A} = \mathrm{var}(V_A^{\top} B_A \xi_A)$. Suppose Assumptions~\ref{part_intf}--\ref{unif_bound_pot_out} hold, and that 
rate condition 
\begin{equation}
    \max(S_3, S_4) \overset{N\rightarrow \infty}{\longrightarrow} 0,
\end{equation}
where $$S_{3} := \max_{ik,jk} S^{3}_{ik,jk} \, N^{3/2} \bar{n}_k^{11/2} 
\log(2\bar{n}_k^{2}N)$$ and $$S_{4} := 
\max_{(i_1k,j_1k),(i_2k,j_2k)} 
S_{i_1k,j_1k} S_{i_2k,j_2k}
N^{1/2} \bar{n}_k^{7/2} 
\log^{1/2}(2\bar{n}_k^{4}N).$$
Then an asymptotically conservative estimator of $\Sigma_{A}$ is
$\hat{\Sigma}_{A}
= \hat{\Omega}^{-1}_{A} \,
\hat{\Gamma}_{A} \,
\hat{\Omega}^{-1}_{A}$, that is, with probability $1 - o(1)$,
\[
\mathrm{var}\!\left(\hat{\tau}(\alpha)\right) 
\;\le\; \hat{\Sigma}_{A,(2,2)} + o_p(1),
\]
where $\hat{\Omega}_{A}
= S_{N}^{-1} V_A^{\top} B_A V_A$, $\hat{\Gamma}_{A}
= \sum_{k=1}^K V_{A,k}^{\top} B_{A,k} 
\hat{\xi}_{A,k} \hat{\xi}_{A,k}^{\top} 
B_{A,k} V_{A,k}$ and 
$\hat{\xi}_{A} = Y_A - V_A \boldsymbol{\hat{\beta}}_{A}(\alpha)$.  
Here, $V_{A,k}^{\top}$, $B_{A,k}$, and $\hat{\xi}_{A,k}$ denote the design matrix, diagonal weight matrix, and estimated residuals for cluster $k$, respectively, for $A \in \{D,R,S\}$.  
\end{proposition}

The proof is in Appendix~\ref{Appendix:inf_ASE}. The term $S_3$ in the rate condition of Proposition~\ref{ase_conserve_variance_est} ensures the consistency of 
$\hat{\Gamma}_{A}$ for 
${\Gamma}_{A}$, where 
${\Gamma}_{A}$ is defined analogously but with 
$\hat{\xi}_{A,k}$ replaced by ${\xi}_{A,k}$.  
The term $S_4$ ensures the consistency of 
${\Gamma}_{A}$ for $\mathbb{E}({\Gamma}_{A})$.

\section{Estimators for conditional spillover effect}
\label{sec:est_CSE}
In this section, we develop WLS estimators for conditional spillover effects. These effects measure the spillover effect on the outcome of a subset of the treatment of an effect sender with a given covariate value $x$\footnote{In general, covariate restrictions can also be applied to effect receivers. In this paper, we focus on conditioning on the covariates of effect senders only. Thus, our estimation strategy is tailored to this setting, but can be extended to other conditional spillover effects with covariate restrictions on different types of units, primarily through modifications of the design matrix.}. Such quantities are of particular interest to researchers and policymakers, as they allow the identification of specific types of units with a higher influence on others and may inform the design of targeting strategies.

Section \ref{sec:est_ase} shows that the three average-type estimators can be directly applied to conditional spillover effects when focusing on categorical covariates with a small number of categories. However, when covariates take many categories or are continuous, we develop parametric estimators that exploit information from the full finite population as a balance between estimator efficiency and flexibility of outcome structures. Specifically, we parameterize the dyadic average potential outcome and investigate the conditions on these parameters under which our weighted least squares (WLS) estimators yield valid inference for them.

The three formulations of conditional spillover effects build naturally on their average counterparts. Intuitively, from both the effect-sender and effect-receiver perspectives, the estimators are obtained by modifying the design matrix to include the conditioning covariate of the effect sender, along with its interaction with treatment so as to capture heterogeneity of the spillover effect. From the effect-receiver perspective, the estimators are constructed by modifying the design matrix to include the aggregated covariates of the effect senders, together with a weight matrix that incorporates additional aggregated and weighted treatments and controls of effect senders. 
A key question is whether such WLS estimators remain valid for conditional spillover effects. We show that additional conditions are required. To clarify the role and strength of these conditions, we introduce two intermediate quantities that link the WLS estimators to the target estimand. The first quantity can be consistently estimated under weak regularity conditions, while the second corresponds to a population-level average. 

The remainder of this section proceeds as follows. We first define the conditional spillover effect estimand and parametrize it under a flexible structural model for dyadic average potential outcomes. Next, we introduce WLS estimators, obtained as modifications of the average-type estimators in Section~\ref{sec:est_ase}, and establish their consistency for intermediate quantities under mild regularity conditions. We then characterize the additional assumptions required to link these intermediate quantities to the conditional spillover effect estimand. Subsequently, we establish a central limit theorem (CLT) for the estimators of the CSE. As in the case of the ASE, we further show that the corresponding cluster-robust variance estimator is asymptotically conservative under Assumption~\ref{part_intf} when both the number of clusters and their sizes diverge.

\subsection{CSE under structural models of dyadic average potential outcomes}
In this section, we employ the rescaled weight 
${S}^r_{ik,jk} := S_{ik,jk}/S_{N}$,
where $S_{N}$, which may itself depend on $N$, is given in Definition \ref{gen_estimand}.
By construction, these weights satisfy $\sum_{k=1}^K \sum_{ik=1}^{n_k} \sum_{jk \neq ik} {S}^r_{ik,jk} = 1$.
This normalization facilitates the derivation of explicit expressions for intermediate quantities, as detailed in the proof of Proposition~\ref{beta_r_to_beta_p}, which will later be connected to the target estimand, the CSE. In addition, we assume that the covariate space is bounded.

\begin{assumption}[Bounded covariate space]
\label{bounded_x}
All covariate values $x \in \mathcal{X}$ of interest lie in a compact interval $[a,b]$, for some constants $a<b$. 
\end{assumption}
Boundedness of the covariate space is a standard assumption in the literature on the estimation of conditional treatment effects (see, e.g., \citealp{wager2018estimation}; \citealp{cui2023estimating}). Assumption \ref{bounded_x} will be used to establish the consistency of the estimators and of the intermediate quantities introduced in Section \ref{Cond_est_consistency}. 

We next specify a general structural model for the dyadic average potential outcome, together with its demeaned representation, to express the target estimand as a function of parameters and to show how it is connected to the components of the estimators.

\begin{definition}[Structural model for dyadic average potential outcomes]
\label{struc_APO}
The dyadic average potential outcome defined in Definition \ref{APO} is parameterized as
\begin{equation*}
    \begin{split}
       \bar{Y}_{ik}(Z_{jk}=z_{jk},\alpha)
       &= \theta_{1,ijk}(\alpha) + \theta_{2,ijk}(\alpha)X_{jk} 
       + \theta_{3,ijk}(\alpha) z_{jk} 
       + \theta_{4,ijk}(\alpha) z_{jk} X_{jk} + \epsilon_{ik} \\
       &:= {\beta}_{1,ijk}(\alpha) 
       + {\beta}_{2,ijk}(\alpha)\tilde{X}_{jk}  
       + {\beta}_{3,ijk}(\alpha) z_{jk} 
       + {\beta}_{4,ijk}(\alpha) z_{jk}\tilde{X}_{jk} + \epsilon_{ik},
    \end{split}
\end{equation*}
where $\bar{X} = \sum_{k=1}^K \sum_{ik=1}^{n_k} \sum_{jk \neq ik} {S}^r_{ik,jk} X_{jk}$,  
$\tilde{X}_{jk} = X_{jk}-\bar{X}$,  
${\beta}_{1,ijk}(\alpha)=\theta_{1,ijk}(\alpha)+ \theta_{2,ijk}(\alpha)\bar{X}$,  
${\beta}_{2,ijk}(\alpha)=\theta_{2,ijk}(\alpha)$,  
${\beta}_{3,ijk}(\alpha)=\theta_{3,ijk}(\alpha)+\theta_{4,ijk}(\alpha)\bar{X}$,  
and ${\beta}_{4,ijk}(\alpha)=\theta_{4,ijk}(\alpha)$.
\end{definition}
The term $\epsilon_{ik}$ denotes a fixed individual shock for unit $ik$. In our setting, the treatment assignment is randomized, implying the independence between the treatment vector $\mathbf{Z}$ and the vector of error terms $\boldsymbol{\epsilon}$. 

Definition~\ref{struc_APO} should only be conceived as a useful parametrization  for dyadic average potential outcomes to isolate the covariate of substantive interest. Importantly, it does not impose any parametric assumptions.  In Definition~\ref{struc_APO}, the coefficients are allowed to vary across $(ik,jk)$ pairs and can also incorporate additional information, such as covariates that are not being conditioned on.
Second, the heterogeneity of $\theta_{2,ijk}(\alpha)$ and $\theta_{4,ijk}(\alpha)$ implies that the dyadic average potential outcome can, in principle, be nonlinear in $X_{jk}$, with such nonlinearities being absorbed into these coefficients. Third, the structural model is written in terms of dyadic average potential outcomes rather than individual potential outcomes, thereby allowing greater flexibility in outcome structures, including interactions between treatments. In fact, any model for potential outcomes can be written as the general parameterization of the dyadic average potential outcomes in Definition~\ref{struc_APO}.
We adopt this parametrization because the WLS estimators are constructed using the corresponding design matrix. Within this framework, we can then characterize the conditions under which our WLS estimators coincide with, or can be linked to, the CSE.


The re-parameterization with $\tilde{X}$ is introduced to align with the intermediate quantities defined in Section \ref{Cond_est_consistency}, where its explicit expression can be derived in part because of demeaning, which implies that \(\sum_{k=1}^K \sum_{ik=1}^{n_k} \sum_{jk \neq ik} S^r_{ik,jk} \bar{X} = \bar{X}\).
 

\begin{definition}[Conditional spillover effect (CSE) under the structural model]
\label{cond_spill_over_under_struc_model}
Under Definition~\ref{struc_APO}, the conditional spillover effect defined in \eqref{estimand_CSE} can be written as
\[
\tau(\alpha,x)
= \sum_{k=1}^K \sum_{ik=1}^{n_k} \sum_{jk \neq ik} S^r_{ik,jk}(x)
\bigl( \beta_{3,ijk}(\alpha) + \beta_{4,ijk}(\alpha)\tilde{x} \bigr)
:= \bar{\beta}_3(\alpha,x) + \bar{\beta}_4(\alpha,x)\tilde{x},
\]
where $\tilde{x}=x-\bar{X}$, with $\bar{X}$ defined in Definition~\ref{struc_APO}, and
\[
S^{r}_{ik,jk}(x)=S_{N}^{-1}(x)S_{ik,jk}1\{X_{jk}=x\}, \quad 
S_{N}(x)=\sum_{k=1}^K \sum_{ik=1}^{n_k} \sum_{jk\neq ik} S_{ik,jk}1\{X_{jk}=x\}.
\]
\end{definition}
The quantities $\bar{\beta}_3(\alpha, x)$ and $\bar{\beta}_4(\alpha, x)$ are weighted averages of $\beta_{3,ijk}(\alpha)$ and $\beta_{4,ijk}(\alpha)$, respectively, restricted to units with covariate value $x$. 
By the definition of $S^r_{ik,jk}(x)$ in Definition~\ref{cond_spill_over_under_struc_model}, it follows that $\sum_{k=1}^K \sum_{ik=1}^{n_k} \sum_{jk \neq ik} S^r_{ik,jk}(x)=1$. 

    


\subsection{Three WLS estimators for CSE: dyadic, effect-receiver, and effect-sender formulations}
\label{Cond_est_consistency}

In this section, we introduce the dyadic, receiver, and sender estimators for the conditional spillover effect evaluated at covariate value \( x \), \(\hat{\tau}_{A}(\alpha, x)\) for \(A \in \{D, R, S\}\). Each estimator is composed of two elements:  
(i) an estimator for the average spillover component \(\bar{\beta}_3(\alpha, x)\), and  
(ii) an estimator for the heterogeneous component \(\bar{\beta}_4(\alpha, x)\tilde{x}\). \(\bar{\beta}_3(\alpha, x)\) and \(\bar{\beta}_4(\alpha, x)\tilde{x}\) are defined in Definition \ref{cond_spill_over_under_struc_model}.
The estimator for \(\bar{\beta}_3(\alpha, x)\) is analogous to the ASE estimators. 
The estimator for the heterogeneous component \(\bar{\beta}_4(\alpha, x)\) is incorporated through the design matrix for the dyadic and sender estimators, and through the weight matrix for the receiver estimator.

Intuitively, the structure of each estimator is motivated by a specific estimand perspective; however, any of the three estimators can incorporate the estimand weights so as to target any estimand. We now describe the intuition for each estimator relative to a particular estimand. The dyadic estimator is motivated by conditional averages of pairwise spillover effects, such as the conditional pairwise spillover effect in Example~\ref{Examp_cond_ind_effect} under Definition~\ref{struc_APO}.
The dyadic estimator regresses the effect of the receiver’s outcome on the effect of the sender’s treatment to recover $\bar{\beta}_{3}(\alpha, x)$, and on the interaction between the sender’s treatment and the covariate to recover $\bar{\beta}_{4}(\alpha, x)$. A linear combination of these two components, evaluated at covariate value $x$, yields the estimator of the CSE at $x$. The receiver estimator is inspired by the effect-receiver perspective, which is used to define estimands such as the conditional inward spillover effect in Example~\ref{Examp_cond_inward}. The idea is to regress one’s outcome on aggregated treatments and their interactions with covariates of effect senders, but without imposing a functional form on these aggregations, in contrast to what is commonly done in practice.
Conversely, the sender estimator is motivated by an effect-sender perspective, which is used to define estimands such as the conditional outward spillover effect in Example~\ref{Examp_cond_outward}. In this case, the estimator is constructed by regressing aggregate outcomes of effect receivers on a unit's treatment and its interaction with the covariate.

Similar to the ASE estimators, the three CSE estimators can target the same conditional spillover effect using the estimand weights. However, unlike the equivalence among $\hat{\tau}_{D}(\alpha)$, $\hat{\tau}_{S}(\alpha)$, and $\hat{\tau}_{R}(\alpha)$ for the ASE (Theorem~\ref{equiv_hj_D_R_S}), we show equivalence between $\hat{\tau}_{D}(\alpha, x)$ and $\hat{\tau}_{S}(\alpha, x)$, but not with $\hat{\tau}_{R}(\alpha, x)$ (Proposition~\ref{relation_est_CSE}). Moreover, the conditions under which each estimator consistently estimates the conditional spillover effect estimand (Definition~\ref{cond_spill_over_under_struc_model}) differ and are detailed in Section~\ref{connect_estimable_quant_CSE}.

We introduce the WLS estimators aligned with the structure specified in Definition~\ref{struc_APO}. We show how to incorporate the treatment vector $\mathbf{Z}$ and the interaction term $\mathbf{Z} \circ \mathbf{X}$ into the design matrices, beyond the terms related to $\mathbf{1}$ and $\mathbf{Z}$ introduced in Definitions~\ref{Est_ASE_dyad}, \ref{est_ASE_R}, and \ref{est_ASE_S}, together with the corresponding outcome vectors and diagonal weight matrices, where $a \circ b$ denotes the element-wise product of vectors $a$ and $b$.  

From Definition \ref{cond_spill_over_under_struc_model}, the causal parameters of interest correspond to $\beta_{3\cdot}(\alpha)$ and $\beta_{4\cdot}(\alpha)$, while $\beta_{1\cdot}(\alpha)$ and $\beta_{2\cdot}(\alpha)$ are nuisance parameters. Furthermore, by the proof of Proposition \ref{consist_ASE}, we know that $(V_A^\top B_A V_A)^{-1}(V_A^\top B_A Y_A)$ consistently estimates $(\mathbb{E}[V_A^\top B_A V_A])^{-1}(\mathbb{E}[V_A^\top B_A Y_A])$ under weak conditions. Let us define this ratio as $\boldsymbol{\beta}^r(\alpha)$. Accordingly, we also study the consistency of the CSE estimators relative to the quantity $\boldsymbol{\beta}^r(\alpha)$. For $A\in \{D,S\}$, since $(\mathbb{E}[V_A^\top B_A V_A])^{-1}$ is a $4 \times 4$ inverse matrix, obtaining explicit expressions for individual components of $\boldsymbol{\beta}^r(\alpha)$ is challenging. To address this, we orthogonalize the covariates for the causal parameters and those for the non-causal parameters, which allows us to write each component of $\boldsymbol{\beta}^r(\alpha)$ explicitly. 

Specifically, we project the regressors $(Z_{jk}, Z_{jk}\tilde{X}_{jk})$ associated with the causal coefficients onto $(1, \tilde{X}_{jk})$, the regressors for the nuisance coefficients, using the weight $B_{ik,jk} = S^r_{ik,jk} W_{jk}(\mathbf{Z}_k)$. The projection operator is then given by $\Lambda = \tfrac{1}{2} I_2$, with the derivation provided in Definition \ref{Lambda_formula} in Appendix \ref{Appendix:connect_CSE}, following \citet{abadie2020sampling}. The resulting orthogonalized regressors for each unit are

\begin{equation}
\label{transformed_causal_component}
\begin{pmatrix}
Z^{*}_{jk} \\
(Z_{jk}\tilde{X}_{jk})^*
\end{pmatrix}
:=
\begin{pmatrix}
Z_{jk} \\
Z_{jk}\tilde{X}_{jk}
\end{pmatrix}
-
\Lambda
\begin{pmatrix}
1 \\
\tilde{X}_{jk}
\end{pmatrix}
=
\begin{pmatrix}
Z_{jk} - \tfrac{1}{2} \\
Z_{jk}\tilde{X}_{jk} - \tfrac{1}{2}\tilde{X}_{jk}
\end{pmatrix}.
\end{equation}
We then define the corresponding transformed vectors $\mathbf{Z}^* = (Z^{*}_{11}, Z^{*}_{21}, \ldots, Z^{*}_{n_K K})^\top$ and $(\mathbf{Z} \circ \tilde{\mathbf{X}})^* = \big((Z_{11}\tilde{X}_{11})^*, (Z_{21}\tilde{X}_{21})^*, \ldots, (Z_{n_K K}\tilde{X}_{n_K K})^*\big)^\top$.  

Finally, for $S^r_{ik,jk}$, let us decompose the rescaled weights as $S^r_{ik,jk} = S^r_{ik|jk} S^r_{jk} = S^r_{jk|ik} S^r_{ik}$ for all $i, j \in \{1, \ldots, n_k\}$ and $k \in \{1, \ldots, K\}$.

\begin{definition}[Dyadic estimator for CSE]
\label{Est_CSE_dyad}
Let \( Y_D \) and \( B_D \) be defined as in Definition~\ref{Est_ASE_dyad}, except that for each component of \( B_D \), \( S_{ik,jk} \) is replaced by its corresponding rescaled weight \( S^r_{ik,jk} \). Specifically, for \( i, j \in \{1, \ldots, n_k\} \) and \( k \in \{1, \ldots, K\} \), define
\( B_{ik,jk} := S^r_{ik,jk}\, W_{jk}(\mathbf{Z}_k) \). Let the design matrix
be
\[ \renewcommand{\arraystretch}{0.7}
V_D = \begin{bmatrix}
\mathbf{1}_{1,-11} & \tilde{\mathbf{X}}_{1,-11} & \mathbf{Z}^*_{1,-11} & (\mathbf{Z}\circ \tilde{\mathbf{X}})^*_{1,-11} \\
\vdots & \vdots & \vdots & \vdots \\
\mathbf{1}_{1,-n_1 1} & \tilde{\mathbf{X}}_{1,-n_1 1} & \mathbf{Z}^*_{1,-n_1 1} & (\mathbf{Z}\circ \tilde{\mathbf{X}})^*_{1,-n_1 1} \\
\vdots & \vdots & \vdots & \vdots \\
\mathbf{1}_{K,-1K} & \tilde{\mathbf{X}}_{K,-1K} & \mathbf{Z}^*_{K,-1K} & (\mathbf{Z}\circ \tilde{\mathbf{X}})^*_{K,-1K} \\
\vdots & \vdots & \vdots & \vdots \\
\mathbf{1}_{K,-n_K K} \quad & \tilde{\mathbf{X}}_{K,-n_K K} \quad & \mathbf{Z}^*_{K,-n_K K} \quad & (\mathbf{Z}\circ \tilde{\mathbf{X}})^*_{K,-n_K K}
\end{bmatrix}.
\]
Then the WLS estimator is
\[
\hat{\boldsymbol{\beta}}_D(\alpha) := \left[ (V_D^\top B_D V_D)^{-1} (V_D^\top B_D Y_D) \right]_{4 \times 1},
\]
and the estimator for conditional spillover effect evaluated at covariate value $x$ is
\[
\hat{\tau}_D(\alpha,x) = S_{N}[\hat{\beta}_{D,3}(\alpha) + \hat{\beta}_{D,4}(\alpha)(x-\bar{X})].
\] 
\end{definition}
In this specification, the coefficient $\hat{\beta}_{D,3}(\alpha)$ is obtained by regressing the effect receiver’s outcome on the effect sender’s treatment, whereas $\hat{\beta}_{D,4}(\alpha)$ is obtained by regressing the effect receiver’s outcome on the interaction term between the effect sender’s treatment and the demeaned covariate. Note that in Definition \ref{Est_CSE_dyad}, the design matrix is constructed under a linear relationship between a unit's outcome and the sender's covariates. The assumption needed for consistency (Assumption \ref{assump:indep_tilde_beta_bar_beta_x}) relies on this construction of the design matrix\footnote{Alternative specifications of the design matrix are possible, incorporating more flexible (e.g., semiparametric) transformations of $X_{jk}$. In this case, the conditions required to link the dyadic estimator to the CSE (Section~\ref{connect_estimable_quant_CSE}) can be weakened and must be adjusted depending on how the functional form of $X_{jk}$ is specified in the design matrix $V_{D}$.}. The receiver and sender estimators follow analogous logic.

\begin{definition}[Estimator of CSE from the effect-receiver’s perspective]
\label{Est_CSE_R}
Let $\mathbf{Y}$ and $\mathbf{B}^z$, for $z \in \{0,1\}$, be as defined in Definition \ref{est_ASE_R}, except that each component of $\mathbf{B}^z$ is replaced by its rescaled counterpart. In particular, for each $i \in \{1,\ldots,n_k\}$ and $k \in \{1,\ldots,K\}$, define $B^{z}_{ik} = \sum_{jk \neq ik} S^r_{ik,jk}\, W_{jk}(\mathbf{Z}_k) \, \mathbf{1}\{ Z_{jk} = z \}, z \in \{0,1\}$ which aggregates, for each effect receiver $ik$, the contributions from all
effect senders whose treatment status is $z$. Similarly, define the aggregated covariate for each effect receiver $ik$ as
\[
X^{\dagger}_{ik} := \sum_{jk \neq ik} {S}^r_{jk|ik} X_{jk}
 - \sum_{k=1}^K \sum_{ik=1}^{n_k} {S}^r_{ik} \sum_{jk \neq ik} {S}^r_{jk|ik} X_{jk}
 := X^{\dagger o}_{ik} - \sum_{k=1}^K \sum_{ik=1}^{n_k} {S}^r_{ik} X^{\dagger o}_{ik}.
\]
Let $\mathbf{X}^{\dagger} := (X^{\dagger}_{11},\ldots,X^{\dagger}_{n_K K})^\top$.  
The augmented outcome vector, weight matrix, and design matrix are
\[
Y_{R} =
\begin{bmatrix}
\mathbf{Y} \\ \mathbf{Y} \\ \mathbf{Y} \\ \mathbf{Y}
\end{bmatrix}, \quad
\mathrm{diag}(B_{R}) =
\begin{bmatrix}
\mathbf{B}^1 \\ \mathbf{B}^1 \\ \mathbf{B}^0 \\ \mathbf{B}^0
\end{bmatrix}, \quad
V_{R} =
\begin{bmatrix}
\mathbf{1}_{N} & \mathbf{0}_{N} & \mathbf{1}_{N} & \mathbf{0}_{N} \\ 
\mathbf{1}_{N} & \mathbf{0}_{N} & \mathbf{0}_{N} & \mathbf{0}_{N} \\ 
\mathbf{0}_{N} & \mathbf{X}^{\dagger}_{N} & \mathbf{0}_{N} & \mathbf{X}^{\dagger}_{N} \\ 
\mathbf{0}_{N} & \mathbf{X}^{\dagger}_{N} & \mathbf{0}_{N} & \mathbf{0}_{N}
\end{bmatrix}.
\]
Then the WLS estimator is
\[
\hat{\boldsymbol{\beta}}_{R}(\alpha) := \left[ (V_{R}^\top B_{R} V_{R})^{-1} (V_{R}^\top B_{R} Y_{R}) \right]_{4 \times 1},
\]
and the estimator for conditional spillover effect is
\[
\hat{\tau}_{R}(\alpha,x) = S_{N}[\hat{\beta}_{R,3}(\alpha)+ \hat{\beta}_{R,4}(\alpha)(x-\bar{X})].
\]
\end{definition}
Here, the coefficient $\hat{\beta}_{R,3}(\alpha)$ is constructed by forming the
contrast between aggregated treatments and aggregated controls for each effect
receiver. This contrast is encoded through the first and third columns of $V_R$,
combined with the corresponding weights $\mathbf{B}^1$ and $\mathbf{B}^0$ in
$B_R$. Likewise, the coefficient $\hat{\beta}_{R,4}(\alpha)$ is obtained by
contrasting treated and control receivers with respect to their aggregated
covariates in $\mathbf{X}^\dagger_{N}$, using the second and fourth columns of
$V_R$, again together with the weights $\mathbf{B}^1$ and $\mathbf{B}^0$ in
$B_R$.

\begin{definition}[Estimator of CSE from the effect sender’s perspective] 
\label{Est_CSE_S}
Let $\mathbf{Y}_S$ and $\mathbf{B}_S$ be as defined in Definition \ref{est_ASE_S}, except that each component is replaced by its rescaled counterpart. 
Specifically, each aggregated outcome is replaced by 
$\sum_{ik \neq jk} ({S^r_{ik|jk}}/{\tilde{S}^r_{jk}})\, Y_{ik}$, where $\tilde{S}^r_{jk} := \sum_{ik \neq jk} S^r_{ik|jk}$, and each weight is replaced by $\tilde{S}^r_{jk} S^r_{jk}$. The design matrix is
\[\renewcommand{\arraystretch}{0.7}
V_S = \begin{bmatrix}
1 & \tilde{X}_{11} & Z^*_{11} & (Z_{11}\tilde{X}_{11})^* \\
\vdots & \vdots & \vdots & \vdots \\
1 & \tilde{X}_{n_1 1} & Z^*_{n_1 1} & (Z_{n_1 1}\tilde{X}_{n_1 1})^* \\
\vdots & \vdots & \vdots & \vdots \\
1 & \tilde{X}_{1K} & Z^*_{1K} & (Z_{1K}\tilde{X}_{1K})^* \\
\vdots & \vdots & \vdots & \vdots \\
1 & \quad \tilde{X}_{n_K K} \quad & Z^*_{n_K K} \quad & (Z_{n_K K}\tilde{X}_{n_K K})^*
\end{bmatrix}.
\]
Then the WLS estimator is
\[
\boldsymbol{\hat{\beta}}_{S}(\alpha) := \left[ (V_S^\top B_S V_S)^{-1} (V_S^\top B_S Y_S) \right]_{4 \times 1},
\]
and the estimator for conditional spillover effect at covariate value $x$ is
\[
\hat{\tau}_{S}(\alpha,x) =S_{N} [\hat{\beta}_{S,3}(\alpha) + \hat{\beta}_{S,4}(\alpha)(x-\bar{X})].
\]
\end{definition}
\noindent Here, the coefficient $\hat{\beta}_{S,3}(\alpha)$ is obtained by regressing the aggregated outcomes of the effect receivers on the effect sender's treatment, whereas $\hat{\beta}_{S,4}(\alpha)$ is obtained by regressing the aggregated outcomes of the effect receivers on the interaction between the effect sender's treatment and its demeaned covariate.

The components of the WLS estimators from different perspectives for the CSE are related as follows.  

\begin{proposition}
\label{relation_est_CSE}
Under the formulations in Definitions \ref{Est_CSE_dyad}, \ref{Est_CSE_R}, and \ref{Est_CSE_S}, we have
\[
\hat{\beta}_{D,3}(\alpha)=\hat{\beta}_{S,3}(\alpha), 
\quad 
\hat{\beta}_{D,4}(\alpha)=\hat{\beta}_{S,4}(\alpha), 
\quad 
\hat{\beta}_{R,3}(\alpha)=\hat{\tau}_{hj}(\alpha),
\]
where $\hat{\tau}_{hj}(\alpha)$ is defined in Theorem \ref{equiv_hj_D_R_S}.
\end{proposition}
\noindent The proof is in Appendix~\ref{Appendix:connect_CSE}. Proposition \ref{relation_est_CSE} implies that $\hat{\tau}_D(\alpha,x) = \hat{\tau}_S(\alpha,x)$, so the equivalence between the dyadic and sender-perspective estimators continues to hold for the CSE. It is important to note, however, that $\hat{\beta}_{D,3}(\alpha)$ and $\hat{\beta}_{S,3}(\alpha)$ do not coincide with $\hat{\tau}_{hj}(\alpha)$ by construction, whereas $\hat{\beta}_{R,3}(\alpha)$ does. That is, the third component of the estimator for the CSE, $\boldsymbol{\hat{\beta}}_{R}(\alpha)$, from the effect-receiver perspective, coincides with the H\'{a}jek estimator for the ASE and thus also with the WLS estimators for ASE, $\hat{\tau}_{A}(\alpha)$ for $A\in \{R,D,S\}$ in Section \ref{sec:estimators_ACE}. Although $\hat{\beta}_{D,3}(\alpha)$ and $\hat{\beta}_{S,3}(\alpha)$ differ from $\hat{\tau}_{hj}(\alpha)$, both remain consistent estimators of the ASE, as established in the next proposition. 

\begin{proposition}[Consistency of estimators for $\boldsymbol{\beta}^r_{A}(\alpha)$]
\label{Const_ratio_estimable_quantity}
Let
\[
\boldsymbol{\beta}^r_A(\alpha) := \big[\mathbb{E}(V_A^\top B_A V_A)\big]^{-1}\mathbb{E}(V_A^\top B_A Y_A), \quad A \in \{D,S,R\}.
\]
If the same assumptions as in Proposition~\ref{consist_ASE}, together with Assumption \ref{bounded_x}, hold, and under the rate condition in~\eqref{ase_est_const_rate_cond}, then
\[
S_{N}|\hat{\beta}_{A,3}(\alpha)-\beta^r_{A,3}(\alpha)| \;\overset{N\to\infty}{\longrightarrow}\; 0, 
\quad A \in \{D,S,R\},
\]
where $S_{N}\beta^r_{A,3}(\alpha)=\tau(\alpha)=\sum_{k=1}^K\sum_{ik=1}^{n_k} \sum_{jk\neq ik} S_{ik,jk}\bigl(\beta_{3,ijk}(\alpha)+\beta_{4,ijk}(\alpha)\tilde{X}_{jk}\bigr)$, as in \eqref{estimand_ASE} under Definition~\ref{struc_APO}. Furthermore,
\[
S_{N}|\hat{\beta}_{A,4}(\alpha)-\beta^r_{A,4}(\alpha)| \;\overset{N\to\infty}{\longrightarrow}\; 0, 
\quad A \in \{D,S,R\},
\]
where
\[
\beta^r_{A,4} (\alpha)
= \Bigg(\sum_{k=1}^K \sum_{ik=1}^{n_k} \sum_{jk\neq ik} S_{ik,jk} a^2_{ik,jk}\Bigg)^{-1}
   \sum_{k=1}^K \sum_{ik=1}^{n_k} \sum_{jk\neq ik} S_{ik,jk} a_{ik,jk}\tau_{ik,jk}(\alpha),
\]
with $a_{ik,jk}=\tilde{X}_{jk}$ for $A \in \{D,S\}$ and $a_{ik,jk}=X^\dagger_{ik}$ for $A \in \{R\}$, and $\tau_{ik,jk}(\alpha)=\beta_{3,ijk}(\alpha)+\beta_{4,ijk}(\alpha)\tilde{X}_{jk}$ as in Definition~\ref{struc_APO}.
\end{proposition}

\noindent The proof is in Appendix~\ref{Appendix:inf_CSE}. Proposition \ref{Const_ratio_estimable_quantity} shows that $\hat{\beta}_{A,3}(\alpha)$ is consistent with $\beta^r_{A,3} (\alpha)=S^{-1}_{N}\tau(\alpha)$ for any $A \in \{D,S,R\}$, whereas $\hat{\beta}_{A,4}(\alpha)$ is consistent with $\beta^r_{A,4}(\alpha)$, which is a weighted average of the pairwise spillover effects $\tau_{ik,jk}(\alpha)$, with weights depending on $X_{jk}$ specific to each $A \in \{D,S,R\}$. This consistency result holds under only mild regularity conditions and the rate condition in \eqref{ase_est_const_rate_cond}, without imposing any additional restrictions on the coefficients or on the relationship between the coefficients and the conditioning covariates in Definition~\ref{struc_APO}. For $h \in \{3,4\}$, the convergence rate of $\hat{\beta}_{A,h}(\alpha)$ to $\beta^{r}_{A,h}(\alpha)$ is the same as in Proposition \ref{consist_ASE}, by the same proof.


\subsection{Connecting intermediate quantities to CSE}
\label{connect_estimable_quant_CSE}
The next step is to examine how $\beta^r_{A,3}(\alpha)$ and $\beta^r_{A,4}(\alpha)$ relate to the quantities $\bar{\beta}_3(\alpha,x)$ and $\bar{\beta}_4(\alpha,x)$ involved in the target estimand CSE, which are weighted averages of $\beta_{3,ijk}(\alpha)$ and $\beta_{4,ijk}(\alpha)$ restricted to units with $X_{jk}=x$. We then present conditions under which $\hat{\beta}_{A,3}(\alpha)$ and $\hat{\beta}_{A,4}(\alpha)$ are consistent estimators of these quantities.
These conditions concern the 
heterogeneity of the coefficients \( \theta_{h,ijk}(\alpha) \) for \( h \in \{3,4\} \) in Definition~\ref{struc_APO}, as well as the relationship between the average of the covariate \( X_{jk} \) within groups with the same value of \( \theta_{h,ijk}(\alpha) \) and the overall average $\bar{X}.$

 
\begin{assumption}[Restriction on heterogeneity of ${\theta}_{h,ijk}(\alpha)$]
\label{coef_hete_restrict}
Consider the following restrictions on the heterogeneity of the coefficients ${\theta}_{h,ijk}(\alpha)$ for $h \in \{3,4\}$ in the structural model of Definition~\ref{struc_APO}:
\begin{itemize}
    \item[1.] There exist $m_h$ finite and distinct values such that
    $\theta_{h,ijk}(\alpha) \in \{\theta_{h,1}(\alpha), \ldots, \theta_{h,m_h}(\alpha)\}$.  
    For each $a \in \{1,\ldots,m_h\}$, 
    \begin{equation*}
       S_{N} \left\lvert 
        \sum_{k=1}^K \sum_{ik=1}^{n_k} \sum_{jk \neq ik} 
        S^r_{ik,jk} (X_{jk}-\bar{X}) \, 
        \mathbf{1}\{\theta_{h,ijk}(\alpha)=\theta_{h,a}(\alpha)\} 
        \right\rvert 
        \;\overset{N\to\infty}{\longrightarrow}\; 0,
    \end{equation*}
    for $h \in \{3,4\}$, where $\bar{X}$ is defined in Definition~\ref{struc_APO}.  

    \item[2.] (a) $\sum_{jk \neq ik} S^r_{jk|ik}=1$ for any $S^r_{jk|ik}\neq 0$ and $ik \in \{1,\ldots,n_k\}$, $k \in \{1,\ldots,K\}$.  
    (b) $\theta_{h,ijk}(\alpha)=\theta_{h,ik}(\alpha)$ for $h\in \{3,4\}$ and for all $i,j \in \{1,\ldots,n_k\}$ and $k \in \{1,\ldots,K\}$.  
    (c) Assume (b) holds. Then there exist $u_h$ finite and distinct values such that 
    $\theta_{h,ik}(\alpha) \in \{\theta_{h,1}(\alpha), \ldots, \theta_{h,u_h}(\alpha)\}$.  
    For each $a \in \{1,\ldots,u_h\}$, 
    \begin{equation*}
      S_{N}  \left\lvert 
        \sum_{k=1}^K \sum_{ik=1}^{n_k} 
        S^r_{ik} (X^{\dagger o}_{ik}-\bar{X}) \, 
        \mathbf{1}\{\theta_{h,ik}(\alpha)=\theta_{h,a}(\alpha)\} 
        \right\rvert 
        \;\overset{N\to\infty}{\longrightarrow}\; 0,
    \end{equation*}
    for $h \in \{3,4\}$. $\bar{X}$ is defined in Definition~\ref{struc_APO}, and $\tilde{X}^{\dagger o}_{ik}$ in Definition~\ref{Est_CSE_R}.  
\end{itemize}  
\end{assumption}

Statement 1 in Assumption \ref{coef_hete_restrict} requires that, within each group of units sharing the same coefficient value $\theta_{h,a}(\alpha)$, the weighted average of $X_{jk}$ converges to the weighted overall mean $\bar{X}$.  
This condition is natural when heterogeneity associated with $X_{jk}$ can be captured by terms of the form $\theta_{h,ijk}(\alpha)\cdot X_{jk}$, where $\theta_{h,ijk}(\alpha)$ is not itself a function of $X_{jk}$.  
Importantly, the subsets of units over which $\theta_{3,h}(\alpha)$ and $\theta_{4,h}(\alpha)$ are homogeneous do not need to coincide. Statement 2 in Assumption \ref{coef_hete_restrict} is stronger than Statement 1. Specifically, it imposes that: (a) the conditional weights sum to 1; (b) conditional on (a), $\theta_{h,ijk}(\alpha)$ is homogeneous across $j$ for a given $i$ within cluster $k$; and (c) conditional on (a) and (b), for each group of units sharing the same coefficient value $\theta_{h,a}(\alpha)$, the weighted average of $X^{\dagger o}_{ik}$ converges to the weighted overall mean $\bar{X}$. 

Note that when the coefficients are homogeneous, i.e., $\theta_{h,ijk}(\alpha)=\theta_{h}(\alpha)$ for $h \in \{3,4\}$, Statement 1 in Assumption \ref{coef_hete_restrict} is automatically satisfied, whereas Statement 2 in Assumption \ref{coef_hete_restrict} need not hold if the sum of the conditional weights $S^r_{jk\mid ik}$ does not equal 1. 

Overall, these conditions restrict coefficient heterogeneity in a way that allows $\boldsymbol{\beta}^r_{A,(3,4)}(\alpha)$ to be linked to population-weighted averages of the coefficients $(\beta^p_{A,3}(\alpha), \beta^p_{A,4}(\alpha))^T$, where $\beta^p_{A,3}(\alpha) = \beta^p_3(\alpha)$ for all $A \in \{D, R, S\}$, and where $\beta^p_3(\alpha)$ and $\beta^p_{A,4}(\alpha)$ will be formally defined in Proposition~\ref{beta_r_to_beta_p}.
Establishing the connection from $\boldsymbol{\beta}^r_{A,(3,4)}(\alpha)$ to these intermediate quantities provides important insights into how Assumption \ref{coef_hete_restrict} serves as a bridge to the CSE.

\begin{proposition}
\label{beta_r_to_beta_p}
Under Statement 1 of Assumption \ref{coef_hete_restrict} for $h=4$, 
we have 
\[
   S_{N}\big|\beta^r_{A,3}(\alpha)-\beta^p_{3}(\alpha)\big| \;\overset{N\rightarrow \infty}{\longrightarrow}\; 0,
\]
for $A \in \{D,S,R\}$, where 
\begin{equation}
\label{beta_r_to_beta_p_equa1}
\beta^p_3(\alpha)=\beta^p_{D,3}(\alpha)=\beta^p_{S,3}(\alpha)=\beta^p_{R,3}(\alpha)
:=
   \sum_{k=1}^K \sum_{ik=1}^{n_k} \sum_{jk\neq ik} S^r_{ik,jk}\, \beta_{3,ijk}(\alpha).
\end{equation}
Furthermore, under Statement 1 of Assumption \ref{coef_hete_restrict}, 
\[
 S_{N}  \big|\beta^r_{A,4}(\alpha)-\beta^p_{A,4}(\alpha)\big| 
   \;\overset{N\rightarrow \infty}{\longrightarrow}\; 0,
\]
for $A \in \{D,S\}$, where 
\[
 \beta^p_{A,4}(\alpha)
 :=\Bigg(\sum_{k=1}^K \sum_{ik=1}^{n_k} \sum_{jk\neq ik} S_{ik,jk}\,\tilde{X}_{jk}^2\Bigg)^{-1}
   \sum_{k=1}^K \sum_{ik=1}^{n_k} \sum_{jk\neq ik} S_{ik,jk}\,\tilde{X}_{jk}^2 \beta_{4,ijk}(\alpha).
\]
Under Statement 2 of Assumption \ref{coef_hete_restrict}, the same result holds for $A=R$, with
\[
 \beta^p_{R,4}(\alpha)
 :=\Bigg(\sum_{k=1}^K \sum_{ik=1}^{n_k} S_{ik}\,{X}^{\dagger 2}_{ik}\Bigg)^{-1}
   \sum_{k=1}^K \sum_{ik=1}^{n_k} S_{ik}\,{X}^{\dagger 2}_{ik}\,\beta_{4,ik}(\alpha).
\]
\end{proposition}
\noindent
The proof is in Appendix~\ref{Appendix:inf_CSE}. Proposition~\ref{beta_r_to_beta_p} states that, under restrictions on the heterogeneity of $\theta_{h,ijk}(\alpha)$ and on the averages of $X_{jk}$ among units with common values of $\theta_{h,ijk}(\alpha)$, the ratio quantity $\beta^r_{A,h}(\alpha)$ converges to the population quantity $\beta^p_{A,h}(\alpha)$ for $h \in \{3,4\}$ and $A \in \{D,R,S\}$. 
Note that the coefficient vector \(\boldsymbol{\beta}^r_A(\alpha)\) is homogeneous and therefore constitutes a misspecification relative to the possibly heterogeneous coefficients \(\beta_{h,ijk}(\alpha)\). However, provided that Assumption \ref{coef_hete_restrict} holds, Proposition \ref{beta_r_to_beta_p} ensures that  $\boldsymbol{\beta}^r_{A,(3,4)}(\alpha)$ converges to population-weighted averages of possibly heterogeneous  coefficients.
Together with Proposition~\ref{Const_ratio_estimable_quantity}, this result further implies that the estimator \( \hat{\beta}_{A,h}(\alpha) \) consistently estimates the population average \( \beta^p_{A,h}(\alpha) \) for \( h \in \{3,4\} \) and \( A \in \{D,R,S\} \).

\begin{remark}  
Under restrictions on the heterogeneity of \( \theta_{4,ijk}(\alpha) \) (rather than \( \theta_{3,ijk}(\alpha) \)) and on the group averages of \( X_{jk} \) corresponding to common values of \( \theta_{4,ijk}(\alpha) \), the quantity \( \beta^r_{A,3}(\alpha) \) converges to its population counterpart \( \beta^p_{3}(\alpha) \).
Statement~1 in Assumption~\ref{coef_hete_restrict} for \( h = 4 \) ensures that the sum of interaction terms \( \theta_{4,ijk}(\alpha) \tilde{X}_{jk} \) in \( \beta^r_{A,3}(\alpha) \) converges to zero.
\end{remark}

\begin{remark}
From Definition~\ref{struc_APO}, \( \beta^p_{3}(\alpha) \) in \eqref{beta_r_to_beta_p_equa1} can be written as
$$
\sum_{k=1}^K \sum_{ik=1}^{n_k} \sum_{jk \neq ik} 
S^r_{ik,jk}\,\bigl(\theta_{3,ijk}(\alpha)
+ \theta_{4,ijk}(\alpha)\,\bar{X}\bigr).
$$
Hence, \( \beta^p_{3}(\alpha) \) generally differs from \( \beta^r_{A,3}(\alpha) \), where
\[
\beta^r_{3}(\alpha)
= 
   \sum_{k=1}^K \sum_{ik=1}^{n_k} \sum_{jk\neq ik} 
S^r_{ik,jk}\,\bigl(\theta_{3,ijk}(\alpha)
+ \theta_{4,ijk}(\alpha)\,X_{jk}\bigr).
\]
The two quantities coincide asymptotically when Statement~1 of Assumption~\ref{coef_hete_restrict} holds.
\end{remark}

\begin{remark}
The conditions required for the consistency of
\( \hat{\beta}_{A,4}(\alpha) \) with respect to
\( \beta^p_{A,4}(\alpha) \) differ across perspectives.
For \( A \in \{D,S\} \), weaker restrictions—namely,
Statement~1 in Assumption~\ref{coef_hete_restrict}—are sufficient.
In contrast, for \( A = R \), stronger and different conditions—namely,
Statement~2 in Assumption~\ref{coef_hete_restrict}—are required.
Statement~2 in Assumption~\ref{coef_hete_restrict} is stronger in the sense that it imposes stronger restrictions on the estimand weights and on the homogeneity of
\( \theta_{h,ijk}(\alpha) \), compared with
Statement~1 in Assumption~\ref{coef_hete_restrict}.
\end{remark}

\begin{remark}
In general, \( \beta^r_{A,4}(\alpha) \) and \( \beta^p_{A,4}(\alpha) \), for 
\( A \in \{D,S\} \), differ from \( \beta^r_{R,4}(\alpha) \) and \( \beta^p_{R,4}(\alpha) \), respectively. 
Consequently, the assumptions required to link 
\( \hat{\beta}_{A,4}(\alpha) \) to the conditional spillover effect for \( A \in \{D,S\} \) differ from those needed for 
\( \hat{\beta}_{R,4}(\alpha) \). 
By contrast, \( \beta^r_{A,3}(\alpha) \) and \( \beta^p_{A,3}(\alpha) \) 
are identical across all perspectives 
\( A \in \{D,S,R\} \). 
Hence, the assumptions connecting \( \hat{\beta}_{A,3}(\alpha) \) to the CSE 
are common to all three formulations.
\end{remark}

\begin{remark}
Different estimand weights for alternative CSE estimands can make Statement~1 or Statement~2 in Assumption~\ref{coef_hete_restrict} easier to satisfy. Thus, for a given estimator, some estimands may be easier to estimate consistently than others. For example, the estimand weights used by the receiver estimator for the conditional inward spillover effect satisfy condition (a) in Statement~2 of Assumption~\ref{coef_hete_restrict}, whereas those for the conditional outward spillover effect do not. Statement~2 links $\beta^r_{R,4}(\alpha)$ and $\beta^p_{R,4}(\alpha)$. Hence, when using the receiver estimator for these two estimands, consistency is easier to achieve for the conditional inward spillover effect than for the conditional outward spillover effect.
\end{remark}
Given the definitions of $\beta^r_{A,h}(\alpha)$, $\beta^p_{A,h}(\alpha)$, and $\bar{\beta}_{h}(\alpha,x)$ for $h \in \{3,4\}$ provided above, the ratio quantity can be expressed as a weighted sum of the dyadic pairwise spillover effects introduced in Definition \ref{pair_spillover}. Substituting the structural model specified in Definition \ref{struc_APO} implies that this ratio is equivalently a weighted sum of $\theta_{3,ijk}(\alpha)$ and $\theta_{4,ijk}(\alpha)$, as characterized in Proposition \ref{Const_ratio_estimable_quantity}, with weights that may differ from $S_{ik,jk}$. In contrast, $\beta^p_{A,h}(\alpha)$ represents a population quantity given by a weighted sum of $\theta_{h,ijk}(\alpha)$ for the corresponding $h$, where the weights are $S_{ik,jk}$. Compared with $\beta^p_{A,h}(\alpha)$, $\bar{\beta}_{h}(\alpha,x)$ denotes a restricted weighted sum of $\theta_{h,ijk}(\alpha)$, defined over the subpopulation of units with covariate values equal to $x$. Consequently, $\beta^p_{A,h}(\alpha)$ can be interpreted as the population-level counterpart of the subpopulation-specific quantity $\bar{\beta}_{h}(\alpha,x)$.

In Proposition \ref{beta_r_to_beta_p}, we have already established the link between 
$\boldsymbol{\beta}_{A,(3,4)}^r(\alpha)$ and $({\beta}_{A,3}^p(\alpha), {\beta}_{A,4}^p(\alpha))^T$, and hence between 
$\boldsymbol{\hat{\beta}}_{A,(3,4)}(\alpha)$ and $({\beta}_{A,3}^p(\alpha), {\beta}_{A,4}^p(\alpha))^T$. We now turn to the connection between $({\beta}_{A,3}^p(\alpha), {\beta}_{A,4}^p(\alpha))^T$ and 
$(\bar{\beta}_{3}(\alpha, x), \bar{\beta}_{4}(\alpha, x))^T$, where the latter is defined in Definition 
\ref{cond_spill_over_under_struc_model}. The following assumption is required for this connection.

\begin{assumption}[Independence between ${{\theta}}_{h,ijk}(\alpha)$ and covariates]
\label{assump:indep_tilde_beta_bar_beta_x}
For $h\in\{3,4\}$, define
\[
\theta^p_{h}(\alpha)
:=\sum_{k=1}^K \sum_{ik=1}^{n_k} \sum_{jk \neq ik} S^r_{ik,jk}\, \theta_{h,ijk}(\alpha), 
\quad 
\theta_{h}(\alpha,x):=\sum_{k=1}^K \sum_{ik=1}^{n_k} \sum_{jk \neq ik} 
S^r_{ik,jk}(x)\,\theta_{h,ijk}(\alpha),
\]
where $S^r_{ik,jk}(x)$ is defined in Definition \ref{cond_spill_over_under_struc_model}. 
We assume that, for each $x \in \mathcal{X}$, where $\mathcal{X}$ denotes the support of the covariate,
\begin{equation}
\label{assump:indep_tilde_beta_bar_beta_x_int1}
S_{N}\bigl\lvert \theta^p_{h}(\alpha)-\theta_{h}(\alpha,x) \bigr\rvert 
\overset{N\to \infty}{\longrightarrow} 0,
\end{equation}
for $h \in \{3,4\}$. Moreover,
\begin{equation}
\label{assump:indep_tilde_beta_bar_beta_x_int2}
S_{N}\bigl\lvert \theta^p_{A,4}(\alpha)-\theta_{4}(\alpha,x) \bigr\rvert 
\overset{N\to \infty}{\longrightarrow} 0,
\end{equation}
where $\theta^p_{A,4}(\alpha) = \beta^p_{A,4}(\alpha)$ and $\beta^p_{A,4}(\alpha)$ is defined in Proposition \ref{beta_r_to_beta_p} for $A \in \{D, S, R\}$.
\end{assumption}

Assumption~\ref{assump:indep_tilde_beta_bar_beta_x} states that the population-weighted average of \( \theta_{h,ijk}(\alpha) \), based on weights \( S^r_{ik,jk} \), is asymptotically equal to its conditional weighted average \( \theta_{h}(\alpha, x) \) for \( h \in \{3,4\} \). 
This condition ensures that the population average \( \beta^p_{A,3}(\alpha) \)\footnote{Note that \( \theta^p_{3}(\alpha) \) and \( \theta^p_{4}(\alpha) \) in Assumption \ref{assump:indep_tilde_beta_bar_beta_x} are components of \( {\beta}^{p}_{A,3}(\alpha) \).} is asymptotically equal to \( \bar{\beta}_{3}(\alpha, x) \). 
In addition, the population averages of \( \theta_{4,ijk}(\alpha) \), based on weights \( S^r_{ik,jk} \tilde{X}^2_{jk} \) (for \( A \in \{D,S\} \)) or \( S^r_{ik} \tilde{X}^{\dagger 2}_{ik} \) (for \( A \in \{R\} \)), are also asymptotically equal to \( \theta_{4}(\alpha, x) \). 
If this condition holds, then \( {\beta}^{p}_{A,4}(\alpha) \) and \( \bar{\beta}_{4}(\alpha,x) \) are asymptotically equivalent. These results are employed to establish the consistency of \( \hat{\beta}_{A,h}(\alpha) \) with respect to \( \bar{\beta}_{h}(\alpha, x) \) for \( h \in \{3, 4\} \).


With regard to the plausibility of Assumption~\ref{assump:indep_tilde_beta_bar_beta_x}, the assumption implies that $\theta_{4,ijk}(\alpha)$ itself does not depend on $X_{jk}$. This, in turn, also implies that the linear relationship between the dyadic pairwise spillover effects and $X_{jk}$ induced by the design matrix in the WLS estimators is correctly specified. In addition, note that $\beta^p_{A,3}(\alpha)$ is weighted only by $S_{ik,jk}$.
In contrast, $\beta^p_{A,4}(\alpha)$ is weighted not only by $S_{ik,jk}$ but also by $\tilde{X}_{jk}$ (for $A \in \{D,S\}$) or by $\tilde{X}^\dagger_{ik}$ (for $A \in \{R\}$). Consequently, \eqref{assump:indep_tilde_beta_bar_beta_x_int2} imposes a more stringent requirement than \eqref{assump:indep_tilde_beta_bar_beta_x_int1} in Assumption~\ref{assump:indep_tilde_beta_bar_beta_x}. Observe that when the coefficients are homogeneous, that is, when $\theta_{h,ijk}(\alpha)=\theta_{h}(\alpha)$ for $h\in\{3,4\}$, conditions \eqref{assump:indep_tilde_beta_bar_beta_x_int1} and \eqref{assump:indep_tilde_beta_bar_beta_x_int2} in Assumption \ref{assump:indep_tilde_beta_bar_beta_x} are satisfied automatically. 


Based on Proposition~\ref{beta_r_to_beta_p} and Assumption~\ref{assump:indep_tilde_beta_bar_beta_x}, we now establish the consistency of the estimated coefficients with respect to the coefficients defining the CSE, as stated below. 

\begin{proposition}
\label{CSE_cons_ratio_population_respectively}
Suppose that Assumptions~\ref{part_intf}--\ref{unif_bound_pot_out}, in conjunction with Assumptions~\ref{bounded_x}, \ref{assump:indep_tilde_beta_bar_beta_x}, and the rate condition \eqref{ase_est_const_rate_cond}, are satisfied. Under Statement~1 of Assumption~\ref{coef_hete_restrict},
\[
   S_{N}\big|\hat{\beta}_{A,h}(\alpha)-\bar{\beta}_{h}(\alpha,x)\big| 
   \;\overset{N\rightarrow \infty}{\longrightarrow}\; 0, 
   \quad h \in \{3,4\}, \; A\in \{D,S\}.
\]
Under Statements~1 and~2 of Assumption~\ref{coef_hete_restrict}, 
\[
   S_{N}\big|\hat{\beta}_{R,h}(\alpha)-\bar{\beta}_{h}(\alpha,x)\big| 
   \;\overset{N\rightarrow \infty}{\longrightarrow}\; 0, 
   \quad h \in \{3,4\}.
\]
\end{proposition}
\noindent The proof is in Appendix~\ref{Appendix:inf_CSE}. Proposition~\ref{CSE_cons_ratio_population_respectively} states the consistency of 
\(\hat{\beta}_{A,h}(\alpha)\) with respect to \(\bar{\beta}_{h}(\alpha,x)\). This follows because
Assumptions~\ref{part_intf}--\ref{unif_bound_pot_out}
ensure the consistency of \(\hat{\beta}_{A,h}(\alpha)\) with respect to
\(\beta^r_{A,h}(\alpha)\), while Assumptions~\ref{bounded_x}--\ref{assump:indep_tilde_beta_bar_beta_x} establish the asymptotic equivalence between
\(\beta^r_{A,h}(\alpha)\) and \(\bar{\beta}_{h}(\alpha,x)\).

Proposition~\ref{CSE_cons_ratio_population_respectively} establishes the connection between the estimators and the coefficients defining the CSE through the corresponding ratio and population quantities.
Although the CSE estimand itself does not depend on the choice of perspective used to formulate the estimators, the associated ratio and population quantities generally do, differing across \( A \in \{D, S\} \) and \( A = R \).


We now describe two settings under which the intermediate ratio and population quantities, respectively, coincide across estimator formulations. These settings concern: (a) the coefficients in the dyadic average potential outcome model of Definition~\ref{struc_APO}; (b) the underlying graph structure; and (c) the estimand weights.

\begin{lemma}
\label{settings_all_conditions_satisfy}
Consider the following two settings. \begin{itemize}[leftmargin=58pt,label={}]
    \item[\textbf{Setting 1.}] For the conditional outward spillover effect defined in Example~\ref{Examp_cond_outward}:  
    (i) $\theta_{h,ijk}(\alpha)=\theta_h(\alpha)$ for $h \in \{3,4\}$ and all 
    $ik,jk \in \{1,\ldots,n_k\}$, $k \in \{1,\ldots,K\}$;  
    (ii) the graph is a directed regular graph with degree $d>0$.  

    \item[\textbf{Setting 2.}] For the conditional inward spillover effect defined in 
    Example~\ref{Examp_cond_inward}:  
    $\theta_{h,ijk}(\alpha)=\theta_h(\alpha)$ for $h \in \{3,4\}$ and all 
    $ik,jk \in \{1,\ldots,n_k\}$, $k \in \{1,\ldots,K\}$.
\end{itemize}
Then,  under either Setting 1 or 2, we have:  
\begin{enumerate}
\item The intermediate coefficients satisfy 
\[
\beta^l_{D,h}(\alpha)=\beta^l_{S,h}(\alpha)=\beta^l_{R,h}(\alpha)
\quad \text{for all } l \in \{r,p\}, \, h \in \{3,4\}.
\]
\item  All conditions in Assumptions~\ref{coef_hete_restrict} and~\ref{assump:indep_tilde_beta_bar_beta_x} 
are automatically satisfied. Consequently, under Assumptions~\ref{part_intf}--\ref{pos_Ratio} 
and~\ref{unif_bound_pot_out}--\ref{struc_APO}, 
Proposition~\ref{CSE_cons_ratio_population_respectively} holds for all 
$A \in \{D,S,R\}$. 
\end{enumerate}
\end{lemma}
\noindent The proof of Lemma~\ref{settings_all_conditions_satisfy} is provided in Appendix~\ref{Appendix:connect_CSE}.
Lemma~\ref{settings_all_conditions_satisfy} identifies conditions under which the ratio quantities are identical across the three formulations, and likewise for the population quantities.
It is important to note that this lemma differs from Proposition~\ref{CSE_cons_ratio_population_respectively}:
Lemma~\ref{settings_all_conditions_satisfy} characterizes the relationships among intermediate quantities across estimator formulations (dyadic, sender, and receiver),
whereas Proposition~\ref{CSE_cons_ratio_population_respectively} concerns the relationships among the estimator, its corresponding ratio quantity, and the population quantity within a given formulation. Thus, even when Assumptions~\ref{coef_hete_restrict} and~\ref{assump:indep_tilde_beta_bar_beta_x} hold—ensuring asymptotic equivalence between the ratio and population quantities within each formulation perspective—the intermediate ratio and population quantities may still differ across formulations in general.


\subsection{Inference for estimators of CSE}
\label{sec:Inference}
Provided the assumptions linking the estimators to the quantities involved in the CSE,
we now present the theorem that establishes the consistency of $\hat{\tau}_{A}(\alpha, x)$ for $\tau(\alpha, x)$.

\begin{theorem}[Consistency of $\hat{\tau}_{A}(\alpha,x)$]
\label{consistency_hat_tau_x_cond_tau_x}
Let $\hat{\tau}_{A}(\alpha,x)$ for $A \in \{D, R, S\}$ be defined as in 
Definitions~\ref{Est_CSE_dyad}, \ref{Est_CSE_R}, and~\ref{Est_CSE_S}, 
and let $\tau(\alpha,x)$ be defined as in 
Definition~\ref{cond_spill_over_under_struc_model}.  
Suppose Assumptions~\ref{part_intf}--\ref{unif_bound_pot_out}, \ref{bounded_x}--\ref{assump:indep_tilde_beta_bar_beta_x}, as well as the rate condition~\eqref{ase_est_const_rate_cond}, hold. For $A \in \{D, S\}$, under Statement~1 of 
Assumption~\ref{coef_hete_restrict}, we have, for each 
$x \in \mathcal{X}$, where $\mathcal{X}$ denotes the support of~$x$,
\[
\bigl| \hat{\tau}_{A}(\alpha,x) - \tau(\alpha,x) \bigr|
\;\overset{N \to \infty}{\longrightarrow}\; 0.
\]
For $A = R$, the same conclusion holds under 
Statements~1 and~2 of Assumptions~\ref{coef_hete_restrict}.
\end{theorem}
\noindent The proof is in Appendix~\ref{Appendix:inf_CSE}. It is worth noting that when the coefficients $\theta_{h,ijk}(\alpha)$ for $h \in \{3,4\}$ are homogeneous, Statement 1 in Assumption \ref{coef_hete_restrict} and Assumption \ref{assump:indep_tilde_beta_bar_beta_x} are automatically satisfied. Consequently, $\hat{\tau}_{A}(\alpha, x)$ for $A \in \{D, S\}$ are consistent estimators of $\tau(\alpha, x)$. On the other hand, for the receiver estimator, $\hat{\beta}_{R,3}(\alpha)$ remains a consistent estimator of $\bar{\beta}_{3}(\alpha, x)$; however, the consistency of $\hat{\beta}_{R,4}(\alpha)$ for $\bar{\beta}_{4}(\alpha, x)$ is not guaranteed. If, in addition to homogeneity of the coefficients, estimand weights are such that Statement 2(a) in Assumption \ref{coef_hete_restrict} holds, then Statement 2 of Assumption \ref{coef_hete_restrict} is also satisfied. This, in turn, implies the consistency of $\hat{\tau}_{R}(\alpha, x)$ for $\tau(\alpha, x)$. 

Furthermore, although Assumptions \ref{coef_hete_restrict} and \ref{assump:indep_tilde_beta_bar_beta_x} do restrict the relationship between $\beta_{4,ijk}(\alpha)$ and $X_{jk}$, as well as the degree of heterogeneity in the coefficients of the structural models in Definition \ref{struc_APO}, they  do not impose any specific functional-form restrictions on the manner in which other units’ treatments are linked to an individual’s potential outcomes. Consequently, in comparison to the conventional practice of employing regression-based methods with exposure mapping functions to estimate spillover effects, our estimators operate under substantially weaker assumptions. Since these common regression-based methods can be viewed as adopting an effect-receiver perspective, it is worth highlighting that our receiver estimator is also able to avoid functional-form restrictions on the treatments of the effect senders by incorporating them in the weight matrix. 

We next establish a central limit theorem for 
$\hat{\tau}_{A}(\alpha, x)$ with $A \in \{D, S, R\}$ \citep{ogburn2022causal}. 
Since most regularity and rate conditions overlap with those in 
Theorem~\ref{CLT_ASE_est}, we highlight here the additional conditions required.

\begin{theorem}[CLT for $\hat{\tau}_{A}(\alpha,x)$]
\label{CLT_est}
Let 
$\tau^r_{A}(\alpha,x)
:= S_{N} (1,\tilde{x})
\bigl[\boldsymbol{\beta}^r_{A,(3,4)} 
- ({\beta}_{A,3}^p(\alpha),{\beta}_{A,4}^p(\alpha))^T\bigr]$, $\tau^p_{A}(\alpha,x)
:= S_{N} (1,\tilde{x})
\bigl[({\beta}_{A,3}^p(\alpha),{\beta}_{A,4}^p(\alpha))^T 
- {\boldsymbol{\bar\beta}}_{(3,4)}(\alpha,x)\bigr]$ for $A \in \{D,S,R\}$. 
Suppose Assumptions~\ref{part_intf}--\ref{unif_bound_pot_out} and~\ref{bounded_x}, as well as the rate condition \eqref{ase_CLT_rate_cond} in Theorem~\ref{CLT_ASE_est}, hold. 
Furthermore, for each conditioning value $x \in \mathcal{X}$, suppose that 
\[
{\tau}^r_{A}(\alpha,x) - {\tau}^p_{A}(\alpha,x) 
= o(\omega_N)
\quad \text{and} \quad
{\tau}^p_{A}(\alpha,x) - {\tau}_{A}(\alpha,x) 
= o(\omega_N),
\]
where $\omega_N
:= K^{-1/2}\, \bar{n}_k^{-2}\,
\min_{(i_1k,j_1k),(i_2k,j_2k)} 
S^{-1/2}_{i_1k,j_1k} S^{-1/2}_{i_2k,j_2k}
\longrightarrow 0.$ Then
\[
\tilde{N}^{1/2} \Bigl[\tilde{N}(1,\tilde{x}) \, 
\Sigma_{A,(3,4),(3,4)} \, (1,\tilde{x})^{\top}\Bigr]^{-1/2}
\bigl(\hat{\tau}_{A}(\alpha,x)
- \tau(\alpha,x)\bigr)
\;\xrightarrow[]{N \to \infty}\;
\mathcal{N}(0,1),
\]
where $\Sigma_{A}
= \Omega_{A}^{-1}\tilde{\Gamma}_{A}\Omega_{A}^{-1}$, with $\Omega_{A} = S_{N}^{-1}\, \mathbb{E}(V_A^{\top} B_A V_A)$, $\tilde{\Gamma}_{A} = \mathrm{var}(V_A^{\top} B_A \xi_A)$, and $\xi_A = Y_A - V_A^{\top}\boldsymbol{\beta}^r_{A}(\alpha)$, determined by 
$Y_A$, $V_A$, and $\boldsymbol{\beta}^r_{A}(\alpha)$ as defined in 
Proposition~\ref{Const_ratio_estimable_quantity}, for $A \in \{D,S,R\}$. The quantity $\tilde{N}$ is defined in Theorem~\ref{CLT_ASE_est}.
\end{theorem}
\noindent The proof is in Appendix~\ref{Appendix:inf_CSE}. Theorem~\ref{CLT_est} shows that, in addition to the assumptions and rate condition 
\eqref{ase_CLT_rate_cond} required in Theorem~\ref{CLT_ASE_est}, the convergence of ${\tau}^r_{A}(\alpha,x)$ to ${\tau}^p_{A}(\alpha,x)$ 
and of ${\tau}^p_{A}(\alpha,x)$ to ${\tau}(\alpha,x)$ 
must occur at a rate faster than the order of the inverse of the standard error of 
$\hat{\tau}_{A}(\alpha,x)$, that is, faster than $O(\omega_N^{-1})$. The rationale for the rate requirement \eqref{ase_CLT_rate_cond} is the same as that discussed in Theorem~\ref{CLT_ASE_est}. Under these conditions, the central limit theorem holds for 
$\hat{\tau}_{A}(\alpha,x)$. The rate of convergence of $\hat{\tau}_{A}(\alpha,x)-\tau(\alpha,x)$ to the normal distribution is $\tilde{N}^{-1/2}$, as shown by the same argument used in the proof of Theorem~\ref{CLT_ASE_est} in Appendix~\ref{Appendix:inf_ASE}.

\begin{remark}
The proof of the CLT, reported in Appendix \ref{Appendix:inf_CSE}, follows the arguments of \citet{ogburn2022causal}. Under partial interference and the asymptotic regime considered here, one could alternatively establish Theorem~\ref{CLT_est} by invoking Lemma~B.1 of \citet{viviano2024policydesignexperimentsunknown}. Moreover, given independence across clusters, a Lyapunov-type argument at the cluster level could also establish asymptotic normality, even when cluster sizes grow. We rely on the framework of \citet{ogburn2022causal} because it naturally extends to more general interference structures.
\end{remark}

We now turn to the cluster-robust variance estimator, which is asymptotically conservative under partial interference.

\begin{proposition}[Cluster-robust variance estimator for $\hat{\tau}_{A}(\alpha,x)$]
\label{conserve_variance_est}
Under the same notation and conditions as in Theorem~\ref{CLT_est}, 
an asymptotically conservative estimator of 
$\Sigma_{A}$ is $\hat{\Sigma}_{A}
= \hat{\Omega}^{-1}_{A} \,
\hat{\Gamma}_{A}\hat{\Omega}^{-1}_{A}$, such that, with probability $1 - o(1)$,
$$
\Sigma_A \preceq \hat{\Sigma}_A + o_p(1),
$$
where $o_p(1)$ denotes a $4\times 4$ random matrix whose entries converge to zero in probability. Here, $\hat{\Omega}_{A}
= \rho^{-1}_{N} V_A^{\top} B_A V_A$, 
and 
$\hat{\Gamma}_{A}
= \sum_{k=1}^K 
V_{A,k}^{\top} B_{A,k} 
\hat{\xi}_{A,k} \hat{\xi}_{A,k}^{\top} 
B_{A,k} V_{A,k}$,
where $\hat{\xi}_{A} = Y_A - V_A \boldsymbol{\hat{\beta}}_{A}(\alpha)$.  
For each cluster $k$, $V_{A,k}^{\top}$, $B_{A,k}$, and $\hat{\xi}_{A,k}$ denote, respectively, 
the design matrix, the diagonal weight matrix, and the estimated residuals, 
for $A \in \{D,R,S\}$.  
A conservative variance estimator for $\hat{\tau}_{A}(\alpha,x)$ is then given by
\[
(1,\tilde{x}) \, 
\hat{\Sigma}_{A,(3,4),(3,4)} \, 
(1,\tilde{x})^{\top}.
\]
\end{proposition}
\noindent The proof is in Appendix~\ref{Appendix:inf_CSE}.
\begin{remark}
Let ${\Gamma}_{A} = \sum_{k=1}^K \mathbb{E}(V_{A,k}^{\top} B_{A,k}  {\xi}_{A,k} {\xi}_{A,k}^{\top}  B_{A,k} V_{A,k})$. Under partial interference, the dependence structure is symmetric: if the outcome of unit $ik$ is affected by the treatment of unit $jk$, then the outcome of $jk$ is likewise affected by the treatment of $ik$. Consequently,  $\Gamma_{A} \succeq \sum_{k=1}^K  \mathrm{var}\!\left(V^{\top}_{A,k} B_{A,k} \xi_{A,k}\right)$, since  $\sum_{k=1}^K  \mathbb{E}\!(V^{\top}_{A,k} B_{A,k} \xi_{A,k}) \![\mathbb{E}\!(V^{\top}_{A,k} B_{A,k} \xi_{A,k})]^{\!\top}$ is positive semidefinite. In this case, $\hat{\Sigma}_{A}$ is indeed an asymptotically conservative estimator of $\Sigma_{A}$. By contrast, when interference sets are not identical across units—as in $d$-distance neighborhood interference with $d \ge 1$—or approximately not identical, as in approximate network interference \citep{leung2022causal}, where interference sets are the same across units (including the unit itself) but the strength of interference decays with distance, then modified versions of $\hat{\Gamma}_{A}$ are required to ensure conservativeness; see 
\citet{wang2024designbasedinferencespatialexperiments}, 
\citet{leung2022causal}, and 
\citet{gao2025causalinferencenetworkexperiments} for detailed discussions. 
\end{remark}

\begin{remark}
Although the estimand weight $S^{r}_{ik,jk}$ is normalized by $S_N$, as defined in Definition~\ref{gen_estimand}, the convergence rates of both the estimator $\hat{\tau}_{A}(\alpha, x)$ and the variance estimator in Proposition~\ref{conserve_variance_est} remain of the same order as those for the ASE estimators. This is because the ratio structure of the WLS estimators cancels out the scaling factor $S_N$.
\end{remark}

We now compare the conservative variances of
$\hat{\tau}_{A}(\alpha, x)$ for $A \in \{D, S\}$ with that of
$\hat{\tau}_{R}(\alpha, x)$. The conservative variance is defined as $\Sigma^{c}_{A}
:= \Omega^{-1}_{A} \, \Gamma_{A} \, \Omega^{-1}_{A}$, where
\begin{equation}
\label{Gamma_A_formula}
    \begin{split}
        {\Gamma}_{A}
:= \sum_{k=1}^K \mathbb{E}(
V_{A,k}^{\top} B_{A,k} 
{\xi}_{A,k} {\xi}_{A,k}^{\top} 
B_{A,k} V_{A,k}),
    \end{split}
\end{equation}
and $\Omega_{A}$ is as defined in Theorem~\ref{CLT_est}.
From the proofs of Propositions~\ref{ase_conserve_variance_est} and~\ref{conserve_variance_est}, $\Sigma^{c}_{A}$ is asymptotically larger than $\Sigma_{A}$, and $\hat{\Sigma}_{A}$ is a consistent estimator for $\Sigma^{c}_{A}$ for all $A \in \{D,S,R\}$. Moreover, because $\Omega_{D}=\Omega_{S}$ and $\Gamma_{D}=\Gamma_{S}$ (see the proof of Proposition~\ref{relation_est_CSE}), we have $\Sigma^{c}_{D}=\Sigma^{c}_{S}$. Thus, the relevant comparison is between $\Sigma^{c}_{A}$ for $A \in \{D,S\}$ and $\Sigma^{c}_{R}$. By contrast, for the ASE estimators, the proof of Theorem~\ref{equiv_hj_D_R_S} implies that $\Omega_{D}=\Omega_{S}=\Omega_{R}$ and $\Gamma_{D}=\Gamma_{S}=\Gamma_{R}$, and hence $\Sigma^{c}_{D}=\Sigma^{c}_{S}=\Sigma^{c}_{R}$. Therefore, no analogue of Proposition~\ref{variance_comp} arises for ASE.

\begin{proposition}
\label{variance_comp}
Let $\tilde{x}:=x-\bar{X}$, $\tilde{X}_{\mathrm{ave}}^{2}
:= \sum_{k=1}^{K} \sum_{\substack{ik=1 \\ jk \neq ik}}^{n_k}
S^{r}_{ik,jk}\, \tilde{X}_{jk}^{2}$, and  $X_{\mathrm{ave}}^{\dagger 2}
:= \sum_{k=1}^{K} \sum_{\substack{ik=1 \\ jk \neq ik}}^{n_k}
S^{r}_{ik,jk}\, (X^{\dagger}_{ik})^{2}$. Consider the expression
\begin{equation}
\label{dif_variance}
\begin{split}
&\Bigl[
4 \Gamma_{A,33}
- \bigl(\Gamma_{R,11} - 4 \Gamma_{R,13} + 4 \Gamma_{R,33}\bigr)
\Bigr] \\
&\quad + 2\tilde{x} \Bigl[
4 (\tilde{X}_{\mathrm{ave}}^{2})^{-1}\Gamma_{A,43}
- (X_{\mathrm{ave}}^{\dagger 2})^{-1}
\bigl(\Gamma_{R,21} - 2\Gamma_{R,23} - 2\Gamma_{R,14} + 4\Gamma_{R,34}\bigr)
\Bigr] \\
&\quad + \tilde{x}^{2}\Bigl[
4 (\tilde{X}_{\mathrm{ave}}^{2})^{-1}\Gamma_{A,44}
- (X_{\mathrm{ave}}^{\dagger 2})^{-1}
\bigl(\Gamma_{R,22} - 4\Gamma_{R,24} + 4\Gamma_{R,44}\bigr)
\Bigr].
\end{split}
\end{equation} 
If the potential outcomes, together with the distribution of the hypothetical and realized treatment assignments, yield a negative value of \eqref{dif_variance}, then for $A \in \{D,S\}$, $(1,\tilde{x})\, \Sigma^{c}_{A}\, (1,\tilde{x})^{\top}
\;<\;
(1,\tilde{x})\, \Sigma^{c}_{R}\, (1,\tilde{x})^{\top}$. If \eqref{dif_variance} is positive, then the inequality is reversed, i.e., $(1,\tilde{x})\, \Sigma^{c}_{A}\, (1,\tilde{x})^{\top}
\;>\;
(1,\tilde{x})\, \Sigma^{c}_{R}\, (1,\tilde{x})^{\top}$.
\end{proposition}
\noindent The proof is in Appendix~\ref{Appendix:inf_CSE}. The sign of \eqref{dif_variance} depends jointly on the potential outcomes, through $\xi_{A}$; on both the hypothetical and realized treatment assignments, through $B_{A}$; and on the construction of the design matrix $V_{A}$.

\section{Simulation study}
\label{sec:Simulation Results}
In this section, we present simulation results to assess the performance of the proposed estimators and their variance estimators for both ASE and CSE.

\subsection{Simulation study for ASE}
As established in Theorem \ref{equiv_hj_D_R_S}, the three estimators for ASE coincide and are equivalent to the nonparametric H\'{a}jek estimator. We conduct a simulation study to both illustrate this result and to show the performance of the WLS estimators.

\subsubsection{Data generating process}
We conduct simulations under two distinct potential outcome models with heterogeneous coefficients:

\begin{equation}
\label{ASE_PO1} 
    Y_{ik}(\mathbf{z}_k)= \zeta_{1ik} + \zeta_{2ik} z_{ik} + \zeta_{3ik} \sum_{hk \in \mathcal{N}^{in}_{ik} } z_{hk} + \epsilon_{ik} 
\end{equation}
\vspace{-0.8cm}
\begin{equation}
\label{ASE_PO2}
     Y_{ik}(\mathbf{z}_k)= \zeta_{1ik} + \zeta_{2ik} z_{ik} + \zeta_{3ik} \sum_{hk \in \mathcal{N}^{in}_{ik} } z_{hk} + \zeta_{4ik} \sum_{hk \in \mathcal{N}^{in}_{ik} } z_{hk} z_{ik} + \epsilon_{ik} 
\end{equation}
In both Models \eqref{ASE_PO1} and \eqref{ASE_PO2}, we generate $\zeta_{1ik} \sim \mathcal{N}(0.8,0.2^2)$ and $\zeta_{2ik} \sim \mathcal{N}(2,0.5^2)$. In Model \eqref{ASE_PO1}, we further set $\zeta_{3ik} \sim \mathcal{N}(1,0.1^2)$, while in Model \eqref{ASE_PO2}, $\zeta_{3ik} \sim \mathcal{N}(0.5,0.1^2)$ and $\zeta_{4ik} \sim \mathcal{N}(1,0.2^2)$. The noise terms are drawn as $\epsilon_{ik} \sim \mathcal{N}(0,0.2^2)$. We conduct $M=500$ Monte Carlo replications. The coefficients $\zeta_{hik}$ for $h \in \{1,2,3\}$ in Model~\eqref{ASE_PO1} and $h \in \{1,2,3,4\}$ in Model~\eqref{ASE_PO2} are held fixed across replications. Similarly, the noise terms $\epsilon_{ik}$ are fixed across replications, reflecting a design-based simulation setting. The network structure is generated from a directed Erd\H{o}s–R\'enyi model with edge probability $4/n_k$, where the cluster size is fixed at $n_k=20$ and the number of clusters varies with $K \in \{50,100,\ldots,500\}$. The treatment is assigned according to an i.i.d. Bernoulli design with probability $\beta=0.5$. We focus on two estimands: the average outward spillover effect (Example \ref{Examp_Ave_outward}, Section \ref{sec:Estimand}) and the average inward spillover effect (Example \ref{Examp_Ave_inward}, Section \ref{sec:Estimand}). The hypothetical treatment assignment coincides with the realized one with $\alpha=\beta$. Given the values of the parameters, the estimands $\tau(\alpha)$ are expected to be close to 1 in both Models \eqref{ASE_PO1} and \eqref{ASE_PO2}. 

\subsubsection{Simulation results}
Results of simulations under Model \eqref{ASE_PO2} for the outward spillover effect are presented in Table \ref{tab:PO2_outward_ASE}, while results for the inward spillover effect under Model \eqref{ASE_PO2} are shown in Table \ref{tab:PO2_inward_ASE} in Appendix B. Results for both outward and inward spillover effects under Model \eqref{ASE_PO1} are reported in Tables \ref{tab:PO1_outward_ASE} and \ref{tab:PO1_inward_ASE} of Appendix B, respectively.
\begin{table}[htbp]
\centering
\caption{Simulation results for the average outward spillover effect under 
potential outcome model~\eqref{ASE_PO2}}
\begin{tabular}{cccccccc}
\toprule
$K$ & $\mathbb{E}(\hat{\tau}_D(\alpha))$ & $\mathbb{E}(\hat{\tau}_S(\alpha))$ & $\mathbb{E}(\hat{\tau}_R(\alpha))$ & Bias & $se(\hat{\tau}_\cdot(\alpha))$ & $\mathbb{E}[\hat{se}(\hat{\tau}_\cdot(\alpha))]$ & 95\% coverage \\
\midrule
  50 & 0.996 & 0.996 & 0.996 &  0.002 & 0.196 & 0.182 & 0.906 \\
 100 & 1.018 & 1.018 & 1.018 &  0.013 & 0.147 & 0.135 & 0.932 \\
 150 & 1.004 & 1.004 & 1.004 &  0.004 & 0.112 & 0.110 & 0.944 \\
 200 & 1.005 & 1.005 & 1.005 &  0.001 & 0.092 & 0.095 & 0.952 \\
 250 & 1.004 & 1.004 & 1.004 &  0.005 & 0.083 & 0.085 & 0.944 \\
 300 & 0.998 & 0.998 & 0.998 &  0.001 & 0.074 & 0.077 & 0.954 \\
 350 & 1.001 & 1.001 & 1.001 & -0.002 & 0.071 & 0.071 & 0.962 \\
 400 & 1.006 & 1.006 & 1.006 &  0.002 & 0.064 & 0.067 & 0.966 \\
 450 & 1.004 & 1.004 & 1.004 &  0.004 & 0.058 & 0.063 & 0.964 \\
 500 & 0.995 & 0.995 & 0.995 & -0.002 & 0.058 & 0.060 & 0.956 \\
\bottomrule
\end{tabular}
\label{tab:PO2_outward_ASE}
\begin{minipage}{0.95\linewidth}
\par\smallskip
\noindent $\mathbb{E}(\hat{\tau}_{\cdot}(\alpha))$ denotes the Monte Carlo mean of the estimator. $se(\hat{\tau}_{\cdot}(\alpha))$ is the empirical standard error of $\hat{\tau}_{\cdot}(\alpha)$, computed as the sample standard deviation. $\mathbb{E}[\hat{se}(\hat{\tau}_{\cdot}(\alpha))]$ denotes the Monte Carlo average of the estimated standard errors.
\end{minipage}
\end{table}

Table \ref{tab:PO2_outward_ASE} shows that the three estimators
$\hat{\tau}_A(\alpha)$ for $A \in \{D,S,R\}$ coincide numerically, as implied by
Theorem~\ref{equiv_hj_D_R_S}. As the number of clusters $K$ increases, the estimated
standard errors become increasingly conservative relative to the true standard error,
which is consistent with Proposition~\ref{ase_conserve_variance_est}. 

The cluster-robust variance estimator shifts from anti-conservative at small $K$ to slightly conservative as $K$ increases. The anti-conservative behavior observed when $K$ is small is likely driven by two factors. First, the variance estimator is asymptotically conservative (Proposition~\ref{ase_conserve_variance_est}), so it may fail to be conservative when the number of clusters is limited. Second, with a small number of clusters $K$, empirical correlations across clusters may be non-negligible due to the limited number of Monte Carlo repetitions and thus contribute to the variance, whereas the cluster-robust variance estimator implicitly sets all cross-cluster covariances to zero, leading to an underestimation of uncertainty in small samples.

The coverage of the cluster-robust confidence intervals is below $95\%$ for small $K$, and slightly above but close to $95\%$ for large $K$. This pattern reflects that asymptotic normality is achieved for our estimators once the number of clusters exceeds $200$ in this setting.


\subsection{Simulation study for CSE}
\subsubsection{Data generating process}
For CSE, we first consider Setting~1 in Lemma~\ref{settings_all_conditions_satisfy}. Specifically, we posit a model for potential outcomes with homogeneous coefficients as follows:
\begin{equation}
\label{CSE_PO3}
    \begin{split}
     Y_{ik}(\mathbf{z}_k)= &  \zeta_{1} + \zeta_{2} X_{ik}+ \zeta_{3}z_{ik}+ \zeta_{4}z_{ik} X_{ik}+ \zeta_{5} \sum_{hk \in \mathcal{N}^{in}_{ik} } z_{hk} +  \zeta_{6} \sum_{hk \in \mathcal{N}^{in}_{ik} } z_{hk} X_{hk}  + \epsilon_{ik}. 
    \end{split}
\end{equation}
We set $\zeta_1=0.8$, $\zeta_2=2$, $\zeta_3=0.5$, $\zeta_4=0.7$, $\zeta_5=0.5$, and $\zeta_6=0.4$.
The network $\mathcal{G}$ is defined as a collection of clustered, directed graphs $\mathcal{G}_k$, each with degree 4.
The cluster size is fixed at $n_k=20$, and the number of clusters is set to $K \in \{50,100,\ldots,500\}$.
The treatment is assigned according to an i.i.d. Bernoulli design with probability $\beta=0.5$. Finally, a binary covariate $X_{ik}$ is generated with $\mathbb{P}(X_{ik}=1)=0.5$.
For the estimand, we focus on the conditional outward spillover effect as defined in Example~\ref{Examp_cond_outward} for $x=1$ and with $\alpha=\beta$.
 
By Lemma \ref{settings_all_conditions_satisfy}, we have $\beta^r_{h,A}(0.5)=\beta^p_{h,A}(0.5)=\bar{\beta}_{h}(0.5,x)$ for $h\in \{3,4\}$, $A\in\{D,S,R\}$, and $x\in\{0,1\}$. Therefore, in Tables \ref{tab:CSE_outward_homo_regular_grap_bias}--\ref{tab:CSE_outward_homo_regular_grap_se_tau}, we report $\beta^r_{3,A}(0.5)$ and $\beta^p_{3,A}(0.5)$ for generic $A \in \{D,S,R\}$, rather than for each estimator separately. We also report the coefficients defining the CSE, as well as the CSE itself, evaluated at $x=1$. As $K$ grows large, the estimand $\tau(\alpha=0.5,x=1)$ is close to $[\zeta_5(0.5)+\zeta_6(0.5)\bar{X}] + \zeta_6(0.5)(x-\bar{X}) = (0.5+0.4\cdot 0.75)+0.4(1-0.75)=0.9$. The number of Monte Carlo replications is $M=500$.
\begin{table}[htbp]
\setlength{\tabcolsep}{3pt} 
\centering
\caption{
Simulation results for the bias of $\hat{\beta}_{A,\cdot}(\alpha)$, $A \in \{D,S\}$, for the conditional outward spillover effect at $x=1$ in Example \ref{Examp_Ave_outward} under model \eqref{CSE_PO3} with directed regular cluster graphs. }
\begin{tabular}{c|ccc|ccc|ccc}
\toprule
$K$ & $\bar{\beta}_3(\alpha,1)$ & $\mathbb{E}[\hat{\beta}_{A,3}(\alpha)]$ & $\mathbb{E}[\hat{\beta}_{R,3}(\alpha)]$ & $\bar{\beta}_4(\alpha,1)$ & $\mathbb{E}[\hat{\beta}_{A,4}(\alpha)]$ & $\mathbb{E}[\hat{\beta}_{R,4}(\alpha)]$ & $\tau(\alpha,1)$ & $\mathbb{E}[\hat{\tau}_{A}(\alpha,1)]$ & $\mathbb{E}[\hat{\tau}_R(\alpha,1)]$ \\
\midrule
50 & 0.800 & 0.800 & 0.800 & 0.400 & 0.394 & 0.300 & 0.900 & 0.898 & 0.874 \\
 100 & 0.800 & 0.795 & 0.795 & 0.400 & 0.401 & 0.413 & 0.900 & 0.898 & 0.901 \\
 150 & 0.800 & 0.797 & 0.797 & 0.400 & 0.396 & 0.385 & 0.900 & 0.897 & 0.894 \\
 200 & 0.800 & 0.799 & 0.799 & 0.400 & 0.399 & 0.344 & 0.900 & 0.900 & 0.886 \\
 250 & 0.800 & 0.803 & 0.803 & 0.400 & 0.401 & 0.409 & 0.900 & 0.899 & 0.901 \\
 300 & 0.800 & 0.798 & 0.798 & 0.400 & 0.406 & 0.392 & 0.900 & 0.900 & 0.897 \\
 350 & 0.800 & 0.799 & 0.799 & 0.400 & 0.399 & 0.350 & 0.900 & 0.900 & 0.887 \\
 400 & 0.800 & 0.800 & 0.800 & 0.400 & 0.396 & 0.423 & 0.900 & 0.899 & 0.905 \\
 450 & 0.800 & 0.799 & 0.799 & 0.400 & 0.398 & 0.400 & 0.900 & 0.899 & 0.900 \\
 500 & 0.800 & 0.800 & 0.800 & 0.400 & 0.399 & 0.353 & 0.900 & 0.900 & 0.888 \\
\bottomrule
\end{tabular}
\begin{minipage}{0.95\linewidth}
\par\smallskip
\noindent $\bar{\beta}_{h}(\alpha,1)$ for $h \in \{3,4\}$ denotes the coefficients in the CSE as in Definition \ref{cond_spill_over_under_struc_model}. $\mathbb{E}[\hat{\beta}_{A,h}(\alpha)]$ denotes the average estimated coefficient across repetitions. $\tau(\alpha,1)$ denotes the CSE as defined in Definition \ref{cond_spill_over_under_struc_model}. $\mathbb{E}[\hat{\tau}_{A}(\alpha,1)]$ denotes the average estimated CSE as defined in Definitions \ref{Est_CSE_dyad}, \ref{Est_CSE_R}, and \ref{Est_CSE_S}.
\end{minipage}
\label{tab:CSE_outward_homo_regular_grap_bias}
\end{table}

\subsubsection{Simulation results}
From Table \ref{tab:CSE_outward_homo_regular_grap_bias}, we have $\hat{\beta}^p_{D,h}(\alpha) = \hat{\beta}^p_{S,h}(\alpha)$ (Proposition \ref{relation_est_CSE}), but this is not necessarily the case for $\hat{\beta}^p_{A,h}(\alpha)$ with $A \in \{D,S\}$ and $\hat{\beta}^p_{R,h}(\alpha)$ for $h \in \{3,4\}$. Furthermore, the bias of $\hat{\beta}_{A,h}(\alpha)$ is comparable across $A \in \{D,S,R\}$ for $h \in \{3,4\}$, respectively. Therefore, the biases of $\hat{\tau}_{A}(\alpha,1)$ for $A \in \{D,S,R\}$ are also comparable.

\begin{table}[htbp]
\setlength{\tabcolsep}{3pt}
\centering
\caption{Simulation results for the standard error and coverage of
$\hat{\beta}_{A,3}(\alpha)$ ($A \in \{D,S\}$) and $\hat{\beta}_{R,3}(\alpha)$
for the conditional outward spillover effect under model~\eqref{CSE_PO3} with
directed regular cluster graphs.
}
\begin{tabular}{c|ccc|ccc}
\toprule
$K$ & $se(\hat{\beta}_{A,3}(\alpha))$ & $\mathbb{E}[\hat{se}(\hat{\beta}_{A,3}(\alpha))]$  & 95\% coverage  & $se(\hat{\beta}_{R,3}(\alpha))$ & $\mathbb{E}[\hat{se}(\hat{\beta}_{R,3}(\alpha))]$  & 95\% coverage  \\
\midrule
  50 & 0.051 & 0.048 & 0.916 & 0.051 & 0.049 & 0.922 \\
 100 & 0.035 & 0.035 & 0.944 & 0.035 & 0.035 & 0.938 \\
 150 & 0.028 & 0.028 & 0.950 & 0.028 & 0.029 & 0.954 \\
 200 & 0.024 & 0.025 & 0.948 & 0.024 & 0.025 & 0.948 \\
 250 & 0.022 & 0.022 & 0.940 & 0.022 & 0.022 & 0.946 \\
 300 & 0.021 & 0.020 & 0.928 & 0.021 & 0.020 & 0.930 \\
 350 & 0.018 & 0.019 & 0.956 & 0.018 & 0.019 & 0.958 \\
 400 & 0.018 & 0.018 & 0.952 & 0.018 & 0.018 & 0.952 \\
 450 & 0.016 & 0.017 & 0.956 & 0.016 & 0.017 & 0.956 \\
 500 & 0.016 & 0.016 & 0.950 & 0.016 & 0.016 & 0.946 \\
\bottomrule
\end{tabular}
\label{tab:CSE_outward_homo_regular_grap_se_beta_3}
\begin{minipage}{0.95\linewidth}
\par\smallskip
\noindent $se(\hat{\beta}_{\cdot,3}(\alpha))$ denotes the empirical standard error of $\hat{\beta}_{\cdot,3}(\alpha)$, computed as the sample standard deviation across Monte Carlo replications. 
$\mathbb{E}[\hat{se}(\hat{\beta}_{\cdot,3}(\alpha))]$ denotes the Monte Carlo average of the estimated standard errors of $\hat{\beta}_{\cdot,3}(\alpha)$.
\end{minipage}
\end{table}

\begin{table}[htbp]
\centering
\caption{Simulation results for the standard error and coverage of
$\hat{\beta}_{A,4}(\alpha)$ ($A \in \{D, S\}$) and $\hat{\beta}_{R,4}(\alpha)$
for the conditional outward spillover effect under model \eqref{CSE_PO3} with directed regular cluster graphs.}
\begin{tabular}{c|ccc|ccc}
\toprule
$K$ & $se(\hat{\beta}_{A,4}(\alpha))$ & $\mathbb{E}[\hat{se}(\hat{\beta}_{A,4}(\alpha))]$  & 95\% coverage  & $se(\hat{\beta}_{R,4}(\alpha))$ & $\mathbb{E}[\hat{se}(\hat{\beta}_{R,4}(\alpha))]$  & 95\% coverage  \\
\midrule
  50 & 0.159 & 0.152 & 0.932 & 0.571 & 1.641 & 0.928 \\
 100 & 0.114 & 0.111 & 0.944 & 0.370 & 1.186 & 0.962 \\
 150 & 0.097 & 0.091 & 0.930 & 0.309 & 0.964 & 0.956 \\
 200 & 0.080 & 0.079 & 0.944 & 0.274 & 0.834 & 0.954 \\
 250 & 0.072 & 0.072 & 0.932 & 0.248 & 0.760 & 0.956 \\
 300 & 0.067 & 0.065 & 0.944 & 0.216 & 0.678 & 0.954 \\
 350 & 0.058 & 0.060 & 0.964 & 0.195 & 0.641 & 0.950 \\
 400 & 0.057 & 0.056 & 0.942 & 0.191 & 0.589 & 0.956 \\
 450 & 0.054 & 0.053 & 0.954 & 0.185 & 0.556 & 0.950 \\
 500 & 0.050 & 0.050 & 0.960 & 0.170 & 0.524 & 0.944 \\
\bottomrule
\end{tabular}
\label{tab:CSE_outward_homo_regular_grap_se_beta_4}
\begin{minipage}{0.95\linewidth}
\par\smallskip
\noindent $se(\hat{\beta}_{\cdot,4}(\alpha))$ denotes the empirical standard error of $\hat{\beta}_{\cdot,4}(\alpha)$, computed as the sample standard deviation across Monte Carlo replications. $\mathbb{E}[\hat{se}(\hat{\beta}_{\cdot,4}(\alpha))]$ denotes the Monte Carlo average of the estimated standard errors of $\hat{\beta}_{\cdot,4}(\alpha)$. 
\end{minipage}
\end{table}

\begin{table}[htbp]
\centering
\caption{Simulation results for the standard error and coverage of
$\hat{\tau}_{A}(\alpha,1)$ ($A \in \{D,S\}$) and $\hat{\tau}_{R}(\alpha,1)$
for the conditional outward spillover effect under model \eqref{CSE_PO3}
with directed regular cluster graphs.
}
\begin{tabular}{c|ccc|ccc}
\toprule
$K$ & $se(\hat{\tau}_{A}(\alpha,1))$ & $\mathbb{E}[\hat{se}(\hat{\tau}_{A}(\alpha,1))]$  & 95\% coverage  & $se(\hat{\tau}_{R}(\alpha,1))$ & $\mathbb{E}[\hat{se}(\hat{\tau}_{R}(\alpha,1))]$  & 95\% coverage  \\
\midrule
  50 & 0.065 & 0.063 & 0.930 & 0.438 & 0.418 & 0.930 \\
 100 & 0.048 & 0.047 & 0.940 & 0.305 & 0.313 & 0.962 \\
 150 & 0.038 & 0.037 & 0.938 & 0.244 & 0.248 & 0.954 \\
 200 & 0.032 & 0.033 & 0.954 & 0.207 & 0.216 & 0.956 \\
 250 & 0.028 & 0.029 & 0.966 & 0.187 & 0.188 & 0.948 \\
 300 & 0.028 & 0.027 & 0.918 & 0.169 & 0.175 & 0.962 \\
 350 & 0.025 & 0.025 & 0.960 & 0.157 & 0.166 & 0.948 \\
 400 & 0.024 & 0.023 & 0.950 & 0.153 & 0.151 & 0.952 \\
 450 & 0.022 & 0.022 & 0.946 & 0.149 & 0.144 & 0.950 \\
 500 & 0.021 & 0.021 & 0.942 & 0.134 & 0.135 & 0.950 \\
\bottomrule
\end{tabular}
\label{tab:CSE_outward_homo_regular_grap_se_tau}
\begin{minipage}{0.95\linewidth}
\par\smallskip
\noindent $se(\hat{\tau}_{\cdot}(\alpha,1))$ denotes the empirical standard error of $\hat{\tau}_{\cdot}(\alpha,1)$, computed as the sample standard deviation across Monte Carlo replications. 
$\mathbb{E}[\hat{se}(\hat{\tau}_{\cdot}(\alpha,1))]$ denotes the Monte Carlo average of the estimated standard errors of $\hat{\tau}_{\cdot}(\alpha,1)$.
\end{minipage}
\end{table}

From Tables \ref{tab:CSE_outward_homo_regular_grap_se_beta_3} and \ref{tab:CSE_outward_homo_regular_grap_se_beta_4}, the standard errors of $\hat{\beta}_{A,3}(\alpha)$ for $A \in \{D,S,R\}$ are smaller than those of $\hat{\beta}_{A,4}(\alpha)$ for $A \in \{D,S,R\}$. This difference arises because $\hat{\beta}_{A,4}(\alpha)$ incorporates the variation of the covariate. Moreover, while the standard errors of $\hat{\beta}_{A,3}(\alpha)$ for $A \in \{D,S\}$ are comparable to those of $\hat{\beta}_{R,3}(\alpha)$, the standard errors of $\hat{\beta}_{A,4}(\alpha)$ for $A \in \{D,S\}$ are substantially smaller than those of $\hat{\beta}_{R,4}(\alpha)$. The reason is that the variation of the terms $\sum_{jk \neq ik} B_{ik,jk} \tilde{x}_{jk} 1\{Z_{jk}=z\}$ in $\hat{\beta}^r_{A,4}(\alpha)$ for $A \in \{D,S\}$ is generally smaller than that of $(\sum_{jk \neq ik} B_{ik,jk}1\{Z_{jk}=z\})\, x^\dagger_{ik}$ for $A \in \{R\}$ and $z \in \{0,1\}$, where the latter includes additional interaction terms. 
In addition, the coverage of $\hat{\beta}_{A,h}(\alpha)$ for $h \in \{3,4\}$ and $A \in \{D,S,R\}$, as well as that of $\hat{\tau}_{A}(\alpha,1)$ for $A \in \{D,S,R\}$, tends to reach the nominal level when $K$ is large, under the data-generating process \ref{CSE_PO3}, and the estimand is the conditional outward spillover effect. This behavior, which appears closer to nominal rather than conservative, is likely due to the relatively simple data-generating process \eqref{CSE_PO3}, where the coefficients are homogeneous; see Theorem 4 in \citet{abadie2020sampling} for further discussion.

In Appendix \ref{Appendix:simulation_empirical}, we also report simulation results corresponding to Setting $2$ in Lemma \ref{settings_all_conditions_satisfy}, where the network is an Erd\H{o}s–R\'{e}nyi directed graph with connection probability $4/n_k$, potential outcomes are simulated under Model \eqref{CSE_PO3} with homogeneous coefficients, and the estimand is the conditional inward spillover effect defined in Example \ref{Examp_cond_inward}. Results are reported in Tables \ref{tab:CSE_inward_homo_ER_grap_bias}, \ref{tab:CSE_inward_homo_regular_grap_se_beta_3}, \ref{tab:CSE_inward_homo_regular_grap_se_beta_4}, and \ref{tab:CSE_inward_regular_directed_ER_grap_se_beta_4} in Appendix \ref{Appendix:simulation_empirical}. The conclusions are qualitatively similar to those obtained for the conditional outward spillover effect.

\section{Real data application}
\label{sec:Real Data Application}

In this section, we apply our proposed estimators to data from a randomized trial designed to assess whether intensive information sessions can increase the uptake of weather insurance among farmers in rural China \citet{cai2015social}. In the experiment, farmers were randomized to receive a simple or an intensive information session in one of two rounds. Here, we define the treatment as participation in an intensive information session during the first round. The network structure is obtained from a pre-experiment survey in which each farmer lists friends with whom they discuss agricultural or financial matters. Based on these friendship nominations, we construct a directed network, where the presence of a link $e_{ik,jk}$ from farmer $i$ to farmer $j$ in village $k$ indicates that farmer $i$ nominated $j$ as a friend. 


Our goal is to estimate the average spillover effect of one farmer receiving the first-round intensive information session on the insurance uptake of a subset of other farmers—such as in-neighbors (i.e., those who nominated the farmer) or all farmers residing in the same village—and to assess the heterogeneity of such effects with respect to farmers’ characteristics.
That is, we focus on the following estimands: average and conditional outward spillover effects on in-neighbors (Examples \ref{Examp_Ave_outward} and \ref{Examp_cond_outward}), average and conditional inward spillover effects from out-neighbors (Examples \ref{Examp_Ave_inward} and \ref{Examp_cond_inward}), as well as average and conditional pairwise spillover effects among farmers in the same village (Examples \ref{Examp_ave_ind_effect} and \ref{Examp_cond_ind_effect}).
In both outward and inward spillover effects, we are interested in the influence on a farmer’s insurance uptake from those he or she considers friends. For this reason, we consider outward spillover effects on in-neighbors and inward spillover effects from out-neighbors. 
The hypothetical treatment assignment probability $\alpha$ is set equal to the realized treatment probability $\beta$, under an i.i.d. Bernoulli design with $\beta = 0.22$, corresponding to the proportion of individuals invited to the first-round intensive session.

We report the point estimates and corresponding 95\% confidence intervals constructed using the variance estimators introduced in Proposition \ref{conserve_variance_est}. The results for the outward spillover effects are summarized in Table \ref{real_data:outward_spillover}, and the results for the inward spillover effects and the pairwise spillover effects are reported in Appendix \ref{Appendix:simulation_empirical}.

\begin{table}[htbp]
\centering
\setlength{\tabcolsep}{5pt}
\caption{Outward spillover effects, as in Examples \ref{Examp_Ave_outward} and \ref{Examp_cond_outward}: estimates and 95\% confidence intervals}
\begin{tabular}{l|cc|cc}
\toprule
& \multicolumn{2}{c|}{$\hat{\tau}_A(\alpha,x)$ for $A \in \{D,S\}$} & \multicolumn{2}{c}{$\hat{\tau}_R(\alpha,x)$} \\
Group & estimate & $95\%$ CI & estimate & $95\%$ CI \\
\midrule
all (ASE) &  $0.022$ &	$[-0.004,0.048]$& $0.022$ &	$[-0.004,	0.048]$ \\
female & $-0.005$ & $[-0.104,\,0.095]$ & $-0.042$ & $[-0.107,\,0.023]$ \\
male & $0.026$ & $[0.000,\,0.053]$ & $0.029$ & $[0.000,\,0.059]$ \\
risk averse $=0$ & $0.011$ & $[-0.020,\,0.043]$ & $0.013$ & $[-0.023,\,0.049]$ \\
risk averse $>0$ & $0.043$ & $[0.003,\,0.083]$ & $0.044$ & $[-0.030,\,0.118]$ \\
insurance repay$=0$ & $0.040$ & $[0.011,\,0.070]$ & $0.018$ & $[-0.014,\,0.049]$ \\
insurance repay$=1$ & $-0.007$ & $[-0.054,\,0.040]$ & $0.030$ & $[-0.030,\,0.089]$ \\
general trust$=0$ & $0.008$ & $[-0.065,\,0.081]$ & $0.013$ & $[-0.048,\,0.074]$ \\
general trust$=1$ & $0.025$ & $[-0.004,\,0.054]$ & $0.024$ & $[-0.008,\,0.055]$ \\
in-degree $<4$ & $0.038$ & $[0.003,\,0.073]$ & $0.110$ & $[0.018,\,0.202]$ \\
in-degree $\ge 4$ & $0.003$ & $[-0.028,\,0.034]$ & $-0.078$ & $[-0.167,\,0.011]$ \\
out-degree $<4$ & $0.012$ & $[-0.028,\,0.051]$ & $0.001$ & $[-0.044,\,0.045]$ \\
out-degree $\ge 4$ & $0.031$ & $[-0.003,\,0.064]$ & $0.038$ & $[-0.005,\,0.082]$ \\
disaster=no & $0.031$ & $[-0.013,\,0.075]$ & $0.018$ & $[-0.018,\,0.055]$ \\
disaster=yes & $0.018$ & $[-0.013,\,0.049]$ & $0.024$ & $[-0.018,\,0.066]$ \\
literacy=no & $0.012$ & $[-0.043,\,0.068]$ & $-0.026$ & $[-0.075,\,0.023]$ \\
literacy=yes & $0.019$ & $[-0.007,\,0.045]$ & $0.029$ & $[-0.003,\,0.062]$ \\
\bottomrule
\end{tabular}
\label{real_data:outward_spillover} \\
\begin{minipage}{0.95\linewidth}
\par\smallskip
\noindent “Risk averse” = 0 denotes households that are less risk averse than those with risk aversion > 0. “Insurance repay” indicates whether the respondent has previously received payouts from other insurance products (1 = yes, 0 = no). “General trust” measures trust in the government, with larger values indicating greater trust. The degree cutoff of 4 corresponds to the 60th percentile of the out-degree and in-degree distributions. “Disaster” indicates whether any disaster occurred in the prior year, and “literacy” denotes whether the household head or respondent is literate.
\end{minipage}
\end{table}

In Table~\ref{real_data:outward_spillover}, where the estimands are outward spillover effects, the first row reports the ASE. This quantity represents the average influential effect of a rice farmers exposed to the intensive information session on their in-neighbors’ decisions to purchase weather insurance. All three estimators indicate that, on average, exposure increases neighboring farmers’ purchasing probability by approximately $2.2\%$, although the effect is not statistically significant.

For the CSE, the estimand captures the influential effect of a household of rice farmers with covariates in a given group on its neighbors’ purchasing behavior. In this application, the directed village networks are not regular, violating one of the conditions required for Setting~1 in Lemma~\ref{settings_all_conditions_satisfy}. Consequently, the intermediate quantities 
$\boldsymbol{\beta}^l_{A}(\alpha)$ for $A \in \{D,S\}$ may not coincide with 
$\boldsymbol{\beta}^l_{R}(\alpha)$ for $l \in \{r,p\}$. 
This implies that the corresponding estimates $\hat{\tau}_{A}(\alpha,x)$ for 
$A \in \{D,S\}$ may differ from $\hat{\tau}_{R}(\alpha,x)$, as observed in Table~\ref{real_data:outward_spillover}, reflecting the fact that these estimators are consistent for distinct intermediate quantities.

The estimated conditional outward spillover effects suggest that information transmitted from (i) male farmers or (ii) farmers with fewer in-neighbors tends to increase their uptake of weather insurance. However, the differences in conditional spillover effects across the range of values for a given covariate are not statistically significant. These findings are consistent with behavioral and network-based explanations for information diffusion. Male farmers appear particularly influential in disseminating such information, possibly due to their broader social roles or stronger positions within local communication networks. Individuals with fewer in-neighbors may devote more attention to those who consider them friends, thereby exerting stronger influence on their peers. In addition, the estimators for the CSE at “in-degree $\geq$ 4” suggest—although not significantly—that farmers with more out-neighbors tend to increase in-neighbors’ uptake of weather insurance. One possible explanation is that individuals with many out-neighbors—who nominate more friends—may be those who pay closer attention to others and actively seek information. In turn, they may be more likely to disseminate information further, thereby promoting broader insurance adoption.

Table \ref{real_data:inward_spillover} in Appendix \ref{Appendix:simulation_empirical} reports the results of estimating the average and conditional inward spillover effects, corresponding to the estimands defined in Examples \ref{Examp_Ave_inward} and \ref{Examp_cond_inward}. In this case, positive inward spillover effects are observed primarily among individuals whose friends have relatively few connections (i.e., low in-degree), which is consistent with the results reported in Table~\ref{real_data:outward_spillover} for the group with low in-degree.


Table \ref{tab:pairwise_spillover} in Appendix \ref{Appendix:simulation_empirical} reports the estimation results for the average and conditional pairwise spillover effects, which correspond to the estimands defined in Examples \ref{Examp_ave_ind_effect} and \ref{Examp_cond_ind_effect}. The key empirical findings from $\hat{\tau}_{A}(\alpha, x)$ for $A \in \{D, S\}$ are as follows. Farmers who exhibit high general trust in the government, or farmers who are not literate, exert a statistically positive spillover effect on weather insurance uptake among other farmers within the same village. Conversely, farmers who reported no loss experience from the previous year's disaster exert a statistically negative spillover effect on the insurance uptake of their peers. The positive spillover effect associated with general trust in the government may stem from increased confidence in the contractual guarantees of the weather insurance. A farmer with higher trust is likely more confident that the government-backed scheme will reliably provide reimbursement for crop losses, thus reducing the perceived risk and encouraging uptake by network members.

The positive spillover from illiterate farmers might be due to their greater propensity to engage in oral communication and rely on social networks. This method of information sharing may facilitate the quicker and broader dissemination of information about the insurance product to other farmers. Another possible explanation is that illiterate farmers may require more persuasion and therefore benefit more from intensive training; that is, when receiving intensive training, they are more likely to purchase the weather insurance and, in turn, influence others to do the same.

The negative spillover effect exerted by farmers with no prior crop loss may be driven by changes in their subjective risk assessments. Such farmers may attribute their past absence of loss to effective prevention strategies, better preparation, or inherently lower exposure, leading them to believe that their risk of experiencing a loss in the current year is low. When exposed to the intensive information sessions, they may therefore perceive the insurance product as less necessary than advertised. This reduced perceived need for insurance can be socially transmitted, thereby exerting a negative influence on the adoption decisions of other farmers in the same village.

Conversely, the estimator $\hat{\tau}_{R}(\alpha,x)$ is statistically insignificant across all examined groups and exhibits a small magnitude of conditional pairwise spillover effects. The observed lack of significance may stem from how the receiver estimator and pairwise spillover estimand weights are specified. These estimand weights make the receiver estimator $\hat{\tau}_{R}(\alpha,x)$ collapse the covariates of all potential senders (i.e., all other farmers in the same village) into a single aggregated covariate for the focal individual. This aggregation process inherently captures the spillover effect at a more average level within the whole village. Consequently, $\hat{\tau}_{R}(\alpha,x)$ may be less sensitive than $\hat{\tau}_{A}(\alpha,x)$ for $A\in\{D,S\}$ in identifying specific conditional pairwise spillover effects.

\section{Conclusion and discussion}
In this paper, we introduce a general framework for representing spillover effect estimands from one unit's treatment on the outcomes of a subset of units, while a hypothetical treatment assignment is applied to other units, encompassing both average- and conditional-type quantities. By modifying the estimand weights, this framework flexibly generates different estimands of interest. Under this unified formulation, the corresponding estimators can be constructed uniformly across estimands by appropriately substituting the relevant plug-in estimand weights.

We develop WLS estimators with design-based inference under partial interference for randomized experiments, without imposing an exposure mapping function or strong assumptions on the functional forms of the dyadic average potential outcomes. In contrast to much of the partial interference literature, which treats cluster sizes as fixed, we consider an asymptotic regime in which both the number of clusters and the cluster sizes diverge. In particular, we develop three estimators that correspond to distinct perspectives: the \emph{dyadic}, \emph{effect-sender}, and \emph{effect-receiver} perspectives. The dyadic perspective views spillovers through pairwise relations, regressing each unit’s outcome on the treatment of another unit in the interference set. The effect-sender perspective instead focuses on how a unit’s treatment affects the aggregated outcomes of others, whereas the effect-receiver perspective regresses each unit’s outcome on an aggregated measure of the treatments it receives, constructed according to the estimator weights.

Certain perspectives tend to align naturally with specific types of estimands. For instance, the average outward spillover effect aligns more closely with the effect-sender perspective, whereas the average inward spillover effect aligns with the effect-receiver perspective. Nonetheless, all three estimators can be constructed to target the estimand of interest using the estimand weights.

There is a close relationship among the three estimators. As shown in Theorem \ref{equiv_hj_D_R_S}, they are equivalent and coincide with the nonparametric H\'{a}jek estimator for the ASE. Hence, they provide invariant constructions of this ASE estimator. For the CSE, these estimators can be extended in a straightforward manner by modifying the design matrix in the dyadic and effect-sender perspectives, and the weighting matrix in the effect-receiver perspective, to incorporate conditioning covariates and their interaction terms with treatments. Under the flexible model for dyadic average potential outcomes specified in Definition~\ref{struc_APO}, the estimated coefficients on the treatment and on the interaction between the covariate and treatment capture the conditional spillover effect.

In contrast to nonparametric approaches for estimating conditional spillover effects under interference, such as the kernel-smoothing method of \citet{bong2024heterogeneous}, which relies on local interference assumptions and requires large samples within bandwidths, our approach strikes a balance between robustness and efficiency. The WLS estimators we developed can accommodate a rich class of functional forms while leveraging the entire population to improve efficiency. In particular, alternative specifications of the design matrix are feasible, permitting more flexible (e.g., semi-parametric) transformations of $\tilde{X}_{jk}$ in the dyadic and sender estimators and of $X^\dagger_{ik}$ in the receiver estimator. In such cases, the conditions required to establish the correspondence between these estimators and the CSE (Section~\ref{connect_estimable_quant_CSE}) may be weakened.

A set of conditions is required for the consistent estimation of the three estimators for the CSE. To establish consistency of our estimators for  $\boldsymbol{\beta}^r(\alpha)$, the ratio of the expected numerator and denominator of the WLS estimators, it suffices to impose mild conditions, including Assumptions \ref{part_intf}--\ref{bounded_x}. Moving from $\boldsymbol{\beta}^r(\alpha)$ to $\boldsymbol{\beta}^p(\alpha)$, the population-weighted average of the structural coefficients in Definition~\ref{struc_APO}, further requires (i) a certain degree of homogeneity in the coefficients and (ii) approximate equivalence between group-specific and population averages within each coefficient value, as stated in Assumption \ref{coef_hete_restrict}. To connect $\boldsymbol{\beta}^p(\alpha)$ to the final estimand $\boldsymbol{\bar{\beta}}(\alpha,x)$, the coefficients need to be independent of the conditioning variable $X_{ik}$, implying that the assumed functional form adequately captures the heterogeneity in spillover effects. In particular, when the coefficients of the dyadic average potential outcome model are homogeneous, Statement 1 of Assumption \ref{coef_hete_restrict} holds automatically, and Statement 2 of Assumption \ref{coef_hete_restrict} is also satisfied, provided that part (a) holds, because the weights for both $\boldsymbol{\beta}^p(\alpha)$ and $\boldsymbol{\bar{\beta}}(\alpha,x)$ sum to one, as defined in Definition \ref{cond_spill_over_under_struc_model}. To sum up, when the coefficients are homogeneous and condition (a) of Assumption \ref{coef_hete_restrict}(2) is satisfied, all three estimators consistently estimate the CSE.

For design-based inference, we establish a unified framework for deriving the central limit theorem (CLT) and constructing asymptotically conservative variance estimators under partial interference. These results extend the sampling-based regression framework of \citet{abadie2020sampling} to settings with partial interference. 
Two observations are worth noting. First, the framework is readily extendable in a natural way to accommodate alternative interference structures by formulating the estimators with respect to the corresponding \emph{interference set}—that is, the set of units whose treatments affect a given unit’s outcome—induced by the assumed interference mechanism. The CLT continues to hold, provided that the growth rate of each node’s degree in the dependence graph, induced by the interference sets, satisfies the rate condition stated in Theorem~\ref{CLT_ASE_est}.

Second, under partial interference, the dependence graph among units is fully connected within clusters, and this symmetric structure ensures that the variance estimator constructed in Proposition \ref{conserve_variance_est} remains asymptotically conservative without imposing additional assumptions on the covariance structure. When more general interference structures arise—where the dependence graph is not symmetric across units—further adjustments to the variance estimator are required. The rationale and corresponding adjustment procedures are discussed in \citet{leung2022causal} and \citet{gao2025causalinferencenetworkexperiments}.

For practical recommendations, we summarize the relative merits of the proposed estimators as follows. For the ASE, since the three estimators are algebraically equivalent, their differences arise primarily from implementation considerations. Because both $\hat{\tau}_{S}(\alpha)$ and $\hat{\tau}_{R}(\alpha)$ involve lower-dimensional matrix computations than $\hat{\tau}_{D}(\alpha)$, they are generally preferable from a computational standpoint. For the CSE, the choice among estimators depends on several factors, including efficiency and the plausibility of the assumptions required for consistency and asymptotic normality of the estimators to the CSE. In terms of efficiency, simulations for different estimands (Tables \ref{tab:CSE_outward_homo_regular_grap_se_tau}, \ref{tab:CSE_outward_homo_regular_grap_se_beta_4}, and \ref{tab:CSE_inward_homo_regular_grap_se_beta_4} and \ref{tab:CSE_inward_regular_directed_ER_grap_se_beta_4} in Appendix \ref{Appendix:simulation_empirical}) show that $\hat{\tau}_{D}(\alpha, x)$ and $\hat{\tau}_{S}(\alpha, x)$ tend to be more efficient than $\hat{\tau}_{R}(\alpha, x)$ in the scenarios of Lemma~\ref{settings_all_conditions_satisfy}, where both $\boldsymbol{\beta}^r_{A,(3,4)}(\alpha)$ and $({\beta}_{A,3}^p(\alpha), {\beta}_{A,4}^p(\alpha))^T$ coincide across $A \in \{D, S, R\}$.  
With respect to assumption strength, $\hat{\tau}_{D}(\alpha, x)$ and $\hat{\tau}_{S}(\alpha, x)$ rely on weaker conditions than $\hat{\tau}_{R}(\alpha, x)$ for consistency of the intermediate quantities $\boldsymbol{\beta}^r_{A,(3,4)}(\alpha)$ and $({\beta}_{A,3}^p(\alpha),{\beta}_{A,4}^p(\alpha))^T$.  
From an intuitive standpoint, $\hat{\tau}_{D}(\alpha, x)$ provides the most direct formulation for all estimands that are averages of pairwise spillover effects, as it considers the influence of one unit’s treatment on another unit’s outcome. For specific estimands, the conditional outward spillover effect aligns more naturally with $\hat{\tau}_{S}(\alpha, x)$, whereas the conditional inward spillover effect corresponds more closely to $\hat{\tau}_{R}(\alpha, x)$. Finally, regarding computational convenience, $\hat{\tau}_{S}(\alpha, x)$ and $\hat{\tau}_{R}(\alpha, x)$ again offer advantages similar to those observed in the ASE case, as they entail less computational burden.

Several avenues for future research emerge from this framework. First, it would be valuable to investigate weaker modeling assumptions on the WLS estimators for CSE, such as more flexible specifications for the effect of the sender’s treatment and the conditioning covariates. while retaining three desirable properties: (i) a tractable regression-based formulation; (ii) the ability to leverage most units in the population to improve efficiency; and (iii) consistency for the CSE. Second, a promising direction is to design optimal treatment rules guided by the estimated conditional spillover effects.
 

\begin{appendix}
\section{Estimators for the average spillover effect}
\label{Appendix_est_ase}
Throughout the section, we use $a\circ b$ to denote the summation of elementwise product of vector of $a$ and $b$.  

\subsection{Equivalence of the estimators for the average spillover effect}
\label{Appendix_est_ase:equivalent_est}
\begin{proof}[Proof of Theorem \ref{equiv_hj_D_R_S}]
Define
\begin{small}
\begin{equation}
\label{eq:matrix_A_F}
A :=
\begin{bmatrix}
\displaystyle \sum_{k=1}^K \sum_{ik=1}^{n_k} \sum_{jk \neq ik} B_{ik,jk} &
\displaystyle \sum_{k=1}^K \sum_{ik=1}^{n_k} \sum_{jk \neq ik} B_{ik,jk} Z_{jk} \\[6pt]
\displaystyle \sum_{k=1}^K \sum_{ik=1}^{n_k} \sum_{jk \neq ik} B_{ik,jk} Z_{jk} &
\displaystyle \sum_{k=1}^K \sum_{ik=1}^{n_k} \sum_{jk \neq ik} B_{ik,jk} Z_{jk}
\end{bmatrix} \quad F :=
\begin{bmatrix}
\displaystyle \sum_{k=1}^K \sum_{ik=1}^{n_k} \sum_{jk \neq ik} B_{ik,jk} Y_{ik} \\[6pt]
\displaystyle \sum_{k=1}^K \sum_{ik=1}^{n_k} \sum_{jk \neq ik} B_{ik,jk} Z_{jk} Y_{ik}
\end{bmatrix}
\end{equation}
\end{small}
From Definition \ref{Est_ASE_dyad}, we have 
\begin{small}
\begin{equation*}
    \begin{split}
      & V^T_D B_D V_D \\
      &= \begin{bmatrix}
         \sum_{k=1}^K \sum_{ik=1}^{n_k} \mathbf{B}^T_{ik,-ik} \mathbf{1}_{k,-ik} & \sum_{k=1}^K \sum_{ik=1}^{n_k} \mathbf{B}^T_{ik,-ik} \mathbf{Z}_{k,-ik} \\
          \sum_{k=1}^K \sum_{ik=1}^{n_k} (\mathbf{Z}_{k,-ik} \circ \mathbf{B}_{ik,-ik} )^T \mathbf{1}_{k,-ik} & \sum_{k=1}^K \sum_{ik=1}^{n_k} (\mathbf{Z}_{k,-ik} \circ \mathbf{B}_{ik,-ik} )^T \mathbf{Z}_{k,-ik} \\
      \end{bmatrix} =A. 
    \end{split}
\end{equation*}
\end{small}
Meanwhile, 
\begin{small}
\begin{equation*}
    \begin{split}
      V^T_D B_D Y_D= \begin{bmatrix}
         \sum_{k=1}^K \sum_{ik=1}^{n_k} \mathbf{B}^T_{ik,-ik} (Y_{ik} \mathbf{1}_{k,-ik}) \\
          \sum_{k=1}^K \sum_{ik=1}^{n_k} (\mathbf{Z}_{k,-ik} \circ \mathbf{B}_{ik,-ik} )^T (Y_{ik}\mathbf{1}_{k,-ik})\\
      \end{bmatrix}.
    \end{split}
\end{equation*}
\end{small}
Based on Definition \ref{est_ASE_R}, we have 
\begin{equation*}
    \begin{split}
      & V^T_R B_R V_R = \begin{bmatrix}
         \mathbf{1}^T_{N}(\mathbf{B}^1+\mathbf{B}^{0}) \mathbf{1}_{N} \quad & \mathbf{1}^T_{N}\mathbf{B}^1 \mathbf{1}_{N} \\
          \mathbf{1}^T_{N}\mathbf{B}^1 \mathbf{1}_{N} &\mathbf{1}^T_{N}\mathbf{B}^1 \mathbf{1}_{N} \\
      \end{bmatrix} =A.   
    \end{split}
\end{equation*}
and 
\begin{equation}
\label{ase_est_numerator_R}
    \begin{split}
      & V^T_R B_R Y_R = \begin{bmatrix}
         \mathbf{1}^T_{N}(\mathbf{B}^1+\mathbf{B}^{0}) \mathbf{Y}  \\
    \mathbf{1}^T_{N}\mathbf{B}^1 \mathbf{Y} \\
      \end{bmatrix} =F.   
    \end{split}
\end{equation}
From Definition \ref{est_ASE_S}, we obtain 
\begin{small}
\begin{equation*}
    \begin{split}
      & V^T_S B_S V_S = \begin{bmatrix}
         \sum_{k=1}^K \sum_{jk=1}^{n_k} \tilde{S}_{jk}S_{jk} W_{jk}(\mathbf{Z}_k) \quad &  \sum_{k=1}^K \sum_{jk=1}^{n_k} \tilde{S}_{jk}S_{jk} W_{jk}(\mathbf{Z}_k) Z_{jk} \\
          \sum_{k=1}^K \sum_{jk=1}^{n_k} \tilde{S}_{jk}S_{jk} W_{jk}(\mathbf{Z}_k) Z_{jk} & \sum_{k=1}^K \sum_{jk=1}^{n_k} \tilde{S}_{jk}S_{jk} W_{jk}(\mathbf{Z}_k) Z_{jk} \\
      \end{bmatrix} =A. 
    \end{split}
\end{equation*}
\end{small}
Moreover, 
\begin{small}
\begin{equation*}
    \begin{split}
      V^T_S B_S Y_S= \begin{bmatrix}
          \sum_{k=1}^K \sum_{jk=1}^{n_k} \tilde{S}_{jk}S_{jk} W_{jk}(\mathbf{Z}_k) \sum_{ik \neq jk} \frac{S_{ik|jk}}{\tilde{S}_{jk}} Y_{ik} \\
          \sum_{k=1}^K \sum_{jk=1}^{n_k} \tilde{S}_{jk}S_{jk} W_{jk}(\mathbf{Z}_k) Z_{jk}  \sum_{ik \neq jk} \frac{S_{ik|jk}}{\tilde{S}_{jk}} Y_{ik}\\
      \end{bmatrix} =_{(1)}F.
    \end{split}
\end{equation*}
\end{small}
(1) follows from the decomposition of the estimand weights. Combining previous steps, we have 
$$(V^T_D B_D V_D)^{-1}V^T_D B_D Y_D=(V^T_R B_R V_R)^{-1}V^T_R B_R Y_R=(V^T_S B_S V_S)^{-1}V^T_S B_S Y_S=A^{-1}F,$$
which implies $\hat{\tau}_D(\alpha)=\hat{\tau}_{R}(\alpha)=\hat{\tau}_S(\alpha)$. Finally, direct calculation yields
\begin{equation*}
    \begin{split}
      &  A^{-1}F= \begin{pmatrix}
         \frac{\sum_{k=1}^K \sum_{ik=1}^{n_k} \sum_{jk \neq ik} B_{ik,jk}(1-Z_{ik}) Y_{ik}}{\sum_{k=1}^K \sum_{ik=1}^{n_k} \sum_{jk \neq ik} B_{ik,jk}(1-Z_{ik})} , \hat{\tau}_{hj}(\alpha)
      \end{pmatrix}^T, 
    \end{split}
\end{equation*}
which establishes that $\hat{\tau}_D(\alpha)=\hat{\tau}_{R}(\alpha)=\hat{\tau}_S(\alpha)=\hat{\tau}_{hj}(\alpha)$.
\end{proof}

\subsection{Inference for the estimators of the average spillover effect}
\label{Appendix:inf_ASE}
Prior to proving Proposition \ref{consist_ASE}, We first establish an intermediate result, which characterizes a concentration inequality for sums of sub-Gaussian variables under cluster-dependent graphs with possibly growing cluster sizes.
\begin{lemma}
\label{conc_clustern_graph}
Let $\{R_{hk}\}_{hk \in \{1,\dots,n^e_k\},\, k \in \{1,\dots,K\}}$ be a collection of sub-Gaussian random variables, where $hk$ denotes unit $h$ in cluster $k$. Let $\bar{n}^e_k := \max_{k \in \{1,\dots,K\}} n^e_k$ and $\bar{\sigma}^2 := \max_{hk,\,k} \sigma_{hk}^2$ with $\sigma_{hk}^2$ denoting the sub-Gaussian parameter of $R_{hk}$. Suppose the dependence graph $\mathcal{A}$ (Definition \ref{def:dependence_graph}) has the following structure: within each cluster $\{R_{hk}\}_{h=1}^{n^e_k}$ forms a complete graph, and across clusters there are no edges. Then, with probability at least $1-\delta$,
$$
\Bigg\lvert \sum_{k=1}^K \sum_{h=1}^{n^e_k} R_{hk} 
- \mathbb{E} \left(\sum_{k=1}^K \sum_{h=1}^{n^e_k} R_{hk}\right) \Bigg\rvert 
\;\leq\; \sqrt{2 \bar{\sigma}^2 \, \bar{n}^e_k \Big(\sum_{k=1}^K n^e_k\Big)\, \log \!\left(\tfrac{2 \bar{n}^e_k}{\delta}\right)}.
$$
\end{lemma}
\begin{proof}[Proof of Lemma \ref{conc_clustern_graph}]
Given the structure of $\mathcal{A}$, we construct a cover (Definition \ref{def:cover}) as follows. For each cluster $k$, list its units as a column vector $(1k, \dots, n^e_k k)^T$. Aligning these vectors side by side produces $K$ columns. Next, consider each row of these $K$ columns. For each row $a \in \{1,\dots,M\}$, with $M \leq \bar{n}^e_k$, define $\mathcal{C}_a$ as the set of elements in row $a$. This yields a partition $
\mathcal{C}(\mathcal{A}) = \{\mathcal{C}_1, \dots, \mathcal{C}_M\}
$ of the vertex set induced by the dependence graph $\mathcal{A}$. By construction, the sets $\mathcal{C}_s$ are non-overlapping across $s$ and $|\mathcal{C}(\mathcal{A})|\leq \bar{n}^e_k$. Hence, with probability at least $1-\delta$,
\begin{small}
\begin{equation*}
        \begin{split}
          &   \left\lvert\sum_{k=1}^K \sum_{hk=1}^{n^e_k} {R}_{hk}-\mathbb{E}(\sum_{hk=1}^{n^e_k} {R}_{hk})\right\rvert \leq \sum_{s=1}^{|\mathcal{C}(\mathcal{A})|} \left\lvert \sum_{hk \in \mathcal{C}_s}  R_{hk} -\mathbb{E}(\sum_{hk \in \mathcal{C}_s}  R_{hk}) \right\rvert\\
          & \leq_{(1)} \sum_{s=1}^{|\mathcal{C}(\mathcal{A})|} \sqrt{2\bar{\sigma}^2 |\mathcal{C}_s|\log({2 |\mathcal{C}(\mathcal{A})|}/{\delta})} \leq_{(2)} |\mathcal{C}(\mathcal{A})|  \sqrt{ (1/|\mathcal{C}(\mathcal{A})|)2\bar{\sigma}^2 (\sum_{s=1}^{|\mathcal{C}(\mathcal{A})|}|\mathcal{C}_s|)\log({2 |\mathcal{C}(\mathcal{A})|}/{\delta})} \\
          & \leq_{(3)}  \sqrt{ 2\bar{\sigma}^2 \bar{n}^e_k (\sum_{k=1}^K n^e_k)\log (2\bar{n}^e_k/\delta)   }
        \end{split}
    \end{equation*}
    \end{small}
Here, (1) follows from the Chernoff bound for sums of sub-Gaussian random variables with probability at least $1-\delta/|\mathcal{C}(\mathcal{A})|$; (2) follows from Jensen’s inequality; and (3) uses the facts that the cover is non-overlapping and exhausts all variables $\{R_{hk}\}_{hk,k}$.
\end{proof}

\begin{proof}[Proof of Proposition \ref{consist_ASE}]
\label{proof_consist_ASE}
From Equation~\eqref{eq:matrix_A_F}, taking expectations of $A$ and $F$ yields 
\begin{equation}
\label{proof_consist_ASE_int0}
    S_{N} [(\mathbb{E}(A))^{-1} \mathbb{E}(F) ]_{2}=S_{N}
\Big[
\big(\mathbb{E}(V_A^\top B_A V_A)\big)^{-1} 
\, \mathbb{E}(V_A^\top B_A Y_A)
\Big]_{2}= \tau(\alpha), \ \  
A \in \{D, R, S\}.
\end{equation}
where $\tau(\alpha)$ is defined in Definition \ref{gen_estimand}. Hence, it suffices to establish the rate of convergence for each component of 
$V_A^\top B_A V_A$ and $V_A^\top B_A Y_A$ toward their expectations.  
The convergence rate of $\hat{\tau}(\alpha)$ to $\tau(\alpha)$ is of the same order as 
$\max\big( \| A_{N} - \mathbb{E}(A_{N}) \|, 
\| b_{N} - \mathbb{E}(b_{N}) \| \big)$, 
as implied by
\begin{equation}
\label{proof_consist_ASE_int1}
\begin{split}
&A^{-1}_{N} b_{N} 
- \mathbb{E}(A^{-1}_{N}) \mathbb{E}(b_{N}) \\
&= 
A^{-1}_{N} \big(b_{N}-\mathbb{E}(b_{N})\big)
- (\mathbb{E}A_{N})^{-1}
\big(A_{N}-\mathbb{E}(A_{N})\big)
(\mathbb{E}A_{N})^{-1} \mathbb{E}(b_{N})
+ R_n,
\end{split}
\end{equation}
where $A_{N}=V^\top_A B_A V_A$, 
$b_{N}=V^\top_A B_A Y_A$, and 
\[
R_n = O_p\!\left(
\max\big(
\| A_{N} - \mathbb{E}(A_{N}) \|,
\| b_{N} - \mathbb{E}(b_{N}) \|
\big)
\right).
\]
From~\eqref{eq:matrix_A_F}, each component of $V_A^\top B_A V_A$ can be written as 
$\sum_{k=1}^K \sum_{ik=1}^{n_k} \sum_{jk \neq ik} S_{ik,jk} R_{ik,jk}$, 
where $\{ S_{ik,jk} R_{ik,jk} \}_{ik,jk}$ are sub-Gaussian with parameters $C S_{ik,jk}^2$ under 
Assumptions~\ref{unif_bound_weight}, ~\ref{pos_Ratio} and ~\ref{unif_bound_pot_out} where $C>0$ is a constant. Applying  Lemma~\ref{conc_clustern_graph}, we obtain that, with probability at least $1-\delta$,
\begin{equation}
\label{proof_consist_ASE_int2}
\begin{split}
& \left|
\sum_{k=1}^K \sum_{ik=1}^{n_k} \sum_{jk \neq ik}
\big[ S_{ik,jk} R_{ik,jk} - \mathbb{E}(S_{ik,jk} R_{ik,jk}) \big]
\right| \\
&\leq
C \, \max_{ik,jk} S_{ik,jk}
\sqrt{ \bar{n}_k (\bar{n}_k - 1) N \, \bar{n}_k \log(2\bar{n}_k^2/\delta)} \leq
C \, \max_{ik,jk} S_{ik,jk} \, \bar{n}_k^{3/2} N^{1/2} \log(2\bar{n}_k^2/\delta).
\end{split}
\end{equation}
Combining~\eqref{proof_consist_ASE_int1} and~\eqref{proof_consist_ASE_int2} yields
\begin{equation*}
\begin{split}
&S_{N}
\Big|
\big(V_A^\top B_A V_A\big)^{-1} (V_A^\top B_A Y_A)
- \big(\mathbb{E}(V_A^\top B_A V_A)\big)^{-1} \mathbb{E}(V_A^\top B_A Y_A)
\Big| \\
&\leq
C \, \max_{ik,jk} S_{ik,jk} \, S_{N} \, 
\bar{n}_k^{3/2} N^{1/2} \log(2\bar{n}_k^2/\delta).
\end{split}
\end{equation*}
Therefore, if 
$\max_{ik,jk} S_{ik,jk} S_{N} \bar{n}_k^{3/2} N^{1/2} 
\log(2\bar{n}_k^2/\delta) \to 0$ as $N \to \infty$, then
\[
\left| \hat{\tau}(\alpha) - \tau(\alpha) \right|
\;\overset{N \to \infty}{\longrightarrow}\; 0.
\]
\end{proof}

\begin{proof}[Proof of Corollary~\ref{consist_out_ASE}]
Under the assumptions of Corollary~\ref{consist_out_ASE}, we have
\begin{equation*}
\begin{split}
&\max_{ik,jk} 
S_{ik,jk} S_{N} 
\bar{n}_k^{3/2} N^{1/2} 
\log(2\bar{n}_k^2N)
= 
\max_{ik,jk} 
\frac{1}{N^{\mathrm{out}} |\mathcal{N}^{\mathrm{out}}_{jk}|} 
\, \bar{n}_k^{3/2} N^{1/2} 
\log(2\bar{n}_k^2N)  \\
&\leq 
\frac{1}{K} \, \bar{n}_k^{3/2} K^{1/2} \bar{n}_k^{1/2} 
\log(2\bar{n}_k^2N) 
= 
\frac{\bar{n}_k^{2}}{K^{1/2}} 
\log(2\bar{n}_k^2N) 
\;\longrightarrow\; 0.
\end{split}
\end{equation*}
Note that if $n_k = O(\bar{n}_k)$ and $|\mathcal{N}^{\mathrm{out}}_{jk}| = O(\bar{n}_k)$ for all $jk$ and $k$, then the rate condition is upper bounded by $K^{-1/2} \log(2\bar{n}_k^2N)$.  
Hence, the convergence rate of $\hat{\tau}_{A}(\alpha)$ toward $\tau(\alpha)$ is faster than that obtained under Assumption~\ref{ordern_eta_cluster} alone.
\end{proof}

\begin{proof}[Proof of Corollary \ref{consist_out_ISE}]
Under the assumptions of Corollary~\ref{consist_out_ISE}, we have
\begin{small}
\begin{equation*}
\begin{split}
&\max_{ik,jk} 
S_{ik,jk} S_{N} 
\bar{n}_k^{3/2} N^{1/2} 
\log(2\bar{n}_k^2N)
= 
\max_{ik,jk} 
\frac{1}{N} 
\, {\sum_{k=1}^K n_k(n_k-1)}{N^{-1}}  \bar{n}_k^{3/2} N^{1/2} 
\log(2\bar{n}_k^2N)  \\
&\leq 
\frac{ N \bar{n}_k}{N}   \, \bar{n}_k^{3/2} N^{-1/2} 
\log(2\bar{n}_k^2N) 
= 
\frac{\bar{n}^{5/2}_{k}}{N^{1/2}} 
\log(2\bar{n}_k^2N) 
\;\longrightarrow\; 0.
\end{split}
\end{equation*}
\end{small}
Note that if $n_k = O(\bar{n}_k)$ for all $k$, then the rate condition is upper bounded by $\frac{\bar{n}^2_{k}}{K^{1/2}} \log(2\bar{n}_k^2N)$.  
Hence, the convergence rate of $\hat{\tau}_{A}(\alpha)$ toward $\tau(\alpha)$ is faster than that obtained under Assumption~\ref{ordern_eta_cluster} alone.
\end{proof}

In order to prove the central limit theorem, we first state the following Lemma taken from Lemma~1 in \citet{ogburn2022causal}.

\begin{lemma}[CLT for the dependent sum] 
\label{Theorem_CLT_dependence}
Let $X_1,\ldots,X_N$ be bounded, mean-zero random variables with finite fourth moments.  
Let $D_i$ denote the neighborhood of unit $i$ in the dependence graph, as defined in Definition \ref{def:dependence_graph}, induced by $X_1,\ldots,X_N$.  
If $|D_i|\leq c(N)$ for all $i \in \{1,\ldots,N\}$ and $c^2(N)/N \rightarrow 0$, then
\[
\big[\mathrm{var}\big(\sum_{i=1}^N X_i\big)\big]^{-1/2} \sum_{i=1}^N X_i \;\;\overset{d}{\longrightarrow}\;\; N(0,1).
\]
\end{lemma}

Furthermore, note that each element in equations~\eqref{eq:matrix_A_F} can be written as \\
$\sum_{k=1}^K \sum_{ik=1}^{n_k} \sum_{jk \neq ik} S_{ik,jk} R_{ik,jk},$
where $\{R_{ik,jk}\}_{ik,jk}$, for $i,j \in \{1, \ldots, n_k\}$ and $k \in \{1, \ldots, K\}$, 
are bounded random variables under Assumptions~\ref{unif_bound_weight} and~\ref{unif_bound_pot_out}. We further impose the following regularity condition on the cluster-level variance without the estimand weights.  

\begin{assumption}
\label{lower_bound_variance}
For each $k \in \{1, \ldots, K\}$, $\mathrm{var}\!\left(\sum_{ik=1}^{n_k} \sum_{jk \neq ik} R_{ik,jk}\right) > C,$
where $C > 0$ is a constant. 
\end{assumption}

Assumption~\ref{lower_bound_variance} is mild. 
Once the estimand weights $S_{ik,jk}$, which may depend on $N$, are factored 
out of $R_{ik,jk}$, the remaining term $R_{ik,jk}$ contains only products of the bounded quantities
$W_{jk}(\mathbf{Z}_k)$, $Y_{ik}(\mathbf{Z}_k)$, and $Z_{jk}$, ensuring nondegenerate variance at the cluster level.

\begin{proof}[Proof of Theorem \ref{CLT_ASE_est}]
Let 
${\boldsymbol{\hat{\beta}}}_A(\alpha):=(V_A^\top B_A V_A)^{-1} V_A^\top B_A Y_A $ \\ and $\boldsymbol{\beta}^r_{A}(\alpha):=\big(\mathbb{E}(V_A^\top B_A V_A)\big)^{-1} 
\, \mathbb{E}(V_A^\top B_A Y_A)$ for $A \in \{D,S,R\}$. Based on \eqref{proof_consist_ASE_int0}, we have $\tau(\alpha)=S_{N}\boldsymbol{\beta}^r_{A,2}(\alpha)$. Let 
\begin{equation}
\label{ase_res_working_model}
    \begin{split}
        \xi_{A}:=Y_A-V_A\boldsymbol{\beta}^r_A(\alpha)
    \end{split}
\end{equation}
for $A \in \{D,S,R\}$. Then 
\begin{small}
\begin{equation}
\label{proof_consist_ASE_int3}
    \begin{split}
      & S_{N} \ (\boldsymbol{\hat{\beta}}_{A}(\alpha)-\boldsymbol{\beta}^r_{A}(\alpha))= S_{N} \left[ (V^T_AB_AV_A)^{-1} V^T_A B_A Y_A- \boldsymbol{\beta}_A^r(\alpha) \right] \\
      & =_{(1)} S_{N} \left[ \boldsymbol{\beta}_{A}^r(\alpha)+ (V^T_AB_AV_A)^{-1} V^T_A B_A {\xi}_A - \boldsymbol{\beta}_{A}^r(\alpha) \right]\\
      & = S_{N} \left[\large(V^T_AB_AV_A\large)^{-1}-\large(\mathbb{E}(V^T_AB_AV_A)\large)^{-1}+\large(\mathbb{E}(V^T_AB_AV_A)\large)^{-1} \right]  V^T_A B_A {\xi}_A\\
      &:= S_{N} \left[ \left( \tilde{R}^{-1}_m - (\mathbb{E}(\tilde{R}_m))^{-1}  \right) \right]  R+ S_{N}  (\mathbb{E}(\tilde{R}_m))^{-1} R  := (I)+(II) \\
    \end{split}
\end{equation}
\end{small}
where $\tilde{R}_{m}=V^T_{A}B_AV_A$ and $R=V^T_AB_A \xi_{A}$. (1) follows from \eqref{ase_res_working_model}.
Define $||A||_{max}:=\max_{i,j}|A_{ij}|$. We first show that $||(I)||_{\max}=o_p(1)$. With probability $1-\delta$, 
\begin{small}
\begin{equation}
\label{proof_consist_ASE_int1}
    \begin{split}
        & S_{N}||(\tilde{R}_m)^{-1} - (\mathbb{E}(\tilde{R}_m))^{-1}||_{\max} \leq \Delta_{N} \ S_{N}||\tilde{R}_m - \mathbb{E}(\tilde{R}_m)||_{\max}  \\
    \end{split}
\end{equation}
\end{small}
where
\begin{small}
\begin{equation}
\label{proof_consist_ASE_int7}
    \begin{split}
      & \Delta_{N}=||(\mathbb{E}(\tilde{R}_m))^{-1}||_{1}||(\mathbb{E}(\tilde{R}_m))^{-1}||_{\infty}/(1-||(\tilde{R}_m)-\mathbb{E}(\tilde{R}_m)||_{\max}) \\
      & = ||(\mathbb{E}(\tilde{R}_m))^{-1}||_{1}||(\mathbb{E}(\tilde{R}_m))^{-1}||_{\infty}\left( 1+\frac{||\tilde{R}_m-\mathbb{E} (\tilde{R}_m)||_{\max}}{1-||\tilde{R}_m-\mathbb{E}\tilde{R}_m)||_{\max}} \right)\\
      & \leq_{(1)} {C}S_{N}^{-2} \left(1+ \frac{o_p(1)}{1-o_p(1)}\right)
    \end{split}
\end{equation}
\end{small}
(1) is based on \eqref{proof_consist_ASE_int2} and based on the formula in \eqref{eq:matrix_A_F}, we have 
\begin{equation*}
    (\mathbb{E}(\tilde{R}_m))^{-1}= \frac{1}{S} \begin{bmatrix}
        1 & -1 \\
        -1 & 2
    \end{bmatrix}
\end{equation*}
Therefore, $||(\mathbb{E}(\tilde{R}_m))^{-1}||_{1}\leq \frac{C}{S_N}= C S_{N}^{-1}$ and $||(\mathbb{E}(\tilde{R}_m))^{-1}||_{\infty}\leq \frac{C}{S_N}= CS_{N}^{-1}$ for some constant $C>0$. Combine formulas \eqref{proof_consist_ASE_int2}, \eqref{proof_consist_ASE_int1}, and \eqref{proof_consist_ASE_int7}, 
\begin{equation}
\label{proof_consist_ASE_int4}
    \begin{split}
        & S_{N}||(\tilde{R}_m)^{-1} - (\mathbb{E}(\tilde{R}_m))^{-1}||_{\max}  \leq CS_{N}^{-2}(1+o_p(1))  \max_{ik,jk} S_{ik,jk} \, S_{N} \, \bar{n}_k^{3/2} N^{1/2} \log(2\bar{n}_k^2/\delta)\\
        & = C(1+o_p(1)) \max_{ik,jk} S_{ik,jk} \, S_{N}^{-1} \, \bar{n}_k^{3/2} N^{1/2} \log(2\bar{n}_k^2/\delta) \\
    \end{split}
\end{equation}
Next, we establish the convergence rate for $R$ in \eqref{proof_consist_ASE_int3}. First note that by Lemma \ref{mean_zero_ASE}, 
$\mathbb{E}(V^{T}_{A,a} B_A \xi_{A,b})=0$ for all $a,b \in \{1,2\}$ in the case of ASE. Therefore, with probability $1-\delta$, we have 
\begin{small}
\begin{equation}
\label{proof_consist_ASE_int5}
    \begin{split}
    | R_a |=_{(1)} \left\lvert  \sum_{k=1}^K \sum_{ik=1}^{n_k} \sum_{jk\neq ik} S_{ik,jk} R_{ik,jk} \right\rvert \leq_{(2)} \max_{ik,jk} S_{ik,jk} \,  \bar{n}_k^{3/2} N^{1/2} \log(2\bar{n}_k^2/\delta).
    \end{split}
\end{equation}
\end{small}
for $a\in\{1,2\}$. (1) follows that $V^{T}_{A,a} B_A \xi_{A,b}$ can be expressed as $\sum_{k=1}^K \sum_{ik=1}^{n_k} \sum_{jk\neq ik} S_{ik,jk} R_{ik,jk}$ where $R_{ik,jk}=V_{A,a,ik,jk} W_{jk}(\mathbf{Z}_k) \xi_{A,b,ik,jk}$. $(2)$ follows Proposition \ref{consist_ASE}. Combine formulas \eqref{proof_consist_ASE_int4} and \eqref{proof_consist_ASE_int5}, we have 
\begin{equation*}
    \begin{split}
        ||(I)||_{max}\leq  C(1+o_p(1)) \max_{ik,jk} S^2_{ik,jk} \, S_{N}^{-1} \, \bar{n}_k^{3} N \log^2(2\bar{n}_k^2/\delta).
    \end{split}
\end{equation*}
if $\max_{ik,jk} S_{ik,jk} \, S_{N}^{-1/2} \, \bar{n}_k^{3/2} N^{1/2} \log(2\bar{n}_k^2/\delta)\overset{N\rightarrow \infty}{\longrightarrow} 0$, then we have $||(I)||_{max}=o_p(1)$.
 Now consider term $(II)$ in \eqref{proof_consist_ASE_int3}. For $a\in \{1,2\}$, we have 
\begin{equation*}
    \begin{split}
        (II)= [ {S_{N}^{-1}}\mathbb{E}(\tilde{R}_m)]^{-1} R
    \end{split}
\end{equation*}
where ${S_{N}^{-1}}\mathbb{E}(\tilde{R}_m)=O(1)$. Therefore, we first consider the CLT on $R$. For each element $R_a$ in $R$ where $a \in \{1,2\}$, we have 
\begin{equation*}
 R_a=_{(1)}V^{T}_{A,a} B_A \xi_A- \mathbb{E}(V^{T}_{A,a} B_A \xi_A) = \sum_{k=1}^K \sum_{ik=1}^{n_k} \sum_{jk\neq ik}[S_{ik,jk}R_{ik,jk}-\mathbb{E}(S_{ik,jk}R_{ik,jk}) ].
\end{equation*}
Step (1) follows because $\mathbb{E}(V^{T}_{A,a} B_A \xi_A)=0$ by Lemma \ref{mean_zero_ASE}. Based on Assumptions \ref{pos_Ratio}, \ref{unif_bound_pot_out} and \ref{bounded_x}, $S_{ik,jk}R_{ik,jk}-\mathbb{E}(S_{ik,jk}R_{ik,jk}) $ is bounded, has finite fourth moment. Then the dependence graph defined in Definition \ref{def:dependence_graph} induced by $\{S_{ik,jk}R_{ik,jk}\}_{ik,jk}$ are fully connected within each cluster and there are no connections across clusters. There are $n_k(n_k-1)$ elements in each cluster. Therefore, if $\bar{n}^2_k(\bar{n}_k-1)^2/(\sum_{k=1}^K n_k(n_k-1))  \longrightarrow 0$, then Lemma \ref{Theorem_CLT_dependence} applies, yielding $(\mathrm{var}(R_a))^{-1/2} R_a \overset{N \rightarrow \infty}{\longrightarrow} \mathcal{N}(0,1)$ for $a \in \{1,2\}$. By the Cramér–Wold device, it then follows that
\begin{equation}
\label{CLT_component}
    \begin{split}
      \tilde{\Gamma}^{-1/2}_{A} R \overset{N \rightarrow \infty }\longrightarrow \mathcal{N}(0,I) 
    \end{split}
\end{equation}
where $\tilde{\Gamma}_{A}=var(R)$. Moreover, because each component of $\frac{1}{S_{N}}\mathbb{E}(\tilde{R}_m)$ is of order $O(1)$, we obtain $var
   \left[\Big(\tfrac{1}{S_{N}}\mathbb{E}(\tilde{R}_m)\Big)^{-1}  R \right] 
    = \Omega^{-1}_{A} \tilde{\Gamma}_{A} \Omega^{-1}_{A}$,
where $\Omega_{A}=\tfrac{1}{S_{N}}\mathbb{E}(\tilde{R}_m)$. Therefore,
\begin{equation}
\label{clt_II_term}
    S_{N}  (\mathbb{E}(\tilde{R}_m))^{-1} R 
    \overset{N\rightarrow\infty}{\longrightarrow} \mathcal{N}(0, \Sigma),
\end{equation}
with $\Sigma=\lim_{N\rightarrow \infty} \Sigma_{A} =\lim_{N \rightarrow\infty}\Omega^{-1}_{A} \tilde{\Gamma}_{A} \Omega^{-1}_{A}$ based on uniform integrability under Assumptions \ref{unif_bound_weight}, \ref{unif_bound_pot_out}, and \ref{bounded_x}. Moreover, we have
\begin{equation}
\label{lower_bound_var}
    \begin{split}
       & \Sigma_{A,(2,2)}\geq_{(1)} C var\left(\sum_{k=1}^K \sum_{ik=1}^{n_k} \sum_{jk\neq ik} S_{ik,jk} R_{ik,jk} \right)\geq C \sum_{k=1}^K \min_{ik,jk}S^2_{ik,jk} var(\sum_{ik=1}^{n_k} \sum_{jk\neq ik}  R_{ik,jk})\\
       & =  C \sum_{k=1}^K \min_{ik,jk}S^2_{ik,jk} var(\sum_{ik=1}^{n_k} \sum_{jk\neq ik}  R_{ik,jk}) \geq_{(2)} C K \min_{ik,jk}S^2_{ik,jk}
    \end{split}
\end{equation}
(1) is by $||\Omega_{A}||_{\max}=O(1)$ and taking the smallest element in $\tilde{\Gamma}_A$. $(2)$ is by Assumption \ref{lower_bound_variance}. Then we have 
\begin{equation*}
    \begin{split}
    &\Sigma^{-1/2}_{A,(2,2)} S_{N} (0,1) \left[ \left( \tilde{R}^{-1}_m - (\mathbb{E}(\tilde{R}_m))^{-1}  \right) \right]  R \\
    &\leq K^{-1/2} \frac{1}{\min_{ik,jk} S_{ik,jk}} (C_0+o_p(1)) \max_{ik,jk} S^2_{ik,jk} \, \rho^{-1}_{N} \, \bar{n}_k^{3} N \log^2(2\bar{n}_k^2/\delta)\\
    & \leq C \frac{\max_{ik,jk} S^2_{ik,jk}}{\min_{ik,jk} S_{ik,jk}} S_{N}^{-1} K^{1/2} \bar{n}^4_k \log^2(2\bar{n}_k^2/\delta)
    \end{split}
\end{equation*}
If $\max_{ik,jk} S^2_{ik,jk}{(\min_{ik,jk} S_{ik,jk})}^{-1} S_{N}^{-1} K^{1/2} \bar{n}^4_k \log^2(2\bar{n}_k^2/\delta)\longrightarrow 0$, then 
\begin{equation}
\label{clt_hat_tau_decomp}
    \begin{split}
    &\Sigma^{-1/2}_{A,(2,2)}  \left( \hat{\tau}_{A}(\alpha)-\tau_{A}(\alpha) \right) =\Sigma^{-1/2}_{A,(2,2)} S_{N} (0,1) \left( \boldsymbol{\hat{\beta} }_{A}(\alpha)-\boldsymbol{\hat{\beta} }^r_{A}(\alpha) \right) \\
    & = \Sigma^{-1/2}_{A,(2,2)} S_{N} (0,1) \left[ \left( \tilde{R}^{-1}_m - (\mathbb{E}(\tilde{R}_m))^{-1}  \right) \right]  R+  \Sigma^{-1/2}_{A,(2,2)} S_{N} (0,1) (\mathbb{E}(\tilde{R}_m))^{-1} R \\
    &= o_p(1) + \Sigma^{-1/2}_{A,(2,2)} S_{N} (0,1) (\mathbb{E}(\tilde{R}_m))^{-1} R
    \end{split}
\end{equation}
Combining results \eqref{clt_hat_tau_decomp} and  \eqref{clt_II_term}, we obtain
\[
    \Sigma^{-1/2}_{A,(2,2)}  \left( \hat{\tau}_{A}(\alpha)-\tau_{A}(\alpha) \right)
    \overset{N\rightarrow\infty}{\longrightarrow} \mathcal{N}(0, I).
\]
To analyze the convergence rate of the estimator, observe that $\Sigma_{A}= \Omega^{-1}_{A}\tilde{\Gamma}_{A} \Omega^{-1}_{A}$, where $\|\Omega_A\|_{\max}=O(1)$. Consequently, it suffices to determine the order of each component of $\tilde{\Gamma}_{A}$. For $h \in \{1,2\}$, we obtain, under Assumptions \ref{unif_bound_weight}--\ref{unif_bound_pot_out}, that 
\begin{small}
\begin{equation*}
    \begin{split}
       & \tilde{\Gamma}_{A,(h,h)}=  var\left(\sum_{k=1}^K \sum_{ik=1}^{n_k} \sum_{jk\neq ik} S_{ik,jk} R_{ik,jk} \right) =O \left(\sum_{k=1}^K \sum_{i_1k=1}^{n_k} \sum_{j_1k\neq i_1k}\sum_{i_2k =1}^{n_k} \sum_{j_2k\neq i_2k} S_{i_1kj_1k} S_{i_2kj_2k} \right). \\
    \end{split}
\end{equation*}
\end{small}
Similarly, 
\begin{small}
\begin{equation*}
    \begin{split}
        & \tilde{\Gamma}_{A,(1,2)}=\tilde{\Gamma}_{A,(2,1)}= cov \left(\sum_{k=1}^K \sum_{i_1k=1}^{n_k} \sum_{j_1k\neq i_1k} S_{i_1k,j_1k} R_{i_1k,j_1k}, \  \sum_{k=1}^K \sum_{i_2k=1}^{n_k} \sum_{j_2k\neq i_2k} S_{i_2k,j_2k} R_{i_2k,j_2k}  \right) \\
        & = O \left(\sum_{k=1}^K \sum_{i_1k=1}^{n_k} \sum_{j_1k\neq i_1k}\sum_{i_2k =1}^{n_k} \sum_{j_2k\neq i_2k} S_{i_1kj_1k} S_{i_2kj_2k} \right)
    \end{split}
\end{equation*}
\end{small}
Consequently, $\Sigma_{A,(2,2)}$ is of order $\tilde{N}^{-1}$, where $\tilde{N}$ is defined in Theorem \ref{CLT_ASE_est}. It follows that the estimator $\hat{\tau}_A(\alpha)$ converges at rate $\tilde{N}^{-1/2}$.
\end{proof}

\begin{proof}[Proof of Proposition \ref{ase_conserve_variance_est}]
Based on the formula of $\Gamma_{N}$ in Proposition \ref{ase_conserve_variance_est}, we have
\begin{equation}
\label{proof:ase_conserve_variance_est_int1}
    \begin{split}
      & \tilde{\Gamma}_{A}= var(V^T_AB_A\xi_A) =_{(1)} \sum_{k=1}^K var(V^T_{A,k} B_{A,k} \xi_{A,k} )\\
      & \preceq_{(2)} \sum_{k=1}^K  \mathbb{E} (V^T_{A,k}B_{A,k} \xi_{A,k} \ \xi^T_{A,k} B_{A,k} V_{A,k} ):= \Gamma_{A} \\ 
    \end{split}
\end{equation}
(1) follows from Assumption \ref{part_intf}. (2) holds in the sense of positive semi-definiteness.  
We now show that $\sum_{k=1}^K  
V^T_{A,k}B_{A,k} \hat{\xi}_{A,k}\hat{\xi}^T_{A,k} B_{A,k} V_{A,k}=\hat{\Gamma}_{A}$ is a consistent estimator for $\Gamma_{A}$. First, note that with probability $1-\delta$,  

\begin{small}
\begin{equation}
\label{proof_ase_conserve_variance_est_int5}
    \begin{split}
    & \left\lvert\left\lvert \sum_{k=1}^K  V^T_{A,k}B_{A,k} \hat{\xi}_{A,k} \ \hat{\xi}^T_{A,k}  B_{A,k} V_{A,k}- \sum_{k=1}^K  V^T_{A,k} B_{A,k} {\xi}_{A,k} \ {\xi}^T_{A,k} B_{A,k} V_{A,k} \right\rvert\right\rvert_{op} \\
    & \leq_{(1)} \sum_{k=1}^K||  V^T_{A,k} B_{A,k}(\hat{\xi}_{A,k} \ \hat{\xi}^T_{A,k}-{\xi}_{A,k} \ {\xi}^T_{A,k})B_{A,k}V_{A,k} ||_{op} \\
    & \leq \sum_{k=1}^K||  V^T_{A,k}B_{A,k}||^2_{op} \ \  ||(\hat{\xi}_{A,k} \ \hat{\xi}^T_{A,k}-{\xi}_{A,k} \ {\xi}^T_{A,k}) ||_{op} \\
    & \leq_{(2)} \sum_{k=1}^K||  V^T_{A,k} B_{A,k}||^2_{op} \ \ ||\hat{\xi}_{A,k}- {\xi}_{A,k}||_2 \left( ||\hat{\xi}_{A,k} ||_2 +||{\xi}_{A,k}||_2  \right) \\
    & \leq   \sum_{k=1}^K ||  V^T_{A,k}B_{A,k}||^2_{F} ||V_{A,k}(\boldsymbol{\hat{\beta}}_{A}(\alpha)-\boldsymbol{\beta}^r(\alpha))||_2 2 \sqrt{n^2_k} \\
    & \leq {C} \sum_{k=1}^K  \max_{ik,jk} S^{2}_{ik,jk} (n^2_k)||V_{A,k}||_{op} ||\boldsymbol{\hat{\beta}}_{A}(\alpha)-\boldsymbol{\beta}^r(\alpha)||_2 (n_k)\\
    & \leq_{(3)} C  \sum_{k=1}^K  \max_{ik,jk} S^{2}_{ik,jk} n_k^{5} \ \ \ { \max_{ik,jk} S_{ik,jk}  \, 
\bar{n}_k^{3/2} N^{1/2} \log(2\bar{n}_k^2/\delta) }\\
    & \leq C \max_{ik,jk} S^{3}_{ik,jk} N^{3/2} \bar{n}^{11/2}_k \log(2\bar{n}^2_k/\delta)  \\
    \end{split}
\end{equation}
\end{small}
Here $\|\cdot\|_{op}$ and $\|\cdot\|_F$ denote the operator and Frobenius norms, respectively.  
(1) follows from the triangle inequality and the submultiplicativity of the operator norm.  
(2) uses the decomposition 
$(\hat{\xi}_{A,k}\hat{\xi}^T_{A,k}-{\xi}_{A,k}{\xi}^T_{A,k})=(\hat{\xi}_{A,k}-{\xi}_{A,k})\hat{\xi}^T_{A,k}+\xi_{A,k}(\hat{\xi}^T_{A,k}-{\xi}^T_{A,k})$.  
(3) relies on $\|V_{A,k}\|_{op}\leq \|V_{A,k}\|_{F}$, the fact that $\boldsymbol{\hat{\beta}}_{A}(\alpha)$ and $\boldsymbol{\beta}^r(\alpha)$ have fixed dimension with $p^2+p$ entries, and Proposition \ref{consist_ASE}. Therefore, if $S_3:=\max_{ik,jk} S^{3}_{ik,jk} N^{3/2} \bar{n}^{11/2}_k \log(2\bar{n}^2_k/\delta)\longrightarrow 0$, 
\begin{small}
\begin{equation}
\label{proof_ase_conserve_variance_est_int1}
    \begin{split}
\left\lvert\hat{\Gamma}_{A }-{\Gamma}_{A}\right\rvert \overset{N\rightarrow \infty}{\longrightarrow}0.
    \end{split}
\end{equation}
\end{small}
We next consider $\bigl({\Gamma}_{A}-\mathbb{E}({\Gamma}_{A}) \bigr)$.  
For each $(a,b)$ entry, with probability $1-\delta$, 
\begin{small}
\begin{equation*}
\label{proof_ase_conserve_variance_est_int2}
    \begin{split}
    & \left\lvert \sum_{k=1}^K  V^{a \ T}_{A,k}B_{A,k} {\xi}_{A,k} \ {\xi}^T_{A,k} B_{A,k} V^b_{A,k}-\mathbb{E}(\sum_{k=1}^K  V^{a \ T}_{A,k}B_{A,k} {\xi}_{A,k} {\xi}^T_{A,k} B_{A,k} V^b_{A,k}) \right\rvert \\
    & = \left\lvert\sum_{k=1}^K ( \sum_{\substack{ik=1 \\ jk \ne ik}}^{n_k} S_{ik, jk} R^a_{ik,jk} ) ( \sum_{\substack{ik=1 \\ jk \ne ik}}^{n_k} S_{ik, jk} R^b_{ik,jk} )-  \sum_{k=1}^K \mathbb{E} \big[( \sum_{\substack{ik=1 \\ jk \ne ik}}^{n_k} S_{ik, jk} R^a_{ik,jk} ) ( \sum_{\substack{ik=1 \\ jk \ne ik}}^{n_k} S_{ik, jk} R^b_{ik,jk} ) \big]    \right\rvert \\
    & \leq_{(1)} C \sqrt{ \max_{(i_1k,j_1k),(i_2k,j_2k)} S^2_{i_1k,j_1k} S^2_{i_2k,j_2k} (\bar{n}^2_k)^{2} \left( \sum_{k=1}^K (\bar{n}^2_k)^{2}  \right) \log(\frac{2(\bar{n}^2_k)^{2} }{\delta}) } \\
    & \leq  C   \sqrt{ \max_{(i_1k,j_1k),(i_2k,j_2k)} S^2_{i_1k,j_1k} S^2_{i_2k,j_2k}   N {(\bar{n}^7_k) \log (2\bar{n}^4_k /\delta) } } \\
    & \leq \max_{(i_1k,j_1k),(i_2k,j_2k)} S_{i_1k,j_1k} S_{i_2k,j_2k}   N^{1/2} {(\bar{n}^{7/2}_k) \log^{1/2} (2\bar{n}^4_k /\delta) }
    \end{split}
\end{equation*}
\end{small}
(1) follows from Lemma~\ref{conc_clustern_graph}, noting that 
$S_{i_1k,j_1k} S_{i_2k,j_2k} R^a_{i_1k,j_1k} R^b_{i_2k,j_2k}$ 
is a sub-Gaussian random variable with parameter 
$CS_{i_1k,j_1k}^2 S_{i_2k,j_2k}^2$. Then if $$S_4:=\max_{(i_1k,j_1k),(i_2k,j_2k)} S_{i_1k,j_1k} S_{i_2k,j_2k}   N^{1/2} {(\bar{n}^{7/2}_k) \log^{1/2} (2\bar{n}^4_k /\delta) } \longrightarrow 0,$$ we have 
\begin{equation}
\label{proof_ase_conserve_variance_est_int3}
    \begin{split}
       \big\lvert  \hat{\Gamma}_{A}-\mathbb{E}(\Gamma_{A}) \big\rvert  \overset{N\rightarrow \infty}{\longrightarrow} 0 
    \end{split}
\end{equation}
From the proof of Proposition \ref{consist_ASE}, we also have if $\max_{ik,jk} S_{ik,jk} S_{N}^{-1}\, \bar{n}_k^{3/2} N^{1/2} \log(2\bar{n}_k^2/\delta)\rightarrow0$, then 
\begin{equation}
\label{proof_ase_conserve_variance_est_int4}
    \begin{split}
        \lvert\hat{\Omega}_{A}-\Omega_{A}\rvert
        \overset{N\rightarrow \infty}{\longrightarrow}0.
    \end{split}
\end{equation}
Therefore, by the continuous mapping theorem and equations \eqref{proof_ase_conserve_variance_est_int1}, \eqref{proof_ase_conserve_variance_est_int3} and \eqref{proof_ase_conserve_variance_est_int4}, it follows that 
\begin{equation*}
\big\lvert{\hat{\Omega}}^{-1}_{A} \hat{\Gamma}_{A} {\hat{\Omega}}^{-1}_{A} 
    - {{\Omega}}^{-1}_{A} {\Gamma}_{A} {{\Omega}}^{-1}_{A}\big\rvert 
    \overset{N\rightarrow \infty}{\longrightarrow }0.
\end{equation*}
Note that the rate condition 
$\max_{ik,jk} S_{ik,jk} S_{N}^{-1}\, 
\bar{n}_k^{3/2} N^{1/2} \log(2\bar{n}_k^2/\delta)$ 
converges faster than 
$\max_{ik,jk} S_{ik,jk}^{3} N^{3/2} 
\bar{n}_k^{11/2} \log(2\bar{n}_k^2/\delta)$. 
Hence, the overall rate condition can be expressed in terms of 
$\max(S_3, S_4)$.
\end{proof}

\subsection{Connecting estimators to CSE}
\label{Appendix:connect_CSE}
\begin{definition}
\label{Lambda_formula}
For $A \in \{D,S\}$, define the linear operator 
$\Lambda:\mathbb{R}^{2\times2}\to\mathbb{R}^{2\times2}$ that orthogonalizes causal components $(Z,Z\tilde{X})$ and non-causal components $(1,\tilde{X})$ as
\begin{small}
  \begin{equation*}
\begin{split}
  & \Lambda= \left\lbrace  \sum_{k=1}^K \sum_{jk=1}^{n_k} \sum_{ik\neq jk} \mathbb{E} \left[ B_{ik,jk} \begin{pmatrix}
   1 \\
    \tilde{X}_{jk} 
  \end{pmatrix} \begin{pmatrix}
   1 \\
    \tilde{X}_{jk} 
  \end{pmatrix}^T  \right]\right\rbrace^{-1} \\
  &\sum_{k=1}^K \sum_{jk=1}^{n_k} \sum_{ik\neq jk} \mathbb{E} \left[ B_{ik,jk} \begin{pmatrix}
    Z_{jk}  \\
    Z_{jk} \tilde{X}_{jk} 
  \end{pmatrix} \begin{pmatrix}
    Z_{jk}  \\
    Z_{jk} \tilde{X}_{jk} 
  \end{pmatrix}^T  \right]= \begin{pmatrix} 
   \frac{1}{2} & 0 \\
    0 & \frac{1}{2}
  \end{pmatrix}.\\
  \end{split}
\end{equation*}
\end{small}
\end{definition}

\begin{proof}[Proof of Proposition \ref{relation_est_CSE}]
We establish a stronger result than stated in Proposition \ref{relation_est_CSE}, namely that 
$\hat{\boldsymbol{\beta}}_{D}(\alpha)=\hat{\boldsymbol{\beta}}_{S}(\alpha)$. 
The argument proceeds by showing that each corresponding component of 
$V_D^\top B_D V_D$ coincides with that of $V_S^\top B_S V_S$, and likewise for 
$V_D^\top B_D Y_D$ and $V_S^\top B_S Y_S$. 

Let \( V_{A,a} \) and \( V_{A,b} \) denote columns \( a \) and \( b \), respectively, of the design matrix \( V_A \), for \( A \in \{D, S\} \). Then
\begin{equation}
\label{proof_relation_est_CSE_int1}
    \begin{split}
        V_{D,a}^\top B_D V_{D,b}
        &= \sum_{k=1}^K \sum_{ik=1}^{n_k} 
        \big[(\mathbf{V}_{D,a})_{k,-ik} \circ \mathbf{B}_{ik,-ik}\big]^\top (\mathbf{V}_{D,b})_{k,-ik} \\
        &= \sum_{k=1}^K \sum_{ik=1}^{n_k} \sum_{jk \neq ik}  
        (\mathbf{V}_{D,a})_{jk} \, B_{ik,jk} \, (\mathbf{V}_{D,b})_{jk},
    \end{split}
\end{equation}
where $\mathbf{V}_{D,1}=\mathbf{1}$, $\mathbf{V}_{D,2}=\tilde{\mathbf{X}}$, $\mathbf{V}_{D,3}=\mathbf{Z}^*$, and $\mathbf{V}_{D,4}=(\mathbf{Z}\circ \tilde{\mathbf{X}})^*$. From the sender’s perspective we have
\begin{small}
\begin{equation}
\label{proof_relation_est_CSE_int2}
    \begin{split}
        V_{S,a}^\top B_S V_{S,b}
        &= \sum_{k=1}^K \sum_{jk=1}^{n_k} (\mathbf{V}_{S,a})_{jk}\,\tilde{S}^r_{jk}S^r_{jk}W_{jk}(\mathbf{Z}_{k})\,(\mathbf{V}_{S,b})_{jk} {=}_{(1)} \sum_{k=1}^K \sum_{ik=1}^{n_k} \sum_{jk \neq ik} 
        (\mathbf{V}_{S,a})_{jk}\, B_{ik,jk}\, (\mathbf{V}_{S,b})_{jk},
    \end{split}
\end{equation}
\end{small}
where step (1) uses the definitions of $\tilde{S}^r_{jk}$ and the identity $S^r_{ik|jk}S^r_{jk}=S^r_{ik,jk}$. 
Since $\mathbf{V}_{D,a}=\mathbf{V}_{S,a}$ for $a\in \{1,\dots,4\}$, combining the last lines of 
\eqref{proof_relation_est_CSE_int1} and \eqref{proof_relation_est_CSE_int2} yields 
$V_{D,a}^\top B_D V_{D,b}=V_{S,a}^\top B_S V_{S,b}$ for all $a,b\in\{1,\dots,4\}$. 

The equivalence of $V_D^\top B_D Y_D$ and $V_S^\top B_S Y_S$ follows analogously by replacing 
$\mathbf{V}_{D,b}$ and $(\mathbf{V}_{D,b})_{k,-ik}$ in \eqref{proof_relation_est_CSE_int1} with 
$Y_D$ and $Y_{ik}\mathbf{1}_{n_k-1}$, and replacing $\mathbf{V}_{S,b}$ with $Y_S$ in 
\eqref{proof_relation_est_CSE_int2}, together with the identity 
$S^r_{jk}\sum_{ik \neq jk} S^r_{ik|jk}Y_{ik} = \sum_{ik \neq jk} S^r_{ik,jk} Y_{ik}$. 

Turning to the receiver’s perspective, for notational convenience we exchange the second and third columns of $V_R$, and correspondingly reorder $\mathrm{diag}(B_R)=({\mathbf{B}^1}^\top,{\mathbf{B}^0}^\top, {\mathbf{B}^1}^\top, {\mathbf{B}^0}^\top)^\top$. A direct calculation gives
\begin{small}
\begin{equation}
\label{proof_relation_est_CSE_int3}
\begin{split}
   &( V_R^\top B_R V_R )^{-1} (V_R^\top B_R Y_R) = 
   \begin{bmatrix}
   F_1^{-1} & 
   \begin{matrix}
   \mathbf{0}_{N} & \mathbf{0}_{N} \\
   \mathbf{0}_{N} & \mathbf{0}_{N}
   \end{matrix} \\[1em]
   \begin{matrix}
   \mathbf{0}_{N} & \mathbf{0}_{N} \\
   \mathbf{0}_{N} & \mathbf{0}_{N}
   \end{matrix} &
   F_2^{-1}
   \end{bmatrix}
   \begin{bmatrix}
   \sum_{k=1}^K \sum_{ik=1}^{n_k} (B^0_{ik}+B^1_{ik})Y_{ik} \\
   \sum_{k=1}^K \sum_{ik=1}^{n_k} B^1_{ik} Y_{ik} \\
   \sum_{k=1}^K \sum_{ik=1}^{n_k} (B^0_{ik}+B^1_{ik}) X^\dagger_{ik} Y_{ik} \\
   \sum_{k=1}^K \sum_{ik=1}^{n_k} B^1_{ik}X^\dagger_{ik} Y_{ik} 
   \end{bmatrix} \\
   & = \begin{bmatrix}
   \Big(\sum_{k=1}^K \sum_{ik=1}^{n_k} \sum_{jk\neq ik}{B}_{ik,jk}(1-Z_{jk}) \Big)^{-1} \sum_{k=1}^K \sum_{ik=1}^{n_k} \sum_{jk\neq ik}{B}_{ik,jk}(1-Z_{jk}) Y_{ik}\\
   \frac{\sum_{k=1}^K \sum_{ik=1}^{n_k} \sum_{jk \neq ik} B_{ik,jk}  Z_{jk} Y_{ik}}{\sum_{k=1}^K \sum_{ik=1}^{n_k} \sum_{jk \neq ik} B_{ik,jk}  Z_{jk} } - \frac{\sum_{k=1}^K \sum_{ik=1}^{n_k} \sum_{jk \neq ik} B_{ik,jk}  (1-Z_{jk}) Y_{ik} }{\sum_{k=1}^K \sum_{ik=1}^{n_k} \sum_{jk \neq ik} B_{ik,jk}  (1-Z_{jk}) } \\
   \Big(\sum_{k=1}^K \sum_{ik=1}^{n_k} \sum_{jk\neq ik}{B}_{ik,jk}(1-Z_{jk})X^{\dagger 2}_{ik} \Big)^{-1} \sum_{k=1}^K \sum_{ik=1}^{n_k} \sum_{jk\neq ik}{B}_{ik,jk}(1-Z_{jk})X^\dagger_{ik} Y_{ik}\\
   \frac{\sum_{k=1}^K \sum_{ik=1}^{n_k} \sum_{jk \neq ik} B_{ik,jk}  Z_{jk} X^\dagger_{ik} Y_{ik}}{\sum_{k=1}^K \sum_{ik=1}^{n_k} \sum_{jk \neq ik} B_{ik,jk}  Z_{jk} X^{\dagger 2}_{ik} } - \frac{\sum_{k=1}^K \sum_{ik=1}^{n_k} \sum_{jk \neq ik} B_{ik,jk}  (1-Z_{jk}) X^\dagger_{ik} Y_{ik} }{\sum_{k=1}^K \sum_{ik=1}^{n_k} \sum_{jk \neq ik} B_{ik,jk}  (1-Z_{jk}) X^{\dagger 2}_{ik} } \\
   \end{bmatrix}
 \end{split}  
\end{equation}
\end{small}
where
\begin{small}
\begin{equation*}
F_1^{-1} =
\begin{bmatrix}
\big(\sum_{k=1}^K \sum_{ik=1}^{n_k} B^0_{ik}\big)^{-1} & -\big(\sum_{k=1}^K \sum_{ik=1}^{n_k} B^0_{ik}\big)^{-1} \\
-\big(\sum_{k=1}^K \sum_{ik=1}^{n_k} B^0_{ik}\big)^{-1} & 
\big(\sum_{k=1}^K \sum_{ik=1}^{n_k} B^0_{ik}\big)^{-1} 
+ \big(\sum_{k=1}^K \sum_{ik=1}^{n_k} B^1_{ik}\big)^{-1}
\end{bmatrix},
\end{equation*}
 \end{small}
and
\begin{small}
\begin{equation}
\label{proof_relation_est_CSE_int4}
F_2^{-1} =
\begin{bmatrix}
\big(\sum_{k=1}^K \sum_{ik=1}^{n_k} B^0_{ik} X^{\dagger 2}_{ik}\big)^{-1} & -\big(\sum_{k=1}^K \sum_{ik=1}^{n_k} B^0_{ik} X^{\dagger 2}_{ik}\big)^{-1} \\
-\big(\sum_{k=1}^K \sum_{ik=1}^{n_k} B^0_{ik} X^{\dagger 2}_{ik}\big)^{-1} & 
\big(\sum_{k=1}^K \sum_{ik=1}^{n_k} B^0_{ik} X^{\dagger 2}_{ik}\big)^{-1} 
+ \big(\sum_{k=1}^K \sum_{ik=1}^{n_k} B^1_{ik} X^{\dagger 2}_{ik}\big)^{-1}
\end{bmatrix}.
\end{equation}
\end{small}
Reordering the second and third elements in \eqref{proof_relation_est_CSE_int3} yields $\hat{\beta}_{R,3}(\alpha) \;=\; \hat{\tau}_{hj}(\alpha)$
since the normalization factor $S_N$ cancels between numerator and denominator in the expression for $\hat{\beta}_{R,3}(\alpha)$.
\end{proof}

\begin{proof}[Proof of Proposition \ref{beta_r_to_beta_p}]
We first derive the explicit form of ${\boldsymbol{\beta}}^{r}_{A}(\alpha)$ for $A \in \{D,R,S\}$. 
From Proposition \ref{relation_est_CSE}, each component of $V_D^\top B_D V_D$ coincides with that of $V_S^\top B_S V_S$, and likewise for $V_D^\top B_D Y_D$ and $V_S^\top B_S Y_S$. 
Hence it suffices to consider ${\boldsymbol{\beta}}^r_D(\alpha)$. 
Using the notation from Proposition \ref{relation_est_CSE}, we obtain
\begin{equation}
\label{proof_Const_ratio_estimable_quantity_int1}
\begin{split}
&\mathbb{E}(V^T_{D,2} B_D V_{D,4})= \sum_{k=1}^K \sum_{ik=1}^{n_k} \sum_{jk \neq ik}  
        \mathbb{E}[\tilde{X}_{jk} \, B_{ik,jk} \, (Z_{jk}\tilde{X}_{jk})^*] \\
        & =\sum_{k=1}^K \sum_{ik=1}^{n_k} \sum_{jk \neq ik}  
        \mathbb{E} [\tilde{X}_{jk} \, B_{ik,jk} \, (Z_{jk}\tilde{X}_{jk}-\frac{1}{2}\tilde{X}_{jk}) ]= \sum_{k=1}^K \sum_{ik=1}^{n_k} \sum_{jk \neq ik}  
        S_{ik,jk} (\tilde{X}^2_{jk} -\tilde{X}^2_{jk} )=0
        \end{split}
\end{equation}
by the de-correlation induced by $\Lambda$. 
Similarly, $\mathbb{E}(V_{D,a}^\top B_D V_{D,b})=0$ for $a \in \{1,2\}$ and $b \in \{3,4\}$. 
Furthermore,
\begin{equation*}
\begin{split}
&\mathbb{E}(V^T_{D,1} B_D V_{D,2})= \sum_{k=1}^K \sum_{ik=1}^{n_k} \sum_{jk \neq ik}  
        \mathbb{E}[1 \, B_{ik,jk} \, \tilde{X}_{jk}]  =\sum_{k=1}^K \sum_{ik=1}^{n_k} \sum_{jk \neq ik} 2S^r_{ik,jk} (X_{jk} -\bar{X} )=_{(1)}0\\
        \end{split}
\end{equation*}
where $(1)$ follows from $\sum_{k=1}^K \sum_{ik=1}^{n_k} \sum_{jk \neq ik} S^r_{ik,jk}=1$ and the definition of $\bar{X}$ in Definition~\ref{struc_APO}.  
Similarly, $\mathbb{E}(V_{D,a}^\top B_D V_{D,b})=0$ for $(a,b)\in \{(2,1),(3,4),(4,3)\}$.  
Hence $\mathbb{E}(V_D^\top B_D V_D)$ is diagonal with
\begin{small}
\begin{equation}
\label{proof_Const_ratio_estimable_quantity_int3}
    \begin{split}
      &\mathrm{diag}(\mathbb{E}(V^T_{D} B_D V_{D})) =\mathbb{E}
\begin{bmatrix}
 \sum_{k=1}^{K}  \sum_{\substack{ik=1\\jk \neq ik}}^{n_k}
B_{ik,jk}  \\ \sum_{k=1}^K \sum_{\substack{ik=1 \\ jk \neq ik}}^{n_k} \tilde{X}^2_{jk} B_{ik,jk} Y_{ik}  \\ \sum_{k=1}^K \sum_{\substack{ik=1 \\ jk \neq ik}}^{n_k} Z^{* 2}_{jk} B_{ik,jk} Y_{ik} \\ \sum_{k=1}^K \sum_{\substack{ik=1 \\ jk \neq ik}}^{n_k}(Z_{jk}\tilde{X}_{jk})^{* 2} B_{ik,jk} Y_{ik} \\ 
\end{bmatrix}=
\begin{bmatrix}
2\sum_{k=1}^{K}  \sum_{\substack{ik=1\\jk \neq ik}}^{n_k} S^r_{ik,jk}  \\ 2 \sum_{k=1}^{K}\sum_{\substack{ik=1\\jk \neq i_k}}^{n_k} S^r_{ik,jk} \tilde{X}^2_{jk}  \\ \frac{1}{2}\sum_{k=1}^{K}\sum_{\substack{ik=1\\jk \neq i_k}}^{n_k} S^r_{ik,jk} \\ \frac{1}{2}\sum_{k=1}^{K}\sum_{\substack{ik=1\\jk \neq i_k}}^{n_k} S^r_{ik,jk} \tilde{X}^2_{jk}
\end{bmatrix}. 
    \end{split}
\end{equation}
\end{small}
Moreover,
\begin{small}
\begin{equation*}
    \begin{split}
      &\mathbb{E}(V^T_{D} B_D Y_{D}) \\
&=\mathbb{E}
\begin{bmatrix}
 \sum_{k=1}^{K}  \sum_{\substack{ik=1\\jk \neq ik}}^{n_k}
B_{ik,jk} Y_{ik} \\ \sum_{k=1}^K \sum_{\substack{ik=1 \\ jk \neq ik}}^{n_k} \tilde{X}_{jk} B_{ik,jk} Y_{ik}  \\ \sum_{k=1}^K \sum_{\substack{ik=1 \\ jk \neq ik}}^{n_k} Z^{*}_{jk} B_{ik,jk} Y_{ik} \\ \sum_{k=1}^K \sum_{\substack{ik=1 \\ jk \neq ik}}^{n_k}(Z_{jk}\tilde{X}_{jk})^* B_{ik,jk} Y_{ik} \\ 
\end{bmatrix}=
\begin{bmatrix}
\sum_{k=1}^{K}  \sum_{\substack{ik=1\\jk \neq ik}}^{n_k} S^r_{ik,jk} (\bar{Y}_{ik}(1,\alpha)+ \bar{Y}_{ik}(0,\alpha)) \\ \sum_{k=1}^{K}\sum_{\substack{ik=1\\jk \neq i_k}}^{n_k} S^r_{ik,jk} \tilde{X}_{jk} (\bar{Y}_{ik}(1,\alpha)+ \bar{Y}_{ik}(0,\alpha)) \\ \sum_{k=1}^{K}\sum_{\substack{ik=1\\jk \neq i_k}}^{n_k} \frac{1}{2}S^r_{ik,jk}  (\bar{Y}_{ik}(1,\alpha)- \bar{Y}_{ik}(0,\alpha)) \\ \sum_{k=1}^{K}\sum_{\substack{ik=1\\jk \neq i_k}}^{n_k} \frac{1}{2}S^r_{ik,jk} \tilde{X}_{jk} (\bar{Y}_{ik}(1,\alpha)- \bar{Y}_{ik}(0,\alpha))
\end{bmatrix}. 
    \end{split}
\end{equation*}
\end{small}

Therefore,
\begin{small}
\begin{equation}
\label{proof_Const_ratio_estimable_quantity_int1}
   \begin{split}
&  \boldsymbol{\beta}^r_{D}(\alpha)= (\mathbb{E}(V^T_{D} B_D V_{D})^{-1}(\mathbb{E}(V^T_{D} B_D Y_{D})) \\  & = \begin{bmatrix}
(2\sum_{k=1}^{K}  \sum_{\substack{ik=1\\jk \neq ik}}^{n_k} S_{ik,jk})^{-1} \sum_{k=1}^{K}  \sum_{\substack{ik=1\\jk \neq ik}}^{n_k} S_{ik,jk} (\bar{Y}_{ik}(1,\alpha)+ \bar{Y}_{ik}(0,\alpha)) \\  (2\sum_{k=1}^{K}\sum_{\substack{ik=1\\jk \neq i_k}}^{n_k} S_{ik,jk} \tilde{X}^2_{jk} )^{-1}\sum_{k=1}^{K}\sum_{\substack{ik=1\\jk \neq i_k}}^{n_k} S_{ik,jk} \tilde{X}_{jk} (\bar{Y}_{ik}(1,\alpha)+ \bar{Y}_{ik}(0,\alpha)) \\ (\sum_{k=1}^{K}\sum_{\substack{ik=1\\jk \neq i_k}}^{n_k} S_{ik,jk})^{-1}\sum_{k=1}^{K}\sum_{\substack{ik=1\\jk \neq i_k}}^{n_k} S_{ik,jk}  (\bar{Y}_{ik}(1,\alpha)- \bar{Y}_{ik}(0,\alpha)) \\ (\sum_{k=1}^{K}\sum_{\substack{ik=1\\jk \neq i_k}}^{n_k} S_{ik,jk} \tilde{X}^2_{jk})^{-1}\sum_{k=1}^{K}\sum_{\substack{ik=1\\jk \neq i_k}}^{n_k} S_{ik,jk} \tilde{X}_{jk} (\bar{Y}_{ik}(1,\alpha)- \bar{Y}_{ik}(0,\alpha))
\end{bmatrix}. 
    \end{split}
\end{equation}
\end{small}
Inherited notation in \eqref{proof_relation_est_CSE_int3} and the proof of Proposition \ref{relation_est_CSE}, a direct calculation gives  
\[
F_{1,11}=\sum_{k=1}^K \sum_{ik=1}^{n_k} \sum_{jk \neq ik} B_{ik,jk}, 
\qquad 
F_{1,12}=F_{1,21}=F_{1,22}=\sum_{k=1}^K \sum_{ik=1}^{n_k} \sum_{jk \neq ik} B_{ik,jk}Z_{jk}.
\]
Similarly,  
\[
F_{2,11}=\sum_{k=1}^K \sum_{ik=1}^{n_k} \sum_{jk \neq ik} X^{\dagger 2}_{ik} B_{ik,jk}, 
\qquad 
F_{2,12}=F_{2,21}=F_{2,22}=\sum_{k=1}^K \sum_{ik=1}^{n_k} \sum_{jk \neq ik} X^{\dagger 2}_{ik} B_{ik,jk}Z_{jk}.
\]
Hence, based on \eqref{proof_relation_est_CSE_int3}, the ratio quantity from receiver perspective is  
\begin{small}
\begin{equation}
\label{proof_Const_ratio_estimable_quantity_int2}
\begin{split}
   & \boldsymbol{\beta}^r_R(\alpha)=[\mathbb{E}( V_R^\top B_R V_R )]^{-1} \mathbb{E}(V_R^\top B_R Y_R) \\
   & = 
   \begin{bmatrix}
   [\mathbb{E}(F_1)]^{-1} & 
   \begin{matrix}
   \mathbf{0}_{N} & \mathbf{0}_{N} \\
   \mathbf{0}_{N} & \mathbf{0}_{N}
   \end{matrix} \\[1em]
   \begin{matrix}
   \mathbf{0}_{N} & \mathbf{0}_{N} \\
   \mathbf{0}_{N} & \mathbf{0}_{N}
   \end{matrix} &
   [\mathbb{E}(F_2)]^{-1}
   \end{bmatrix} \cdot \begin{bmatrix}
   \sum_{k=1}^K \sum_{ik=1}^{n_k} S^r_{ik,jk}(Y_{ik}(0,\alpha)+Y_{ik}(1,\alpha)) \\
   \sum_{k=1}^K \sum_{ik=1}^{n_k} S^r_{ik,jk}Y_{ik}(1,\alpha) \\
   \sum_{k=1}^K \sum_{ik=1}^{n_k} S^r_{ik,jk}X^\dagger_{ik} (Y_{ik}(0,\alpha)+Y_{ik}(1,\alpha)) \\
    \sum_{k=1}^K \sum_{ik=1}^{n_k} S^r_{ik,jk}X^\dagger_{ik} Y_{ik}(1,\alpha) \\
   \end{bmatrix} \\
   & = \begin{bmatrix}
   \Big(\sum_{k=1}^K \sum_{ik=1}^{n_k} \sum_{jk\neq ik}S_{ik,jk}\Big)^{-1} \sum_{k=1}^K \sum_{ik=1}^{n_k} \sum_{jk\neq ik}S_{ik,jk} Y_{ik}(0, \alpha )\\
   \Big(\sum_{k=1}^K \sum_{ik=1}^{n_k} \sum_{jk\neq ik}S_{ik,jk}\Big)^{-1} \sum_{k=1}^K \sum_{ik=1}^{n_k} \sum_{jk\neq ik}S_{ik,jk}(Y_{ik}(1, \alpha )-Y_{ik}(0, \alpha ) )\\
   \Big(\sum_{k=1}^K \sum_{ik=1}^{n_k} \sum_{jk\neq ik}S_{ik,jk}X^{\dagger 2}_{ik} \Big)^{-1} \sum_{k=1}^K \sum_{ik=1}^{n_k} \sum_{jk\neq ik} X^{\dagger}_{ik}S_{ik,jk} Y_{ik}(0, \alpha )\\
   \Big(\sum_{k=1}^K \sum_{ik=1}^{n_k} \sum_{jk\neq ik}S_{ik,jk}X^{\dagger 2}_{ik}\Big)^{-1} \sum_{k=1}^K \sum_{ik=1}^{n_k} \sum_{jk\neq ik}S_{ik,jk}X^{\dagger}_{ik}(Y_{ik}(1, \alpha )-Y_{ik}(0, \alpha ) )\\
   \end{bmatrix}
 \end{split}  
\end{equation}
\end{small}
where \begin{small} \[[\mathbb{E}(F_1)]^{-1}=(\sum_{k=1}^K \sum_{\substack{ik=1\\jk \neq ik}}^{n_k} S^r_{ik,jk})^{-1}  \begin{pmatrix}
   1 & -1 \\
  -1 & 2
   \end{pmatrix}, \quad [\mathbb{E}(F_2)]^{-1}=(\sum_{k=1}^K \sum_{\substack{ik=1\\jk \neq ik}}^{n_k} S^r_{ik,jk}X^{\dagger 2}_{ik} )^{-1}  \begin{pmatrix}
   1 & -1\\
   -1 & 2
   \end{pmatrix} \] \end{small}
Reordering the second and third components of $\boldsymbol{\beta}^r_R(\alpha)$ yields $\beta^r_{D,3}(\alpha)=\beta^r_{S,3}(\alpha)=\beta^r_{R,3}(\alpha)$, $ \beta^r_{D,4}(\alpha)=\beta^r_{S,4}(\alpha)$, while $\beta^r_{R,4}(\alpha)$ corresponds to a differently weighted average of the pairwise spillover effects.

We now turn to the proof of Proposition~\ref{beta_r_to_beta_p}.  
We begin by considering $\beta^r_{A,3}(\alpha)$ for $A \in \{D, S, R\}$.  
From \eqref{proof_Const_ratio_estimable_quantity_int1} and the reparametrized structural model in Definition~\ref{struc_APO},  
\begin{equation*}
    \begin{split}
      & \beta^r_{A,3}(\alpha)= (\sum_{k=1}^K \sum_{ik=1}^{n_k} \sum_{jk \neq ik} S^r_{ik,jk})^{-1}  \sum_{k=1}^K \sum_{ik=1}^{n_k} \sum_{jk \neq ik} S^r_{ik,jk} \left[ \beta_{3,ijk}(\alpha)+ \beta_{4,ijk}(\alpha)(X_{jk}-\bar{X})   \right] \\
      & = \sum_{k=1}^K \sum_{ik=1}^{n_k} \sum_{jk \neq ik} S^r_{ik,jk} \beta_{3,ijk}(\alpha)+ \sum_{a=1}^{m_4} \sum_{k=1}^K \sum_{ik=1}^{n_k} \sum_{jk \neq ik} \beta_{4,a}(\alpha) S^r_{ik,jk} (X_{jk}-\bar{X}) \mathbf{1}\{ \beta_{4,ijk}(\alpha)=\beta_{4,h}(\alpha) \}. \\
    \end{split}
\end{equation*}
Hence,
      $$ S_{N} \left\lvert \beta^r_{A,3}(\alpha)- \sum_{k=1}^K \sum_{ik=1}^{n_k} \sum_{jk \neq ik} S^r_{ik,jk} \beta_{3,ijk}(\alpha) \right\rvert \overset{N\rightarrow \infty}{\longrightarrow} 0,$$
by Statement 1 of Assumption \ref{coef_hete_restrict} with $h=4$ and the fact that $\beta_{4,ijk}(\alpha)=\theta_{4,ijk}(\alpha)$. Next consider $\beta^r_{A,4}(\alpha)$ for $A \in \{D,S\}$.  
From \eqref{proof_Const_ratio_estimable_quantity_int1} and Definition~\ref{struc_APO},
\begin{equation}
\label{proof_beta_r_to_beta_p_int1}
\beta^r_{A,4}(\alpha) 
= (\sum_{k=1}^K \sum_{ik=1}^{n_k} \sum_{jk \neq ik} 
   S^r_{ik,jk}\,\tilde{X}_{jk}^2 )^{-1} 
   \sum_{k=1}^K \sum_{ik=1}^{n_k} \sum_{jk \neq ik} 
   S^r_{ik,jk}\big[ \beta_{3,ijk}(\alpha)\tilde{X}_{jk} + \beta_{4,ijk}(\alpha)\tilde{X}_{jk}^2 \big].
\end{equation}
The first term equals
\begin{small}
\begin{equation}
\label{proof_beta_r_to_beta_p_int2}
    \begin{split}
      & \sum_{k=1}^K \sum_{ik=1}^{n_k} \sum_{jk \neq ik} S^r_{ik,jk}  \beta_{3,ijk}(\alpha)\tilde{X}_{jk} = \sum_{k=1}^K \sum_{ik=1}^{n_k} \sum_{jk \neq ik} S^r_{ik,jk}  (\theta_{3,ijk}(\alpha)+\theta_{4,ijk}(\alpha)\bar{X})( {X}_{jk}-\bar{X} ) \\
      & =  \sum_{k=1}^K \sum_{ik=1}^{n_k} \sum_{jk \neq ik} S^r_{ik,jk}  \theta_{3,ijk}(\alpha)( {X}_{jk}-\bar{X} )+ \bar{X} \sum_{k=1}^K \sum_{ik=1}^{n_k} \sum_{jk \neq ik} S^r_{ik,jk} \theta_{4,ijk}(\alpha)( {X}_{jk}-\bar{X} )\\
      &=_{(1)}\sum_{a=1}^{m_3} \sum_{k=1}^K \sum_{ik=1}^{n_k} \sum_{jk \neq ik} \theta_{3,a}(\alpha) S^r_{ik,jk} (X_{jk}-\bar{X}) 1\{ \theta_{3,ijk}(\alpha)=\theta_{3,a} (\alpha)\}\\
      & + \bar{X} \sum_{a=1}^{m_4} \sum_{k=1}^K \sum_{ik=1}^{n_k} \sum_{jk \neq ik} \theta_{4,a}(\alpha)S^r_{ik,jk} (X_{jk}-\bar{X}) 1\{ \theta_{4,ijk}(\alpha)=\theta_{4,a}(\alpha) \} \overset{N\rightarrow \infty}{\longrightarrow}_{(2)} 0
    \end{split}
\end{equation}
\end{small}
$(1)$ and $(2)$ are by statement $1$ in Assumption \ref{coef_hete_restrict}. Combining \eqref{proof_beta_r_to_beta_p_int1} and \eqref{proof_beta_r_to_beta_p_int2} yields
\begin{equation*}
    \begin{split}
     S_{N} \left\lvert  \beta^r_{A,4}(\alpha)- (\sum_{k=1}^K \sum_{ik=1}^{n_k} \sum_{jk \neq ik} S^r_{ik,jk} \tilde{X}^2_{jk} )^{-1}  \sum_{k=1}^K \sum_{ik=1}^{n_k} \sum_{jk \neq ik} S^r_{ik,jk}   \beta_{4,ijk}(\alpha)\tilde{X}^2_{jk}    \right\rvert \overset{N\rightarrow\infty}{\longrightarrow} 0. 
    \end{split}
\end{equation*}
Finally consider $\beta^r_{R,4}(\alpha)$.  
From \eqref{proof_Const_ratio_estimable_quantity_int2} and Definition~\ref{struc_APO},  
\begin{equation}
\label{proof_beta_r_to_beta_p_int3}
    \begin{split}
      &   \beta^r_{R,4}(\alpha)= \Big(\sum_{k=1}^K \sum_{ik=1}^{n_k} \sum_{jk\neq ik}S^r_{ik,jk}X^{\dagger 2}_{ik}\Big)^{-1} \sum_{k=1}^K \sum_{ik=1}^{n_k} \sum_{jk\neq ik}S^r_{ik,jk}X^{\dagger}_{ik}\left( \theta_{3,ijk}(\alpha) + \theta_{4,ijk}(\alpha)X_{jk} \right).\\
    \end{split}
\end{equation}
The second term satisfies
\begin{equation}
\label{proof_beta_r_to_beta_p_int4}
    \begin{split}
      & \sum_{k=1}^K \sum_{ik=1}^{n_k} \sum_{jk\neq ik}S^r_{ik,jk}X^{\dagger}_{ik}\theta_{4,ijk}(\alpha)X_{jk}=_{(1)} \sum_{k=1}^K \sum_{ik=1}^{n_k} S^r_{ik} X^\dagger_{ik} \theta_{4,ik} (\alpha) \left[\sum_{jk \neq ik} S^r_{jk|ik} X_{jk}-\bar{X}+\bar{X}\right] \\
      & =_{(2)} \sum_{k=1}^K \sum_{ik=1}^{n_k} S^r_{ik} X^{\dagger 2}_{ik} \theta_{4,ik}(\alpha) +\bar{X} \sum_{k=1}^K \sum_{ik=1}^{n_k} S^r_{ik} X^{\dagger}_{ik} \theta_{4,ik} (\alpha)
    \end{split}
\end{equation}
where $(1)$ uses Statement (2b) in Assumption  \ref{coef_hete_restrict}, and $(2)$ uses the definition of $X_{ik}^\dagger$.  
Substituting \eqref{proof_beta_r_to_beta_p_int4} into \eqref{proof_beta_r_to_beta_p_int3} gives
\begin{small}
\begin{equation}
\label{proof_beta_r_to_beta_p_int5}
\begin{split}
  &   \beta^r_{R,4}(\alpha)=_{(1)} \Big(\sum_{k=1}^K \sum_{ik=1}^{n_k} S^r_{ik} X^{\dagger 2}_{ik}\Big)^{-1} \sum_{k=1}^K \sum_{ik=1}^{n_k}  S^r_{ik} \left(X^{\dagger}_{ik}\theta_{3ik}(\alpha) +  X^{\dagger 2}_{ik} \theta_{4ik}(\alpha)+ \bar{X} \cdot X^{\dagger}_{ik} \theta_{4ik}(\alpha) \right)\\
  &= \Big(\sum_{k=1}^K \sum_{ik=1}^{n_k} S^r_{ik} X^{\dagger 2}_{ik}\Big)^{-1} \sum_{k=1}^K \sum_{ik=1}^{n_k}  S^r_{ik}\theta_{4ik}(\alpha) X^{\dagger 2}_{ik} \\
  & + \Big(\sum_{k=1}^K \sum_{ik=1}^{n_k} S^r_{ik} X^{\dagger 2}_{ik}\Big)^{-1} \sum_{k=1}^K \sum_{ik=1}^{n_k} \left( \sum_{a=1}^{u_3} S^r_{ik} \theta_{3,a}(\alpha) (X^{\dagger 0}_{ik}-\bar{X})+ \bar{X} \sum_{a=1}^{u_4} S^r_{ik} \theta_{4,a}(\alpha) (X^{\dagger 0}_{ik}-\bar{X})  \right) \\
    \end{split}
\end{equation}
\end{small}
$(1)$ is by Statement $(2a)$ and $(2b)$ in Assumption \ref{coef_hete_restrict}. Hence
\begin{equation}
\label{proof_beta_r_to_beta_p_int6}
    \begin{split}
   S_{N}    \Bigg\lvert \beta^r_{R,4}(\alpha) - 
\Bigg(\sum_{k=1}^K \sum_{i=1}^{n_k} S^r_{ik}X_{ik}^{\dagger 2}\Bigg)^{-1}
   \sum_{k=1}^K \sum_{i=1}^{n_k} S^r_{ik}\theta_{4,ik}(\alpha)X_{ik}^{\dagger 2}
\Bigg\rvert \;\overset{N\to\infty}{\longrightarrow}\; 0. 
    \end{split}
\end{equation}
\eqref{proof_beta_r_to_beta_p_int6} is based on Statement (2c) in Assumption \ref{coef_hete_restrict}. This establishes the result.  
\end{proof}

\begin{proof}[Proof of Lemma \ref{settings_all_conditions_satisfy}]
\label{proof_settings_all_conditions_satisfy} 
We begin with Setting 1 in Lemma \ref{settings_all_conditions_satisfy}.  
From the definition of $S_{ik,jk}$ in Example \ref{Examp_Ave_outward}, it follows immediately that 
$\sum_{k=1}^K \sum_{ik=1}^{n_k} \sum_{jk \neq ik} S_{ik,jk}=1$, so that $S_{ik,jk}=S^r_{ik,jk}$. Next, we verify Assumption \ref{coef_hete_restrict} under Setting 1.  
Since $\theta_{a,ijk}(\alpha)$ is homogeneous for each $a \in \{3,4\}$, we obtain
\[
\sum_{k=1}^K \sum_{ik=1}^{n_k} \sum_{jk \neq ik} 
        S^r_{ik,jk} (X_{jk}-\bar{X}) \,
        \mathbf{1}\{\theta_{a,ijk}(\alpha)=\theta_{a}(\alpha)\} 
    =_{(1)} \bar{X}-\bar{X}=0,
\]
where $(1)$ follows from the definition of $\bar{x}$. For statement 2 of Assumption \ref{coef_hete_restrict}, note that
\[
S^r_{jk|ik}=S_{jk|ik}
=\frac{1}{|\mathcal{N}^{\mathrm{out}}_{jk}|}\, 
\mathbf{1}\{ik \in \mathcal{N}^{\mathrm{out}}_{jk}\}
\mathbf{1}\{jk \in \mathcal{N}^{\mathrm{out}}_{k}\},
\quad
S^r_{ik}=S_{ik}=\frac{1}{N^{\mathrm{out}}},
\]
and therefore we have 
\begin{small}
\begin{equation}
\label{proof_settings_all_conditions_satisfy_int2}
    \begin{split}
       & \sum_{jk\neq ik} S^r_{jk|ik}=  \sum_{jk\neq ik} \frac{1}{|\mathcal{N}^{out}_{jk}|} 1 \{ik \in \mathcal{N}^{out}_{jk}\} 1\{jk \in \mathcal{N}^{out}_{k}\}=_{(1)} \frac{1}{d} \sum_{jk \neq ik} 1 \{ik \in \mathcal{N}^{out}_{jk}\}=_{(2)} \frac{1}{d} |\mathcal{N}^{in}_{ik}|=1
    \end{split}
\end{equation}
\end{small}
$(1)$ and $(2)$ both follow from the assumption of the regular directed graph with $d>0$.  
In particular, $1\{jk \in \mathcal{N}^{\mathrm{out}}_{k}\}=1$ and 
$|\mathcal{N}^{\mathrm{in}}_{ik}|=d$, where $\mathcal{N}^{\mathrm{in}}_{ik}$ is defined in Example \ref{Examp_Ave_inward}.  
Condition (b) in statement 2 then holds trivially.  
For condition (c), observe that for $a \in \{3,4\}$, 
\begin{small}
\begin{equation}
\label{proof_settings_all_conditions_satisfy_int4}
    \begin{split}
   & \sum_{k=1}^K \sum_{ik=1}^{n_k} 
        S^r_{ik} (X^{\dagger 0}_{ik}-\bar{X}) \, 
        \mathbf{1}\{\theta_{a,ik}(\alpha)=\theta_{a}(\alpha)\} = \sum_{k=1}^K \sum_{ik=1}^{n_k} S^{r}_{ik} 
        \Big( \sum_{jk\neq ik}  S^{r}_{jk|ik} X_{jk}\Big) 
        - \sum_{k=1}^K \sum_{ik=1}^{n_k} S^r_{ik} \bar{X} 
   =_{(1)} 0,
    \end{split}
\end{equation}
\end{small}
where $(1)$ holds because the first term equals $\bar{X}$ and 
$\sum_{k=1}^K \sum_{ik=1}^{n_k} \tfrac{1}{N^{\mathrm{out}}}=1$ since every unit has at least one out-neighbor in the regular directed graph with $d>0$. We now verify Assumption \ref{assump:indep_tilde_beta_bar_beta_x}.  
For $h \in \{3,4\}$,  
\[
\theta^p_h(\alpha)-\theta_{h}(\alpha,x)
= \theta_h \left[
        \sum_{k=1}^K \sum_{ik=1}^{n_k} \sum_{jk \neq ik} S^r_{ik,jk}
        - \sum_{k=1}^K \sum_{ik=1}^{n_k} \sum_{jk \neq ik} 
S^r_{ik,jk}(x) \right] = 0,
\]
since both $S^r_{ik,jk}$ and $S^r_{ik,jk}(x)$ sum to one (see Definition \ref{cond_spill_over_under_struc_model}). Finally, from the definitions of $\theta^p_{A,4}(\alpha)$ for $A \in \{D,S,R\}$ in Assumption \ref{assump:indep_tilde_beta_bar_beta_x}, and given that $\theta_{4,ijk}(\alpha)$ is homogeneous, it follows directly that $\bigl\lvert \theta^p_{A,4}(\alpha)-\theta_{4}(\alpha,x) \bigr\rvert
=\lvert \theta_4-\theta_4 \rvert=0$, where $A \in \{D,S,R\}$.

We now turn to the equivalence among $\beta^r_{A,h}(\alpha)$ for $A\in \{D,S,R\}$ and 
$h\in \{3,4\}$. 
Proposition 
\ref{Const_ratio_estimable_quantity} establishes that 
$\beta^{r}_{D,3}(\alpha)=\beta^{r}_{S,3}(\alpha)=\beta^{r}_{R,3}(\alpha)$. To establish equivalence among $\beta^r_{A,4}(\alpha)$ for $A \in \{D,S,R\}$, recall the formula in \eqref{proof_beta_r_to_beta_p_int1} for $A \in \{D,S\}$.  
With homogeneous coefficients,  
\begin{equation}
\begin{split}
\label{proof_settings_all_conditions_satisfy_int1}
&\beta^r_{A,4} (\alpha)
= (\sum_{k=1}^K \sum_{ik=1}^{n_k} \sum_{jk \neq ik} 
   S^r_{ik,jk}\,\tilde{X}_{jk}^2 )^{-1} 
   \sum_{k=1}^K \sum_{ik=1}^{n_k} \sum_{jk \neq ik} 
   S^r_{ik,jk}\big[ \beta_3(\alpha) \tilde{X}_{jk} + \theta_{4}(\alpha)\tilde{X}_{jk}^2 \big] \\
   & =_{(1)} (\sum_{k=1}^K \sum_{ik=1}^{n_k} \sum_{jk \neq ik} 
   S^r_{ik,jk}\,\tilde{X}_{jk}^2 )^{-1} 
   \big[ 0 + \theta_{4}(\alpha) \sum_{k=1}^K \sum_{ik=1}^{n_k} \sum_{jk \neq ik} 
   S^r_{ik,jk}(\alpha)\tilde{X}_{jk}^2 \big] =\theta_4(\alpha)
   \end{split}
\end{equation}
where $\beta_3(\alpha)=\theta_{3}(\alpha)+ \theta_4(\alpha)\bar{X}$ 
(by homogeneity), and $(1)$ uses the fact that 
$\sum_{k=1}^K \sum_{ik=1}^{n_k}$ $\sum_{jk \neq ik} S^r_{ik,jk}\tilde{X}_{jk}=0$. Next, recall $\beta^r_{R,4}(\alpha)$ from 
\eqref{proof_beta_r_to_beta_p_int5}.  
Under statement 2 of Assumption \ref{coef_hete_restrict},  
\begin{small}
\begin{equation}
\label{proof_settings_all_conditions_satisfy_int3}
\begin{split}
  &   \beta^r_{R,4}(\alpha)=_{(1)} \Big(\sum_{k=1}^K \sum_{ik=1}^{n_k} S^r_{ik} X^{\dagger 2}_{ik}\Big)^{-1} \sum_{k=1}^K \sum_{ik=1}^{n_k}  S^r_{ik} \left(X^{\dagger}_{ik}\theta_{3}(\alpha) +  X^{\dagger 2}_{ik} \theta_{4}(\alpha)+ \bar{X} \cdot X^{\dagger}_{ik} \theta_{4}(\alpha) \right)\\
  &=  \Big(\sum_{k=1}^K \sum_{ik=1}^{n_k} S^r_{ik} X^{\dagger 2}_{ik}\Big)^{-1} \left[(\theta_{3}(\alpha)+\theta_4(\alpha))  \sum_{k=1}^K \sum_{ik=1}^{n_k}  S^r_{ik} X^{\dagger}_{ik}+ \theta_4(\alpha)  \sum_{k=1}^K \sum_{ik=1}^{n_k}  S^r_{ik} X^{\dagger 2}_{ik}  \right] = 0+ \theta_{4}(\alpha) \\
    \end{split}
\end{equation}
\end{small}
where $(1)$ applies statement 2 of Assumption \ref{coef_hete_restrict}. From \eqref{proof_settings_all_conditions_satisfy_int1} and 
\eqref{proof_settings_all_conditions_satisfy_int3}, it follows that $\beta^{r}_{D,4}(\alpha)=\beta^{r}_{S,4}(\alpha)=\beta^{r}_{R,4}(\alpha)$. 

Turning to the population quantities, Proposition \ref{beta_r_to_beta_p} and 
statement 1 of Assumption \ref{coef_hete_restrict} imply 
$\beta^p_{D,3}(\alpha)=\beta^p_{S,3}(\alpha)=\beta^p_{R,3}(\alpha)$.  
Moreover, from the form of $\beta^p_{A,4}(\alpha)$ in Proposition 
\ref{beta_r_to_beta_p}, and using homogeneity of $\theta_{4ijk}(\alpha)$, it follows directly that $\beta^p_{D,4}(\alpha)=\beta^p_{S,4}(\alpha)=\beta^p_{R,4}(\alpha).$

Now consider setting 2. By the definition of $S_{ik,jk}$ in Example \ref{Examp_Ave_inward}, it is straightforward to verify that 
$\sum_{k=1}^K \sum_{ik=1}^{n_k} \sum_{jk \neq ik} S_{ik,jk}=1$, and hence $S_{ik,jk}=S^r_{ik,jk}$. For statements 1 and 2(a) in Assumption \ref{coef_hete_restrict}, the argument follows directly from the proof under setting $1$ in Lemma \ref{settings_all_conditions_satisfy}.  
For statement 2(b), note that $S^r_{jk|ik}=S_{jk|ik}=\sum_{jk \neq ik} \frac{\mathbf{1}\{jk \in \mathcal{N}^{in}_{ik}\}}{|\mathcal{N}^{in}_{ik}|}=1$. For statement 2(c), observe that $S^{r}_{ik}=S_{ik}=(N^{in})^{-1}\mathbf{1}\{ik \in \mathcal{N}^{in}_{k}\} $ and $\sum_{k=1}^K \sum_{ik=1}^{n_k}S^r_{ik}=1$, so the proof follows the same argument as in \eqref{proof_settings_all_conditions_satisfy_int2} under setting 1. Verification of Assumption \ref{assump:indep_tilde_beta_bar_beta_x} is also identical to setting 1, since the proof relies only on the homogeneity of $\theta_{h\cdot}$, which holds in both settings. Finally, the equivalence 
\[
\beta^h_{D l}(\alpha)=\beta^h_{S l}(\alpha)=\beta^h_{R l}(\alpha), 
\quad h \in \{r,p\}, \; l \in \{3,4\},
\]
follows the same reasoning as in setting 1: given Assumptions \ref{coef_hete_restrict} and \ref{assump:indep_tilde_beta_bar_beta_x}, the proof depends only on coefficient homogeneity.  
\end{proof}

\subsection{Inference for estimators of CSE}
\label{Appendix:inf_CSE}

\begin{proof}[Proof of Proposition~\ref{Const_ratio_estimable_quantity}]
From the expressions of $V_A^{\top} B_A V_A$ and $V_A^{\top} B_A Y_A$ for 
$A \in \{D, R, S\}$ in the proof of Proposition~\ref{relation_est_CSE}, 
each component can be written as \\
$\sum_{k=1}^K \sum_{ik=1}^{n_k} \sum_{jk \neq ik} S_{ik,jk} R_{ik,jk}$, 
where $\{R_{ik,jk}\}_{ik,jk}$ are bounded random variables under 
Assumptions~\ref{unif_bound_weight}, \ref{unif_bound_pot_out}, 
and~\ref{bounded_x}.  
Following the same argument as in the proof of Proposition~\ref{consist_ASE}, 
we then obtain that 
$\bigl\lvert \hat{\boldsymbol{\beta}}_A(\alpha) 
- \boldsymbol{\beta}^r_A(\alpha) \bigr\rvert \to 0$ 
for $A \in \{D, R, S\}$.
\end{proof}

\begin{proof}[Proof of Proposition~\ref{CSE_cons_ratio_population_respectively}]
\label{proof:CSE_cons_ratio_population_respectively}
Based on the definitions of 
$\theta_{h}(\alpha,x)$ in Assumption~\ref{assump:indep_tilde_beta_bar_beta_x} 
and $\bar{\beta}_{h}(\alpha,x)$ for $h \in \{3,4\}$ in 
Definition~\ref{cond_spill_over_under_struc_model}, for $h = 3$ and $A \in \{D, S, R\}$, we have
\begin{equation}
\label{consistency_hat_tau_x_cond_tau_x_int1}
\begin{split}
S_{N}\bigl\lvert \beta^p_{A,3}(\alpha)-\bar{\beta}_{3}(\alpha,x) \bigr\rvert
&= S_{N}\Bigg\lvert 
   \sum_{k=1}^K \sum_{ik=1}^{n_k} \sum_{jk \neq ik} 
   S^{-1} S_{ik,jk}\,\beta_{3,ijk}(\alpha) - \bar{\beta}_{3}(\alpha,x)
   \Bigg\rvert  \\
&=_{(1)} S_{N}\bigl\lvert 
   \theta^p_{A,3}(\alpha)+ \theta^p_{A,4}(\alpha)\bar{X} 
   - \bigl(\theta_{3}(\alpha,x)+\theta_{4}(\alpha,x)\bar{X}\bigr)
   \bigr\rvert \overset{N\rightarrow \infty}{\longrightarrow}_{(2)} 0,
\end{split}
\end{equation}
where (1) uses $S^{-1}S_{ik,jk} = S^r_{ik,jk}$, the identity 
$\beta_{3,ijk}(\alpha) = \theta_{3,ijk}(\alpha) + \theta_{4,ijk}(\alpha)\bar{X}$, 
and Definition~\ref{cond_spill_over_under_struc_model}, while (2) follows from 
\eqref{assump:indep_tilde_beta_bar_beta_x_int1} in 
Assumption~\ref{assump:indep_tilde_beta_bar_beta_x}. Hence,
\begin{equation*}
\begin{split}
S_{N} \bigl\lvert \hat{\beta}_{A,3}(\alpha)-\bar{\beta}_{3}(\alpha,x) \bigr\rvert
&\leq S_{N}
\bigl\lvert 
   \hat{\beta}_{A,3}(\alpha)
   - \beta^r_{A,3}(\alpha)
\bigr\rvert +S_{N}\bigl\lvert 
   \beta^r_{A,3}(\alpha)
   - \beta^p_{A,3}(\alpha)
\bigr\rvert \\
&\quad+ S_{N}
\bigl\lvert 
   \beta^p_{A,3}(\alpha)
   - \bar{\beta}_{3}(\alpha,x)
\bigr\rvert \overset{N\rightarrow \infty}{\longrightarrow}_{(1)} 0,
\end{split}
\end{equation*}

where (1) follows from Propositions~\ref{Const_ratio_estimable_quantity} and~\ref{beta_r_to_beta_p}, together with~\eqref{consistency_hat_tau_x_cond_tau_x_int1}. For $h = 4$, we have
\begin{equation}
\label{consistency_hat_tau_x_cond_tau_x_int2}
\begin{split}
S_{N}\bigl\lvert \beta^p_{A,4}(\alpha)-\bar{\beta}_{4}(\alpha,x) \bigr\rvert
&= S_{N}\Bigg\lvert 
   \sum_{k=1}^K \sum_{ik=1}^{n_k} \sum_{jk \neq ik} 
   S^{-1} S_{ik,jk}\,\theta_{4,ijk}(\alpha) - \bar{\beta}_{4}(\alpha,x)
   \Bigg\rvert \\
&=_{(1)} S_{N}\bigl\lvert \theta^p_{A,4}(\alpha)-\theta_{4}(\alpha,x)\bigr\rvert 
   \overset{N\rightarrow \infty}{\longrightarrow}_{(2)} 0,
\end{split}
\end{equation}
where (1) uses $\beta_{4,ijk}(\alpha) = \theta_{4,ijk}(\alpha)$, 
and (2) follows from~\eqref{assump:indep_tilde_beta_bar_beta_x_int2}.  
Therefore, by the same argument as above and by 
Propositions~\ref{Const_ratio_estimable_quantity} and~\ref{beta_r_to_beta_p}, together with~\eqref{consistency_hat_tau_x_cond_tau_x_int2}, we have
\[
S_{N} \bigl\lvert \hat{\beta}_{A,4}(\alpha) - \bar{\beta}_{4}(\alpha,x) \bigr\rvert 
\;\overset{N \to \infty}{\longrightarrow}\; 0.
\]
\end{proof}

\begin{proof}[Proof of Theorem~\ref{consistency_hat_tau_x_cond_tau_x}]
To establish the consistency of $\hat{\tau}_{A}(\alpha,x)$, note that
\begin{small}
\begin{equation*}
\begin{split}
\bigl\lvert \hat{\tau}_{A}(\alpha,x)-\tau(\alpha,x) \bigr\rvert
&\leq S_{N}
\bigl\lvert 
   \hat{\beta}_{A,3}(\alpha)+\hat{\beta}_{A,4}(\alpha)\tilde{X}
   - \bigl(\beta^r_{A,3}(\alpha)+\beta^r_{A,4}(\alpha)\tilde{X}\bigr)
\bigr\rvert \\
&\quad+
S_{N}\bigl\lvert 
   \beta^r_{A,3}(\alpha)+\beta^r_{A,4}(\alpha)\tilde{X}
   - \bigl(\beta^p_{A,3}(\alpha)+\beta^p_{A,4}(\alpha)\tilde{X}\bigr)
\bigr\rvert \\
&\quad+ S_{N}
\bigl\lvert 
   \beta^p_{A,3}(\alpha)+\beta^p_{A,4}(\alpha)\tilde{X}
   - \bigl(\bar{\beta}_{3}(\alpha,x)+\bar{\beta}_{4}(\alpha,x)\tilde{X}\bigr)
\bigr\rvert \overset{N\rightarrow \infty}{\longrightarrow}_{(1)} 0,
\end{split}
\end{equation*}
\end{small}
where~(1) follows from Proposition~\ref{CSE_cons_ratio_population_respectively}.
\end{proof}

\begin{proof}[Proof of Theorem~\ref{CLT_est}]
\label{proof:CLT_ASE_est}
Under the same assumptions as in Theorem~\ref{CLT_ASE_est}, together with 
Assumption~\ref{bounded_x}, and following the same line of argument, we obtain
\begin{equation}
\label{clt_cse_est_beta_r}
    \Sigma^{-1/2}_{A}\, S_{N}
    \bigl(\boldsymbol{\hat{\beta}}_{A}(\alpha)
    - \boldsymbol{\beta}^r_{A}(\alpha)\bigr) 
    \;\overset{N\rightarrow\infty}{\longrightarrow}\;
    \mathcal{N}(0, I),
\end{equation}
where $\Sigma_{A}$ is defined in Theorem~\ref{CLT_est}. 

We next examine the order of 
$\bigl[(1,\tilde{x}) \, \Sigma_{A,(3,4),(3,4)} \, (1,\tilde{x})^{\top}\bigr]^{1/2}$.
Note that, in the expression $\Omega^{-1}_{A}\Gamma_{A}\Omega^{-1}_{A}$,
the normalizing weight $S$ cancels out. Hence, in $\Omega_{A}$ and $\Gamma_{A}$,
we work with $S_{ik,jk}$ rather than $S^{r}_{ik,jk}$.

Let first consider $A \in \{D, S\}$. 
From the proof of Proposition~\ref{beta_r_to_beta_p},
$\Omega_{A}$ is diagonal, with each diagonal element of order $O(1)$ 
by~\eqref{proof_Const_ratio_estimable_quantity_int3} and the definition of $\rho_N$.
We now bound each term in~$\Gamma_A$. Recalling its expression in~\eqref{proof:ase_conserve_variance_est_int1},
\begin{equation*}
    \Gamma_{A} 
    \preceq 
    \sum_{k=1}^K  
    \mathbb{E}\!\bigl(V^{\top}_{A,k} B_{A,k} \xi_{A,k} 
    \xi^{\top}_{A,k} B_{A,k} V_{A,k}\bigr).
\end{equation*}
Each term in the upper bound can be written as
\begin{small}
\begin{equation}
\label{proof:CLT_ASE_est_int1}
\begin{split}
   & \sum_{k=1}^K 
    \sum_{i_1k=1}^{n_k}\sum_{j_1k\neq i_1k}
    \sum_{i_2k=1}^{n_k}\sum_{j_2k\neq i_2k}
    \mathbb{E}\!\bigl(
        S_{i_1k,j_1k} R_{i_1k,j_1k}
        S_{i_2k,j_2k} R_{i_2k,j_2k}
    \bigr)\\
   & \;\leq_{(1)}\; C K \bar{n}_k^4 
    \max_{(i_1k,j_1k),(i_2k,j_2k)} 
    S_{i_1k,j_1k} S_{i_2k,j_2k},
    \end{split}
\end{equation}
\end{small}
where $C>0$ is a constant and $(1)$ follows from 
Assumptions~\ref{unif_bound_weight}, \ref{pos_Ratio}, \ref{unif_bound_pot_out}, 
and~\ref{bounded_x}. Since each element of $\Omega^{-1}_{A}\Gamma_{A}\Omega^{-1}_{A}$ 
is a linear combination of the entries of $\Gamma_{A}$, 
with coefficients of order~$O(1)$, 
the largest term in $\Omega^{-1}_{A}\Gamma_{A}\Omega^{-1}_{A}$ 
is of order 
$O\!\bigl(K\bar{n}_k^4 
\max_{(i_1k,j_1k),(i_2k,j_2k)} S_{i_1k,j_1k}S_{i_2k,j_2k}\bigr)$. Hence,
\[
\Bigl[(1,\tilde{x}) \, 
\Sigma_{A,(3,4),(3,4)} \, (1,\tilde{x})^{\top}\Bigr]^{1/2}
\leq 
O\!\bigl(
K^{1/2}\bar{n}_k^2
\max_{(i_1k,j_1k),(i_2k,j_2k)} 
S^{1/2}_{i_1k,j_1k} S^{1/2}_{i_2k,j_2k}
\bigr).
\]
Under the rate conditions of ${\tau}^r_{A}(\alpha,x) - {\tau}^p_{A}(\alpha,x)$ and ${\tau}^p_{A}(\alpha,x) - {\tau}_{A}(\alpha,x)$ in Theorem~\ref{CLT_est}, we obtain
\begin{equation}
\label{proof:CLT_ASE_est_int2}
 \Bigl[(1,\tilde{x}) \, 
\Sigma_{A,(3,4),(3,4)} \, (1,\tilde{x})^{\top}\Bigr]^{-1/2}
\bigl({\tau}^r_{A}(\alpha,x) 
- {\tau}^p_{A}(\alpha,x)\bigr)
= o(1),
\end{equation}
and similarly,
\begin{equation}
\label{proof:CLT_ASE_est_int3}
 \Bigl[(1,\tilde{x}) \, 
\Sigma_{A,(3,4),(3,4)} \, (1,\tilde{x})^{\top}\Bigr]^{-1/2}
\bigl({\tau}^p_{A}(\alpha,x) 
- {\tau}_{A}(\alpha,x)\bigr)
= o(1).
\end{equation}

For $A=R$, from~\eqref{proof_Const_ratio_estimable_quantity_int2},
each diagonal block of $\Omega^{-1}_R$ is of order $O(1)$, 
and the off-diagonal blocks are zero.
From~\eqref{proof_ase_indep_xi_covariate_int1}, 
each term in $V^{\top}_{R,h} B_R {\xi}_R$ for $h \in \{1,\ldots,4\}$
can be written as 
$\sum_{k=1}^K \sum_{ik=1}^{n_k} \sum_{jk \neq ik} S_{ik,jk} R_{ik,jk}$.
Thus, each term in $\Gamma_R$ admits the same representation as in 
\eqref{proof:CLT_ASE_est_int1}, and 
\eqref{proof:CLT_ASE_est_int2}–\eqref{proof:CLT_ASE_est_int3} 
follow analogously for $A=R$. Combining the above results, we have
\begin{equation}
\begin{split}
& \Bigl[(1,\tilde{x}) \, 
\Sigma_{A,(3,4),(3,4)} \, (1,\tilde{x})^{\top}\Bigr]^{-1/2}
\bigl(\hat{\tau}_{A}(\alpha,x)
- \tau(\alpha,x)\bigr)\\
&= \Bigl[(1,\tilde{x}) \, 
\Sigma_{A,(3,4),(3,4)} \, (1,\tilde{x})^{\top}\Bigr]^{-1/2}
S_{N} (1,\tilde{x})
\bigl(\boldsymbol{\hat{\beta}}_{A,(3,4)}(\alpha)-\boldsymbol{\bar{\beta}}_{(3,4)}(\alpha,x)\bigr)\\
& = \Bigl[(1,\tilde{x}) \, 
\Sigma_{A,(3,4),(3,4)} \, (1,\tilde{x})^{\top}\Bigr]^{-1/2}
S_{N} (1,\tilde{x}) \\
& \bigl(\boldsymbol{\hat{\beta}}_{A,(3,4)}(\alpha)-\boldsymbol{\beta}^r_{A,(3,4)}(\alpha)+ \boldsymbol{{\beta}}^r_{A,(3,4)}(\alpha)-\boldsymbol{{\beta}}^p_{A,(3,4)}(\alpha)+\boldsymbol{{\beta}}^p_{A,(3,4)}(\alpha)- \boldsymbol{\bar{\beta}}_{(3,4)}(\alpha,x)\bigr)\\
&=_{(1)}\Bigl[(1,\tilde{x}) \, 
\Sigma_{A,(3,4),(3,4)} \, (1,\tilde{x})^{\top}\Bigr]^{-1/2}
S_{N} (1,\tilde{x})
\bigl(\boldsymbol{\hat{\beta}}_{A, (3,4)}(\alpha)-\boldsymbol{{\beta}}^r_{(3,4)}(\alpha)\bigr)+ o(1)  \\
& \overset{N\rightarrow\infty}{\longrightarrow}_{(2)} \mathcal{N}(0, I).
\end{split}
\end{equation}
where~(1) follows from 
\eqref{proof:CLT_ASE_est_int2}–\eqref{proof:CLT_ASE_est_int3},
and~(2) follows from~\eqref{clt_cse_est_beta_r}.
\end{proof}

\begin{proof}[Proof of Proposition~\ref{conserve_variance_est}]
Under the same assumptions used in the proof of Proposition~\ref{ase_conserve_variance_est}, together with Assumption~\ref{bounded_x}, the argument and the resulting rate conditions required for the consistency of the cluster-robust variance estimator follow identically to those in Proposition~\ref{ase_conserve_variance_est}.
\end{proof}

\begin{proof}[Proof of Proposition \ref{variance_comp}]
\label{proof:variance_comp}
   Based on the formula in \eqref{Gamma_A_formula}, first write
\begin{equation*}
   \Gamma_{A}
   = \sum_{k=1}^K 
   \mathbb{E}\!\bigl(
      V_{A,k}^{\top} B_{A,k} 
      \xi_{A,k} \xi_{A,k}^{\top} 
      B_{A,k} V_{A,k}
   \bigr)
   =
   \begin{bmatrix} 
   \Gamma_{A,11} & \Gamma_{A,12} & \Gamma_{A,13} & \Gamma_{A,14} \\
   \Gamma_{A,21} & \Gamma_{A,22} & \Gamma_{A,23} & \Gamma_{A,24} \\
   \Gamma_{A,31} & \Gamma_{A,32} & \Gamma_{A,33} & \Gamma_{A,34} \\
   \Gamma_{A,41} & \Gamma_{A,42} & \Gamma_{A,43} & \Gamma_{A,44} \\
   \end{bmatrix}
\end{equation*}
for $A \in \{D,S,R\}$.  
Let $V_{A,h,k}$ denote the $h$-th column of $V_{A}$ restricted to cluster $k$.  
Then, for $A = R$,
\begin{equation*}
    \begin{split}
      V^{\top}_{R,1,k} B_{R,k} \xi_{R,k}
      &= \sum_{ik=1}^{n_k} \sum_{jk\neq ik} 
      B_{ik,jk} \Bigl( Y_{ik} - \beta^r_{R,1}(\alpha)
      - Z_{jk}\beta^r_{R,3}(\alpha) \Bigr), \\
      V^{\top}_{R,2,k} B_{R,k} \xi_{R,k}
      &= \sum_{ik=1}^{n_k} \sum_{jk\neq ik} 
      B_{ik,jk} X^\dagger_{ik} \Bigl( Y_{ik} - \beta^r_{R,2}(\alpha)
      - Z_{jk} X^\dagger_{ik}\beta^r_{R,4}(\alpha) \Bigr), \\
      V^{\top}_{R,3,k} B_{R,k} \xi_{R,k}
      &= \sum_{ik=1}^{n_k} \sum_{jk\neq ik} 
      B_{ik,jk} Z_{jk} \Bigl( Y_{ik} - \beta^r_{R,1}(\alpha)
      - Z_{jk}\beta^r_{R,3}(\alpha) \Bigr), \\
      V^{\top}_{R,4,k} B_{R,k} \xi_{R,k}
      &= \sum_{ik=1}^{n_k} \sum_{jk\neq ik} 
      B_{ik,jk} Z_{jk} X^\dagger_{ik} \Bigl( Y_{ik} - \beta^r_{R,2}(\alpha)
      - Z_{jk} X^\dagger_{ik}\beta^r_{R,4}(\alpha) \Bigr),
    \end{split}
\end{equation*}
where $B_{ik} = B^1_{ik} + B^0_{ik}$ and $B^z_{ik}$ is defined in Definition~\ref{est_ASE_R}. For $A \in \{D,S\}$, we have
\begin{small}
\begin{equation*}
\begin{split}
V^{\top}_{A,1,k} B_{A,k} \xi_{A,k}
&=
\sum_{ik=1}^{n_k} \sum_{jk\neq ik} 
B_{ik,jk}
\bigl[
  Y_{ik}
  - \beta^r_{A,1}(\alpha)
  - \beta^r_{A,2}(\alpha)\tilde{X}_{jk}
  - (\beta^r_{A,3}(\alpha)+\beta^r_{A,4}(\alpha)\tilde{X}_{jk}) Z^{*}_{jk}
\bigr],\\
V^{\top}_{A,2,k} B_{A,k} \xi_{A,k}
&=
\sum_{ik=1}^{n_k} \sum_{jk\neq ik} 
B_{ik,jk}\,\tilde{X}_{jk}
\bigl[
  Y_{ik}
  - \beta^r_{A,1}(\alpha)
  - \beta^r_{A,2}(\alpha)\tilde{X}_{jk}
  - (\beta^r_{A,3}(\alpha)+\beta^r_{A,4}(\alpha)\tilde{X}_{jk}) Z^{*}_{jk}
\bigr],\\
V^{\top}_{A,3,k} B_{A,k} \xi_{A,k}
&=
\sum_{ik=1}^{n_k} \sum_{jk\neq ik} 
B_{ik,jk}\, Z^{*}_{jk}
\bigl[
  Y_{ik}
  - \beta^r_{A,1}(\alpha)
  - \beta^r_{A,2}(\alpha)\tilde{X}_{jk}
  - (\beta^r_{A,3}(\alpha)+\beta^r_{A,4}(\alpha)\tilde{X}_{jk}) Z^{*}_{jk}
\bigr],\\
V^{\top}_{A,4,k} B_{A,k} \xi_{A,k}
&=
\sum_{ik=1}^{n_k} \sum_{jk\neq ik} 
B_{ik,jk} Z^{*}_{jk} \tilde{X}_{jk}
\bigl[
  Y_{ik}
  - \beta^r_{A,1}(\alpha)
  - \beta^r_{A,2}(\alpha)\tilde{X}_{jk}
  - (\beta^r_{A,3}(\alpha)+\beta^r_{A,4}(\alpha)\tilde{X}_{jk}) Z^{*}_{jk}
\bigr].
\end{split}
\end{equation*}
\end{small}
Next, using Theorem~\ref{CLT_est}, we can express $\Omega_{A}$ for 
$A\in\{D,S\}$ and $\Omega_{R}$ as
\begin{equation*}
    \Omega_{A} 
    =_{(1)} 
    S_{N}^{-1}
    \begin{bmatrix} 
   \frac{1}{2} & 0 & 0 & 0 \\
   0 & \frac{1}{2}(\tilde{X}_{\mathrm{ave}}^{2})^{-1} & 0 & 0 \\
   0 & 0 & 2 & 0 \\
   0 & 0 & 0 & 2(\tilde{X}_{\mathrm{ave}}^{2})^{-1} \\
   \end{bmatrix}, \ \   
    \Omega_{R}
    =_{(2)} 
    S_{N}^{-1}
    \begin{bmatrix} 
   1 & 0 & -1 & 0 \\
   0 & (X_{\mathrm{ave}}^{\dagger 2})^{-1} & 0 & -(X_{\mathrm{ave}}^{\dagger 2})^{-1} \\
   -1 & 0 & 2 & 0 \\
   0 & -(X_{\mathrm{ave}}^{\dagger 2})^{-1} & 0 & 2(X_{\mathrm{ave}}^{\dagger 2})^{-1} \\
   \end{bmatrix},
\end{equation*}
where $(1)$ follows from \eqref{proof_Const_ratio_estimable_quantity_int3} and 
$(2)$ from \eqref{proof_Const_ratio_estimable_quantity_int2}. By direct calculation, for $A \in \{D,S\}$ we obtain
\begin{small}
\begin{equation*}
    \bigl(\Omega^{-1}_{A}\Gamma_{A}\Omega^{-1}_{A}\bigr)_{(3,4),(3,4)}
    =
    S_{N}^{2}
    \begin{bmatrix} 
   4\Gamma_{A,33} & 4(\tilde{X}_{\mathrm{ave}}^{2})^{-1}\Gamma_{A,34} \\
   4(\tilde{X}_{\mathrm{ave}}^{2})^{-1}\Gamma_{A,43} & 4(\tilde{X}_{\mathrm{ave}}^{2})^{-2}\Gamma_{A,44}
   \end{bmatrix}.
\end{equation*}
\end{small}
The conservative variance for $\hat{\tau}_{A}(\alpha,x)$ is then
\begin{equation}
\label{cons_var_D}
    \begin{split}
       (1,\tilde{x}) (\Sigma^c_{A})_{(3,4),(3,4)} (1,\tilde{x})^{\top}
       &= S_{N}^{2}
       \Bigl[
         4\Gamma_{A,33}
         + 2\tilde{x}\cdot 4(\tilde{X}_{\mathrm{ave}}^{2})^{-1}\Gamma_{A,43}
         + \tilde{x}^{2}\cdot 4(\tilde{X}_{\mathrm{ave}}^{2})^{-2}\Gamma_{A,44}
       \Bigr].
    \end{split}
\end{equation}
For $A = R$, we similarly obtain
\begin{equation*}
\begin{split}
& \bigl(\Omega^{-1}_{R}\Gamma_{R}\Omega^{-1}_{R}\bigr)_{(3,4),(3,4)} \\
 &=
 S_{N}^{2}
 \begin{bmatrix} 
   \Gamma_{R,11}-2\Gamma_{R,31}-2\Gamma_{R,13}+4\Gamma_{R,33} &
   (X_{\mathrm{ave}}^{\dagger 2})^{-1}
   (\Gamma_{R,12}-2\Gamma_{R,32}-2\Gamma_{R,14}+4\Gamma_{R,34}) \\
   (X_{\mathrm{ave}}^{\dagger 2})^{-1}
   (\Gamma_{R,21}-2\Gamma_{R,41}-2\Gamma_{R,23}+4\Gamma_{R,43}) &
   (X_{\mathrm{ave}}^{\dagger 2})^{-2}
   (\Gamma_{R,22}-2\Gamma_{R,42}-2\Gamma_{R,24}+4\Gamma_{R,44}) 
 \end{bmatrix}.\\
\end{split}
\end{equation*}
Hence, the conservative variance for $\hat{\tau}_{R}(\alpha,x)$ is
\begin{equation}
\label{cons_var_R}
    \begin{split}
       & (1,\tilde{x}) (\Sigma^c_{R})_{(3,4),(3,4)} (1,\tilde{x})^{\top}\\
       &\quad= S_{N}^{2}
       \Bigl[
         \Gamma_{R,11}-4\Gamma_{R,13}+4\Gamma_{R,33}
         + 2\tilde{x}(X_{\mathrm{ave}}^{\dagger 2})^{-1}
         \bigl(\Gamma_{R,21}-2\Gamma_{R,23}-2\Gamma_{R,14}+4\Gamma_{R,34}\bigr) \\
       &\qquad\quad
         + \tilde{x}^{2}(X_{\mathrm{ave}}^{\dagger 2})^{-2}
         \bigl(\Gamma_{R,22}-4\Gamma_{R,24}+4\Gamma_{R,44}\bigr)
       \Bigr].
    \end{split}
\end{equation}

Taking the difference between \eqref{cons_var_D} and \eqref{cons_var_R} yields
\begin{equation}
\label{dif_var_S_var_R}
\begin{split}
& (1,\tilde{x})(\Sigma^c_{A})_{(3,4),(3,4)}(1,\tilde{x})^{\top}
- (1,\tilde{x})(\Sigma^c_{R})_{(3,4),(3,4)}(1,\tilde{x})^{\top} \\
&\quad = S_{N}^{2}
\Bigl\{
   \bigl[4\Gamma_{A,33} - (\Gamma_{R,11}-4\Gamma_{R,13}+4\Gamma_{R,33})\bigr] \\
&\qquad\quad
   + 2\tilde{x}
     \Bigl[
       4(\tilde{X}_{\mathrm{ave}}^{2})^{-1}\Gamma_{A,43}
       - (X_{\mathrm{ave}}^{\dagger 2})^{-1}
         \bigl(\Gamma_{R,21}-2\Gamma_{R,23}-2\Gamma_{R,14}+4\Gamma_{R,34}\bigr)
     \Bigr] \\
&\qquad\quad
   + \tilde{x}^{2}
     \Bigl[
       4(\tilde{X}_{\mathrm{ave}}^{2})^{-2}\Gamma_{A,44}
       - (X_{\mathrm{ave}}^{\dagger 2})^{-2}
         \bigl(\Gamma_{R,22}-4\Gamma_{R,24}+4\Gamma_{R,44}\bigr)
     \Bigr]
\Bigr\}.
\end{split}
\end{equation}
Therefore, for the potential outcomes, together with the distribution of the hypothetical and realized treatment assignments, such that the expression in \eqref{dif_var_S_var_R} is negative, we obtain
\[
(1,\tilde{x})(\Sigma^c_{A})_{(3,4),(3,4)}(1,\tilde{x})^{\top}
<
(1,\tilde{x})(\Sigma^c_{R})_{(3,4),(3,4)}(1,\tilde{x})^{\top}
\quad\text{for } A \in \{D,S\}.
\]
Similarly, if the expression in \eqref{dif_var_S_var_R} is positive, the inequality is reversed, and
\[
(1,\tilde{x})(\Sigma^c_{A})_{(3,4),(3,4)}(1,\tilde{x})^{\top}
>
(1,\tilde{x})(\Sigma^c_{R})_{(3,4),(3,4)}(1,\tilde{x})^{\top}
\quad\text{for } A \in \{D,S\}.
\]
This completes the proof.
\end{proof}


\section{Additional simulation and empirical results}
\label{Appendix:simulation_empirical}

In this section, we present additional simulation results for several estimands under the corresponding data-generating processes (DGPs) defined by the potential outcome models. Specifically, we report results for:
(i) the average outward spillover effect in Example~\ref{Examp_Ave_outward} under Model~\eqref{ASE_PO1};
(ii) the average inward spillover effect in Example~\ref{Examp_Ave_inward} under Models~\eqref{ASE_PO1} and~\eqref{ASE_PO2};
(iii) the conditional outward spillover effect in Example~\ref{Examp_cond_outward} under Model~\eqref{CSE_PO3}; and
(iv) the conditional inward spillover effect in Example~\ref{Examp_cond_inward} under Model~\eqref{CSE_PO3}.

We also provide two supplementary tables reporting empirical results for the inward spillover estimands (Examples~\ref{Examp_Ave_inward} and~\ref{Examp_cond_inward}) and for the pairwise spillover estimands (Examples~\ref{Examp_ave_ind_effect} and~\ref{Examp_cond_ind_effect}) using the real-world application described in Section~\ref{sec:Real Data Application}.

\begin{table}[H]
\centering
\caption{Simulation results for the average outward spillover effect in Example \ref{Examp_Ave_outward} under model \eqref{ASE_PO1}. 
}
\begin{tabular}{cccccccc}
\toprule
$K$ & $\mathbb{E}(\hat{\tau}_D(\alpha))$ & $\mathbb{E}(\hat{\tau}_S(\alpha))$ & $\mathbb{E}(\hat{\tau}_R(\alpha))$ & Bias & ${se}(\hat{\tau}_\cdot(\alpha))$ & $\mathbb{E}[\hat{se}(\hat{\tau}_\cdot(\alpha))]$ & 95\% coverage \\
\midrule
  50 & 1.006 & 1.006 & 1.006 &  0.001 & 0.139 & 0.130 & 0.918 \\
 100 & 1.011 & 1.011 & 1.011 &  0.007 & 0.103 & 0.095 & 0.932 \\
 150 & 0.999 & 0.999 & 0.999 &  0.003 & 0.078 & 0.076 & 0.948 \\
 200 & 1.004 & 1.004 & 1.004 &  0.002 & 0.065 & 0.067 & 0.940 \\
 250 & 1.000 & 1.000 & 1.000 &  0.004 & 0.057 & 0.059 & 0.944 \\
 300 & 0.997 & 0.997 & 0.997 &  0.000 & 0.053 & 0.054 & 0.952 \\
 350 & 1.003 & 1.003 & 1.003 & -0.002 & 0.050 & 0.050 & 0.946 \\
 400 & 1.005 & 1.005 & 1.005 &  0.002 & 0.045 & 0.047 & 0.962 \\
 450 & 0.998 & 0.998 & 0.998 &  0.001 & 0.043 & 0.044 & 0.962 \\
 500 & 0.994 & 0.994 & 0.994 & -0.003 & 0.041 & 0.042 & 0.948 \\
\bottomrule
\end{tabular}
\label{tab:PO1_outward_ASE}
\begin{minipage}{0.95\linewidth}
\par\smallskip
\noindent $\mathbb{E}(\hat{\tau}_{\cdot}(\alpha))$ denotes the Monte Carlo mean of the estimator. 
$se(\hat{\tau}_{\cdot}(\alpha))$ is the empirical standard error of $\hat{\tau}_{\cdot}(\alpha)$, computed as the sample standard deviation. 
$\mathbb{E}[\hat{se}(\hat{\tau}_{\cdot}(\alpha))]$ denotes the Monte Carlo average of the 
estimated standard errors.
\end{minipage}
\end{table}

\begin{table}[H]
\centering
\caption{Simulation results for the average inward spillover effect in Example \ref{Examp_Ave_inward} under model \eqref{ASE_PO1}.}
\begin{tabular}{cccccccc}
\toprule
$K$ & $\mathbb{E}(\hat{\tau}_D(\alpha))$ & $\mathbb{E}(\hat{\tau}_S(\alpha))$ & $\mathbb{E}(\hat{\tau}_R(\alpha))$ & Bias & ${se}(\hat{\tau}_\cdot(\alpha))$ & $\mathbb{E}[\hat{se}(\hat{\tau}_\cdot(\alpha))]$ & 95\% coverage \\
\midrule
  50 & 1.000 & 1.000 & 1.000 &  0.003 & 0.129 & 0.125 & 0.934 \\
 100 & 1.007 & 1.007 & 1.007 &  0.008 & 0.097 & 0.091 & 0.928 \\
 150 & 1.002 & 1.002 & 1.002 &  0.004 & 0.074 & 0.074 & 0.958 \\
 200 & 0.998 & 0.998 & 0.998 &  0.000 & 0.061 & 0.064 & 0.954 \\
 250 & 1.001 & 1.001 & 1.001 &  0.003 & 0.057 & 0.058 & 0.948 \\
 300 & 0.998 & 0.998 & 0.998 &  0.002 & 0.050 & 0.052 & 0.962 \\
 350 & 1.002 & 1.002 & 1.002 & -0.001 & 0.049 & 0.048 & 0.944 \\
 400 & 1.004 & 1.004 & 1.004 &  0.001 & 0.042 & 0.045 & 0.968 \\
 450 & 0.998 & 0.998 & 0.998 &  0.000 & 0.042 & 0.043 & 0.956 \\
 500 & 0.997 & 0.997 & 0.997 & -0.002 & 0.040 & 0.040 & 0.956 \\
\bottomrule
\end{tabular}
\label{tab:PO1_inward_ASE}
\begin{minipage}{0.95\linewidth}
\par\smallskip
\noindent $\mathbb{E}(\hat{\tau}_{\cdot}(\alpha))$ denotes the Monte Carlo mean of the estimator.$se(\hat{\tau}_{\cdot}(\alpha))$ is the empirical standard error of $\hat{\tau}_{\cdot}(\alpha)$, computed as the sample standard deviation. 
$\mathbb{E}[\hat{se}(\hat{\tau}_{\cdot}(\alpha))]$ denotes the Monte Carlo average of the 
estimated standard errors of $\hat{\tau}_{\cdot}(\alpha)$.
\end{minipage}
\end{table}

\begin{table}[H]
\centering
\caption{Simulation results for the average inward spillover effect in Example \ref{Examp_Ave_inward} under potential outcome model \eqref{ASE_PO2} }
\begin{tabular}{cccccccc}
\toprule
$K$ & $\mathbb{E}(\hat{\tau}_D(\alpha))$ & $\mathbb{E}(\hat{\tau}_S(\alpha))$ & $\mathbb{E}(\hat{\tau}_R(\alpha))$ & Bias & ${se}(\hat{\tau}_\cdot(\alpha))$ & $\mathbb{E}[\hat{se}(\hat{\tau}_\cdot(\alpha))]$ & 95\% coverage \\
\midrule
  50 & 0.994 & 0.994 & 0.994 & -0.003 & 0.175 & 0.169 & 0.910 \\
 100 & 1.025 & 1.025 & 1.025 &  0.026 & 0.138 & 0.124 & 0.925 \\
 150 & 1.011 & 1.011 & 1.011 &  0.010 & 0.100 & 0.102 & 0.970 \\
 200 & 1.003 & 1.003 & 1.003 &  0.002 & 0.080 & 0.087 & 0.960 \\
 250 & 0.999 & 0.999 & 0.999 &  0.001 & 0.079 & 0.078 & 0.935 \\
 300 & 1.003 & 1.003 & 1.003 &  0.005 & 0.071 & 0.071 & 0.950 \\
 350 & 1.005 & 1.005 & 1.005 &  0.003 & 0.066 & 0.065 & 0.955 \\
 400 & 1.004 & 1.004 & 1.004 &  0.002 & 0.052 & 0.062 & 0.985 \\
 450 & 1.004 & 1.004 & 1.004 &  0.003 & 0.056 & 0.058 & 0.970 \\
 500 & 1.000 & 1.000 & 1.000 &  0.002 & 0.051 & 0.055 & 0.980 \\
\bottomrule
\end{tabular}
\label{tab:PO2_inward_ASE}
\begin{minipage}{0.95\linewidth}
\par\smallskip
\noindent $\mathbb{E}(\hat{\tau}_{\cdot}(\alpha))$ denotes the Monte Carlo mean of the estimator. 
$se(\hat{\tau}_{\cdot}(\alpha))$ is the empirical standard error of $\hat{\tau}_{\cdot}(\alpha)$, computed as the sample standard deviation. 
$\mathbb{E}[\hat{se}(\hat{\tau}_{\cdot}(\alpha))]$ denotes the Monte Carlo average of the 
estimated standard errors of $\hat{\tau}_{\cdot}(\alpha)$. 
\end{minipage}
\end{table}

For the average outward and inward spillover effects, the heterogeneous
coefficients in Models~\eqref{ASE_PO1} and~\eqref{ASE_PO2} imply that the two
estimands generally differ \citep{fang2025inwardoutwardspillovereffects}. 
This is reflected in the discrepancies between 
$\mathbb{E}(\hat{\tau}_{A}(\alpha))$ for $A \in \{D,S,R\}$ 
and the corresponding Bias reported in 
Tables~\ref{tab:PO1_outward_ASE} and~\ref{tab:PO1_inward_ASE} for 
Model~\eqref{ASE_PO1}, and in 
Tables~\ref{tab:PO2_outward_ASE} and~\ref{tab:PO2_inward_ASE} for 
Model~\eqref{ASE_PO2}. 
Moreover, under the more complex data-generating process in 
Model~\eqref{ASE_PO2} (relative to Model~\eqref{ASE_PO1}), 
the variance estimators exhibit greater conservativeness, 
consistent with the increased complexity in the underlying potential outcomes. We next turn to the simulation results for the CSE under Setting~2 of
Lemma~\ref{settings_all_conditions_satisfy}, focusing on the conditional inward
spillover effect defined in Example~\ref{Examp_cond_inward}.


\begin{table}[H]
\setlength{\tabcolsep}{3pt} 
\centering
\caption{Simulation results for the bias of $\hat{\beta}_{A,\cdot}(\alpha)$, $A \in \{D,S\}$, for the conditional inward spillover effect at $x=1$ in Example \ref{Examp_cond_inward}  
under model \eqref{CSE_PO3} with directed Erd\H{o}s--R\'enyi cluster graphs}
\begin{tabular}{c|ccc|ccc|ccc}
\toprule
$K$ & $\bar{\beta}_3(\alpha,1)$ & $\mathbb{E}[\hat{\beta}_{A,3}(\alpha)]$ & $\mathbb{E}[\hat{\beta}_{R,3}(\alpha)]$ & $\bar{\beta}_4(\alpha,1)$ & $\mathbb{E}[\hat{\beta}_{A,4}(\alpha)]$ & $\mathbb{E}[\hat{\beta}_{R,4}(\alpha)]$ & $\tau(\alpha,1)$ & $\hat{\tau}_{A}(\alpha,1)$ & $\hat{\tau}_R(\alpha,1)$ \\
\midrule
  50 & 0.801 & 0.800 & 0.800 & 0.400 & 0.394 & 0.300 & 0.900 & 0.898 & 0.874 \\
 100 & 0.797 & 0.795 & 0.795 & 0.400 & 0.401 & 0.413 & 0.900 & 0.898 & 0.901 \\
 150 & 0.800 & 0.797 & 0.797 & 0.400 & 0.396 & 0.385 & 0.900 & 0.897 & 0.894 \\
 200 & 0.799 & 0.799 & 0.799 & 0.400 & 0.399 & 0.344 & 0.900 & 0.900 & 0.886 \\
 250 & 0.804 & 0.803 & 0.803 & 0.400 & 0.401 & 0.409 & 0.900 & 0.899 & 0.901 \\
 300 & 0.800 & 0.798 & 0.798 & 0.400 & 0.406 & 0.392 & 0.900 & 0.900 & 0.897 \\
 350 & 0.799 & 0.799 & 0.799 & 0.400 & 0.401 & 0.387 & 0.900 & 0.900 & 0.897 \\
 400 & 0.801 & 0.801 & 0.801 & 0.400 & 0.400 & 0.398 & 0.900 & 0.899 & 0.899 \\
 450 & 0.800 & 0.799 & 0.799 & 0.400 & 0.401 & 0.428 & 0.900 & 0.899 & 0.906 \\
 500 & 0.799 & 0.799 & 0.799 & 0.400 & 0.400 & 0.378 & 0.900 & 0.900 & 0.895 \\
\bottomrule
\end{tabular}
\label{tab:CSE_inward_homo_ER_grap_bias}
\begin{minipage}{0.95\linewidth}
\par\smallskip
\noindent $\bar{\beta}_{h}(\alpha,1)$ for $h\in \{3,4\}$ denotes the coefficients in CSE as in Definition \ref{cond_spill_over_under_struc_model}. $\mathbb{E}[\hat{\beta}_{A,h}(\alpha)]$ denotes the average estimated coefficient across repetitions. $\tau(\alpha,1)$ denote the CSE as defined in Definition \ref{cond_spill_over_under_struc_model}. $\mathbb{E}[\hat{\tau}_{A}(\alpha,1)]$ denotes the average estimated CSE as defined in Definitions \ref{Est_CSE_dyad}, \ref{Est_CSE_R} and \ref{Est_CSE_S}.
\end{minipage}
\end{table}

\begin{table}[H]
\setlength{\tabcolsep}{3pt}
\centering
\caption{Simulation results for the standard errors of $\hat{\beta}_{A,3}(\alpha)$ for $A \in \{D,S\}$  
and $\hat{\beta}_{R,3}(\alpha)$, for the conditional inward spillover effect in  
Example \ref{Examp_cond_inward} under model \eqref{CSE_PO3} with directed Erd\H{o}s--R\'enyi cluster graphs
}
\begin{tabular}{c|ccc|ccc}
\toprule
$K$ & $se(\hat{\beta}_{A,3}(\alpha))$ & $\mathbb{E}[\hat{se}(\hat{\beta}_{A,3}(\alpha))]$ & 95\% coverage &
$se(\hat{\beta}_{R,3}(\alpha))$ & $\mathbb{E}[\hat{se}(\hat{\beta}_{R,3}(\alpha))]$ & 95\% coverage \\
\midrule
  50 & 0.051 & 0.048 & 0.916 & 0.051 & 0.049 & 0.922 \\
 100 & 0.035 & 0.035 & 0.944 & 0.035 & 0.035 & 0.938 \\
 150 & 0.028 & 0.028 & 0.950 & 0.028 & 0.029 & 0.954 \\
 200 & 0.024 & 0.025 & 0.948 & 0.024 & 0.025 & 0.948 \\
 250 & 0.022 & 0.022 & 0.940 & 0.022 & 0.022 & 0.946 \\
 300 & 0.021 & 0.020 & 0.928 & 0.021 & 0.020 & 0.930 \\
 350 & 0.018 & 0.019 & 0.956 & 0.018 & 0.019 & 0.954 \\
 400 & 0.018 & 0.018 & 0.958 & 0.018 & 0.018 & 0.956 \\
 450 & 0.016 & 0.017 & 0.958 & 0.016 & 0.017 & 0.952 \\
 500 & 0.016 & 0.016 & 0.952 & 0.016 & 0.016 & 0.952 \\
\bottomrule
\end{tabular}
\label{tab:CSE_inward_homo_regular_grap_se_beta_3}
\begin{minipage}{0.95\linewidth}
\par\smallskip
\noindent $se(\hat{\beta}_{\cdot,3}(\alpha))$ denotes the empirical standard error of 
$\hat{\beta}_{\cdot,3}(\alpha)$, computed as the sample standard deviation across 
Monte Carlo replications. 
$\mathbb{E}[\hat{se}(\hat{\beta}_{\cdot,3}(\alpha))]$ denotes the Monte Carlo average 
of the estimated standard errors of $\hat{\beta}_{\cdot,3}(\alpha))$.
\end{minipage}
\end{table}

\begin{table}[H]
\centering
\caption{Simulation results for the standard errors of $\hat{\beta}_{A,4}(\alpha)$ for $A \in \{D,S\}$  
and $\hat{\beta}_{R,4}(\alpha)$, for the conditional inward spillover effect in  
Example \ref{Examp_cond_inward} under model \eqref{CSE_PO3} with directed Erd\H{o}s--R\'enyi cluster graphs
}
\begin{tabular}{c|ccc|ccc}
\toprule
$K$ & $se(\hat{\beta}_{A,4}(\alpha))$ & $\mathbb{E}[\hat{se}(\hat{\beta}_{A,4}(\alpha))]$ & 95\% coverage &
$se(\hat{\beta}_{R,4}(\alpha))$ & $\mathbb{E}[\hat{se}(\hat{\beta}_{R,4}(\alpha))]$ & 95\% coverage \\
\midrule
  50 & 0.159 & 0.152 & 0.932 & 0.571 & 1.641 & 0.928 \\
 100 & 0.114 & 0.111 & 0.944 & 0.370 & 1.186 & 0.962 \\
 150 & 0.097 & 0.091 & 0.930 & 0.309 & 0.964 & 0.956 \\
 200 & 0.080 & 0.079 & 0.944 & 0.274 & 0.834 & 0.954 \\
 250 & 0.072 & 0.072 & 0.932 & 0.248 & 0.760 & 0.956 \\
 300 & 0.067 & 0.065 & 0.944 & 0.216 & 0.678 & 0.954 \\
 350 & 0.060 & 0.060 & 0.946 & 0.200 & 0.625 & 0.954 \\
 400 & 0.057 & 0.056 & 0.942 & 0.178 & 0.596 & 0.962 \\
 450 & 0.053 & 0.053 & 0.948 & 0.181 & 0.565 & 0.946 \\
 500 & 0.054 & 0.050 & 0.936 & 0.165 & 0.533 & 0.958 \\
\bottomrule
\end{tabular}
\label{tab:CSE_inward_homo_regular_grap_se_beta_4}
\begin{minipage}{0.95\linewidth}
\par\smallskip
\noindent $se(\hat{\beta}_{\cdot,4}(\alpha))$ denotes the empirical standard error of 
$\hat{\beta}_{\cdot,4}(\alpha)$, computed as the sample standard deviation across 
Monte Carlo replications. 
$\mathbb{E}[\hat{se}(\hat{\beta}_{\cdot,4}(\alpha))]$ denotes the Monte Carlo average 
of the estimated standard errors of $\hat{\beta}_{\cdot,4}(\alpha)$
\end{minipage}
\end{table}

\begin{table}[H]
\centering
\setlength{\tabcolsep}{3pt}
\caption{Simulation results for the standard errors and coverage of $\hat{\tau}_{A}(\alpha,1)$ for $A \in \{D,S\}$  
and $\hat{\tau}_{R}(\alpha,1)$, for the conditional inward spillover effect in  
Example \ref{Examp_cond_inward} under Model \eqref{CSE_PO3} with directed Erd\H{o}s--R\'enyi cluster graphs 
 }
\begin{tabular}{c|ccc|ccc}
\toprule
$K$ & $se(\hat{\tau}_{A}(\alpha,1))$ & $\mathbb{E}[\hat{se}(\hat{\tau}_{A}(\alpha,1))]$ & $95\%$ coverage &
$se(\hat{\tau}_{R}(\alpha,1))$ & $\mathbb{E}[\hat{se}(\hat{\tau}_{R}(\alpha,1))]$ & $95\%$ coverage \\
\midrule
  50 & 0.065 & 0.063 & 0.930 & 0.438 & 0.418 & 0.930 \\
 100 & 0.048 & 0.047 & 0.940 & 0.305 & 0.313 & 0.962 \\
 150 & 0.038 & 0.037 & 0.938 & 0.244 & 0.248 & 0.954 \\
 200 & 0.032 & 0.033 & 0.954 & 0.207 & 0.216 & 0.956 \\
 250 & 0.028 & 0.029 & 0.966 & 0.187 & 0.188 & 0.948 \\
 300 & 0.028 & 0.027 & 0.918 & 0.169 & 0.175 & 0.962 \\
 350 & 0.024 & 0.025 & 0.958 & 0.160 & 0.162 & 0.954 \\
 400 & 0.023 & 0.023 & 0.940 & 0.138 & 0.151 & 0.954 \\
 450 & 0.021 & 0.022 & 0.960 & 0.143 & 0.145 & 0.944 \\
 500 & 0.021 & 0.021 & 0.940 & 0.133 & 0.138 & 0.954 \\
\bottomrule
\end{tabular}
\label{tab:CSE_inward_regular_directed_ER_grap_se_beta_4}
\begin{minipage}{0.95\linewidth}
\par\smallskip
\noindent $\mathbb{E}(\hat{\tau}_{\cdot}(\alpha,1))$ denotes the Monte Carlo mean of the estimator. 
$se(\hat{\tau}_{\cdot}(\alpha,1))$ is the empirical standard error of $\hat{\tau}_{\cdot}(\alpha,1)$, computed as the sample standard deviation. 
$\mathbb{E}[\hat{se}(\hat{\tau}_{\cdot}(\alpha,1))]$ denotes the Monte Carlo average of the 
estimated standard errors of $\hat{\tau}_{\cdot}(\alpha,1)$.
\end{minipage}
\end{table}

From Table~\ref{tab:CSE_inward_homo_ER_grap_bias}, all three estimators display
small bias in estimating the coefficients $\bar{\beta}_3(\alpha,1)$ and
$\bar{\beta}_4(\alpha,1)$. The variance estimators for
$\hat{\beta}_{3,A}(\alpha)$, $A \in \{D,S,R\}$, attain coverage close to the
nominal level, as reported in
Table~\ref{tab:CSE_inward_homo_regular_grap_se_beta_3}. Moreover,
Table~\ref{tab:CSE_inward_regular_directed_ER_grap_se_beta_4} shows that the
variance of $\hat{\beta}_{4,A}(\alpha)$ for $A \in \{D,S\}$ is smaller than that
of $\hat{\beta}_{4,R}(\alpha)$, and this discrepancy naturally carries over to
the variances of $\hat{\tau}_{A}(\alpha,1)$ for $A \in \{D,S\}$ and of
$\hat{\tau}_{R}(\alpha,1)$. The underlying explanation for these patterns mirrors
the discussion provided for
Tables~\ref{tab:CSE_outward_homo_regular_grap_se_beta_3} and
\ref{tab:CSE_outward_homo_regular_grap_se_beta_4} in
Section~\ref{sec:Simulation Results}. Regarding interval coverage, the confidence intervals based on
$\hat{\beta}_{4,A}(\alpha)$ for $A \in \{D,S\}$ fall slightly below the nominal
$95\%$ level, whereas those based on $\hat{\beta}_{4,R}(\alpha)$ tend to be
slightly above it. Overall, however, both sets of intervals achieve coverage
reasonably close to the targeted nominal level.

\begin{table}[H]
\centering
\setlength{\tabcolsep}{5pt}
\caption{Inward spillover effects as Examples \ref{Examp_Ave_inward} and \ref{Examp_cond_inward}: estimates and 95\% confidence intervals}
\begin{tabular}{l|cc|cc}
\toprule
& \multicolumn{2}{c|}{$\hat{\tau}_A(\alpha,x)$ for $A \in \{D,S\}$} & \multicolumn{2}{c}{$\hat{\tau}_R(\alpha,x)$} \\
Group & estimate & $95\%$ CI & estimate & $95\%$ CI \\
\midrule
all (ASE) &  $0.006$ & $[-0.019, 0.031]$ &	$0.006$	& $[-0.019, 0.031] $ \\
female & $-0.028$ & $[-0.103,\,0.047]$ & $-0.149$ & $[-0.374,\,0.075]$ \\
male & $0.010$ & $[-0.016,\,0.036]$ & $0.020$ & $[-0.014,\,0.053]$ \\
risk averse $=0$ & $0.002$ & $[-0.032,\,0.037]$ & $-0.027$ & $[-0.084,\,0.031]$ \\
risk averse $>0$ & $0.010$ & $[-0.036,\,0.057]$ & $0.059$ & $[-0.045,\,0.164]$ \\
insurance repay$=0$ & $0.024$ & $[-0.004,\,0.052]$ & $0.011$ & $[-0.036,\,0.058]$ \\
insurance repay$=1$ & $-0.022$ & $[-0.070,\,0.026]$ & $-0.001$ & $[-0.078,\,0.075]$ \\
general trust$=0$ & $0.001$ & $[-0.081,\,0.083]$ & $0.141$ & $[-0.047,\,0.330]$ \\
general trust$=1$ & $0.007$ & $[-0.017,\,0.032]$ & $-0.012$ & $[-0.050,\,0.026]$ \\
in-degree $<4$ & $0.030$ & $[-0.001,\,0.061]$ & $0.113$ & $[0.028,\,0.198]$ \\
in-degree $\ge 4$ & $-0.004$ & $[-0.035,\,0.027]$ & $-0.040$ & $[-0.089,\,0.009]$ \\
out-degree $<4$ & $0.005$ & $[-0.028,\,0.038]$ & $0.038$ & $[-0.031,\,0.106]$ \\
out-degree $\ge 4$ & $0.007$ & $[-0.028,\,0.043]$ & $-0.015$ & $[-0.072,\,0.042]$ \\
disaster=no & $0.014$ & $[-0.027,\,0.056]$ & $0.028$ & $[-0.054,\,0.111]$ \\
disaster=yes & $0.008$ & $[-0.025,\,0.041]$ & $0.000$ & $[-0.063,\,0.062]$ \\
literacy=no & $-0.004$ & $[-0.058,\,0.050]$ & $-0.039$ & $[-0.181,\,0.104]$ \\
literacy=yes & $0.007$ & $[-0.023,\,0.036]$ & $0.014$ & $[-0.027,\,0.056]$ \\
\bottomrule
\end{tabular}
\label{real_data:inward_spillover}
\end{table}

\begin{table}[H]
\centering
\setlength{\tabcolsep}{5pt}
\caption{Average and conditional pairwise spillover effects, corresponding to Examples~\ref{Examp_ave_ind_effect} and \ref{Examp_cond_ind_effect}: estimates and 95\% confidence intervals}
\begin{tabular}{l|cc|cc}
\toprule
& \multicolumn{2}{c|}{$\hat{\tau}_A(\alpha,x)$ for $A \in \{D,S\}$} & \multicolumn{2}{c}{$\hat{\tau}_R(\alpha,x)$} \\
Group & estimate & $95\%$ CI & estimate & $95\%$ CI \\
\midrule
all (ASE) & $0.486$	& $[-0.123,1.095]$	&	$0.486$	& $[-0.123, 1.095]$ \\
female & $3.065$ & $[-1.814,\,7.943]$ & $0.492$ & $[-0.120,\,1.104]$ \\
male & $0.415$ & $[-0.263,\,1.093]$ & $0.495$ & $[-0.123,\,1.112]$ \\
risk averse $=0$ & $0.270$ & $[-0.555,\,1.094]$ & $0.484$ & $[-0.120,\,1.088]$ \\
risk averse $>0$ & $0.809$ & $[-0.681,\,2.298]$ & $0.490$ & $[-0.128,\,1.109]$ \\
insurance repay$=0$ & $0.204$ & $[-1.031,\,1.440]$ & $0.481$ & $[-0.124,\,1.086]$ \\
insurance repay$=1$ & $0.793$ & $[-1.171,\,2.757]$ & $0.495$ & $[-0.121,\,1.111]$ \\
general trust$=0$ & $-2.536$ & $[-5.457,\,0.384]$ & $0.482$ & $[-0.122,\,1.085]$ \\
general trust$=1$ & $0.990$ & $[0.115,\,1.865]$ & $0.487$ & $[-0.123,\,1.096]$ \\
in-degree $<4$ & $1.037$ & $[-0.250,\,2.324]$ & $0.483$ & $[-0.121,\,1.086]$ \\
in-degree $\ge 4$ & $-0.343$ & $[-1.818,\,1.133]$ & $0.490$ & $[-0.126,\,1.106]$ \\
out-degree $<4$  & $0.736$ & $[-0.423,\,1.895]$ & $0.481$ & $[-0.122,\,1.084]$ \\
out-degree $\ge 4$ & $0.288$ & $[-0.640,\,1.216]$ & $0.490$ & $[-0.124,\,1.103]$ \\
disaster=no & $-1.245$ & $[-2.308,\, -0.183]$ & $0.366$ & $[-0.304,\,1.035]$ \\
disaster=yes & $1.372$ & $[-0.054,\,2.797]$ & $0.371$ & $[-0.308,\,1.050]$ \\
literacy=no & $1.536$ & $[0.062,\,3.010]$ & $0.428$ & $[-0.184,\,1.040]$ \\
literacy=yes & $0.174$ & $[-0.575,\,0.923]$ & $0.431$ & $[-0.189,\,1.050]$ \\
\bottomrule
\end{tabular}
\label{tab:pairwise_spillover}
\end{table}

The detailed analysis of Tables~\ref{real_data:inward_spillover} and \ref{tab:pairwise_spillover} is presented in Section~\ref{sec:Real Data Application}. 

\section{Preliminary materials and intermediate proofs}
In this section, we collect several preliminary definitions and intermediate
lemmas whose proofs are lengthy but not central to the flow of the main text.
These definitions and results are used in the proofs of the main theorems presented in the paper.

\begin{definition}[Dependence graph \citep{viviano2024policydesignexperimentsunknown}]
\label{def:dependence_graph}
A dependence graph \( \mathcal{A} \) associated with random variables \( X_1, \ldots, X_N \) is an \( N \times N \) adjacency matrix, where \( \mathcal{A}_{ij} = 1 \) indicates that \( X_i \) and \( X_j \) are dependent, and \( \mathcal{A}_{ij} = 0 \) otherwise.
\end{definition}

\begin{definition}[Cover \citep{viviano2024policydesignexperimentsunknown}]
\label{def:cover}
A cover of a dependence graph \( \mathcal{A} \), denoted by \( \mathcal{C}(\mathcal{A}) \), is a collection of subsets such that within each subset, any pair of random variables \( X_i \) and \( X_j \) are independent. That is, for any \( i, j \) belonging to the same subset, we have \( \mathcal{A}_{ij} = 0 \).
\end{definition}

\begin{lemma}
\label{mean_zero_ASE}
For the estimators of ASE, we have $\mathbb{E}\!\left(V_{A,a}^{\top} B_A \, {\xi}_A \right) = 0$, where ${\xi}_A$ is defined in Proposition~\ref{ase_conserve_variance_est}, for any $a \in \{1,2\}$ and $A \in \{D, R, S\}$.
\end{lemma}

\begin{proof}
Following the notation in Proposition~\ref{ase_conserve_variance_est}, we first derive the explicit form of 
$\boldsymbol{\beta}^{r}(\alpha)$ for the ASE estimators. 
Since both $V_A^{\top} B_A V_A$ and $V_A^{\top} B_A Y_A$ are identical for all 
$A \in \{D, S, R\}$, their expectations are also the same. 
Using the expressions in \eqref{eq:matrix_A_F}, we obtain
\begin{small}
\[
\mathbb{E}(V_A^{\top} B_A V_A)
= \frac{1}{S_{N}}
\begin{bmatrix}
1 & -1\\
-1 & 2
\end{bmatrix},
\ \ 
\mathbb{E}(V_A^{\top} B_A Y_A)
=
\begin{bmatrix}
\sum_{k=1}^K \sum_{ik=1}^{n_k} \sum_{jk \neq ik} 
S_{ik,jk} \!\left[ \bar{Y}_{ik}(1,\alpha) + \bar{Y}_{ik}(0,\alpha) \right]\\[4pt]
\sum_{k=1}^K \sum_{ik=1}^{n_k} \sum_{jk \neq ik} 
S_{ik,jk} \, \bar{Y}_{ik}(1,\alpha)
\end{bmatrix}.
\]
\end{small}
Hence,
\begin{small}
\begin{equation}
\label{explicit_ase_beta_r}
\boldsymbol{\beta}^r(\alpha)
= \bigl( \mathbb{E}(V_D^{\top} B_D V_D) \bigr)^{-1} 
\mathbb{E}(V_A^{\top} B_A Y_A)
=
\begin{bmatrix}
S_{N}^{-1}\!\sum_{k,ik,jk\neq i_k} S_{ik,jk} \bar{Y}_{ik}(0,\alpha) \\[4pt]
S_{N}^{-1}\!\sum_{k,ik,jk\neq ik} S_{ik, jk}\bigl[\bar{Y}_{ik}(1,\alpha)-\bar{Y}_{ik}(0,\alpha)\bigr]
\end{bmatrix}.
\end{equation}
\end{small}

For the dyadic estimator, using \eqref{eq:matrix_A_F}, we have
\begin{small}
\begin{equation*}
\label{proof_ase_indep_xi_covariate}
    \begin{split}
     & \mathbb{E}\large( V^T_{D,1} B_D {\xi}_D)=\mathbb{E} (\sum_{k=1}^K \sum_{ik=1}^{n_k} \sum_{jk \neq ik} B_{ik,jk} \xi_{ik} \large) \\
     & =_{(1)} \sum_{k=1}^K \sum_{\substack{ik=1\\jk \neq ik}}^{n_k} \mathbb{E} \left[B_{ik,jk} \left( \bar{Y}_{ik}-\beta^r_{D,1}(\alpha)- \beta^r_{D,2}(\alpha) Z_{jk})    \right)\right]\\
     & =_{(2)} \sum_{k=1}^{K}\sum_{\substack{ik=1\\jk \neq i_k}}^{n_k} S_{ik,jk} \left[ \bar{Y}_{ik}(Z_{jk}=1,\alpha)+\bar{Y}_{ik}(Z_{jk}=0,\alpha)-2\bar{Y}_{ik}(Z_{jk}=0,\alpha)\right.\\
     &\left.-(\bar{Y}_{ik}(Z_{jk}=1,\alpha)-\bar{Y}_{ik}(Z_{jk}=0,\alpha)) \right] =0 \\
    \end{split}
\end{equation*}
\end{small}
Here, (1) follows from the definition of ${\xi}_A$ in Proposition~\ref{ase_conserve_variance_est}, and (2) uses \eqref{explicit_ase_beta_r}. Similarly,
\begin{small}
\begin{equation*}
\label{proof_ase_indep_xi_covariate}
    \begin{split}
     & \mathbb{E}\large( V^T_{D,2} B_D {\xi}_D)=\mathbb{E} (\sum_{k=1}^K \sum_{ik=1}^{n_k} \sum_{jk \neq ik} B_{ik,jk} \xi_{ik} \large) \\
     & =\sum_{k=1}^K \sum_{\substack{ik=1\\jk \neq ik}}^{n_k} \mathbb{E} \left[ Z_{jk}B_{ik,jk} \left( \bar{Y}_{ik}-\beta^r_{D,1}(\alpha)- \beta^r_{D,2}(\alpha) Z_{jk})    \right)\right]\\
     & = \sum_{k=1}^{K}\sum_{\substack{ik=1\\jk \neq ik}}^{n_k} S_{ik,jk} \left[ \bar{Y}_{ik}(Z_{jk}=1,\alpha)-\bar{Y}_{ik}(Z_{jk}=0,\alpha)-(\bar{Y}_{ik}(Z_{jk}=1,\alpha)-\bar{Y}_{ik}(Z_{jk}=0,\alpha)) \right] =0 \\
    \end{split}
\end{equation*}
\end{small}
For the sender estimator, from \eqref{eq:matrix_A_F} and the relation $\tilde{S}_{jk}S_{jk}=\sum_{ik\neq jk} S_{ik|jk}S_{jk}= \sum_{ik\neq jk} S_{ik,jk}$, it follows that
\begin{small}
\begin{equation*}
\label{proof_ase_indep_xi_covariate}
    \begin{split}
     & \mathbb{E}\large( V^T_{S,1} B_S {\xi}_S)=\mathbb{E} (\sum_{k=1}^K \sum_{jk=1}^{n_k}  \tilde{S}_{jk} S_{jk} \xi_{ik} \large) = \sum_{k=1}^K \sum_{jk=1}^{n_k} \mathbb{E} \left[\tilde{S}_{jk} S_{jk} \left( {Y}_{ik}-\beta^r_{S,1}(\alpha)- \beta^r_{S,2}(\alpha) Z_{jk})    \right)\right]\\
     & = \sum_{k=1}^{K}\sum_{\substack{ik=1\\jk \neq ik}}^{n_k} S_{ik,jk} \left[ \bar{Y}_{ik}(Z_{jk}=1,\alpha)+\bar{Y}_{ik}(Z_{jk}=0,\alpha)-2\bar{Y}_{ik}(Z_{jk}=0,\alpha)\right.\\
     &\left.-(\bar{Y}_{ik}(Z_{jk}=1,\alpha)-\bar{Y}_{ik}(Z_{jk}=0,\alpha)) \right] =0 \\
    \end{split}
\end{equation*}
\end{small}
Similarly, 
\begin{small}
\begin{equation*}
\label{proof_ase_indep_xi_covariate}
    \begin{split}
     & \mathbb{E}\large( V^T_{S,2} B_S {\xi}_S)=\mathbb{E} (\sum_{k=1}^K \sum_{jk=1}^{n_k} Z_{jk} \tilde{S}_{jk} S_{jk} \xi_{ik} \large) = \sum_{k=1}^K \sum_{jk=1}^{n_k} \mathbb{E} \left[Z_{jk}\tilde{S}_{jk} S_{jk} \left( {Y}_{ik}-\beta^r_{S,1}(\alpha)- \beta^r_{S,2}(\alpha) Z_{jk})    \right)\right]\\
     & = \sum_{k=1}^{K}\sum_{\substack{ik=1\\jk \neq ik}}^{n_k} S_{ik,jk} \left[ \bar{Y}_{ik}(Z_{jk}=1,\alpha)+\bar{Y}_{ik}(Z_{jk}=0,\alpha)-2\bar{Y}_{ik}(Z_{jk}=0,\alpha)\right.\\
     &\left.-(\bar{Y}_{ik}(Z_{jk}=1,\alpha)-\bar{Y}_{ik}(Z_{jk}=0,\alpha)) \right] =0 \\
    \end{split}
\end{equation*}
\end{small}

For the receiver estimator, from \eqref{ase_est_numerator_R}, replacing $\mathbf{Y}_R$ by $\boldsymbol{\xi}_R$, we have
\begin{small}
\begin{equation}
\label{proof_ase_indep_xi_covariate_int1}
    \begin{split}
     & \mathbb{E}( V^T_{R,1} B_R {\xi}_R)=\mathbb{E} \sum_{k=1}^K \sum_{ik=1}^{n_k}  \left[(\sum_{jk\neq ik} Z_{jk} B_{ik,jk}) \xi^1_{ik} + (\sum_{jk\neq ik} (1-Z_{jk}) B_{ik,jk}) \xi^0_{ik}\right] \\
     & = \mathbb{E} \sum_{k=1}^K \sum_{ik=1}^{n_k}  \left[\sum_{jk\neq ik} Z_{jk} B_{ik,jk}(Y_{ik}-\beta^r_{1}(\alpha)- \beta^r_{2}(\alpha)) + \sum_{jk\neq ik} (1-Z_{jk}) B_{ik,jk} (Y_{ik}-\beta^r_{1}(\alpha))\right]\\
     & =_{(1)} \sum_{k=1}^{K}\sum_{\substack{ik=1\\jk \neq ik}}^{n_k} S_{ik,jk} \left[ \bar{Y}_{ik}(Z_{jk}=1,\alpha)- \bar{Y}_{ik}(Z_{jk}=0,\alpha)-(\bar{Y}_{ik}(Z_{jk}=1,\alpha)-\bar{Y}_{ik}(Z_{jk}=0,\alpha))\right.\\
     &\left.+(\bar{Y}_{ik}(Z_{jk}=0,\alpha)-Y_{ik}(Z_{jk}=0,\alpha)) \right] =0 \\
    \end{split}
\end{equation}
\end{small}
where $\xi^1_{i_k} = Y_{i_k}-\beta^r_{1}(\alpha)-\beta^r_{2}(\alpha)$ and $\xi^0_{i_k}=Y_{i_k}-\beta^r_{1}(\alpha)$. (1) follows by substituting the expressions in \eqref{explicit_ase_beta_r}. Similarly, 
\begin{small}
\begin{equation*}
    \begin{split}
     & \mathbb{E}( V^T_{R,2} B_R {\xi}_R)=\mathbb{E} \sum_{k=1}^K \sum_{ik=1}^{n_k} (\sum_{jk\neq ik} Z_{jk} B_{ik,jk}) \xi^1_{ik}  \\
     & = \mathbb{E} \sum_{k=1}^K \sum_{ik=1}^{n_k}  \left[\sum_{jk\neq ik} Z_{jk} B_{ik,jk}(Y_{ik}-\beta^r_{1}(\alpha)- \beta^r_{2}(\alpha)) \right]\\
     & = \sum_{k=1}^{K}\sum_{\substack{ik=1\\jk \neq ik}}^{n_k} S_{ik,jk} \left[ \bar{Y}_{ik}(Z_{jk}=1,\alpha)- \bar{Y}_{ik}(Z_{jk}=0,\alpha)-(\bar{Y}_{ik}(Z_{jk}=1,\alpha)-\bar{Y}_{ik}(Z_{jk}=0,\alpha))\right]=0.\\
    \end{split}
\end{equation*}
\end{small}
Combining the above results establishes that 
\[
\mathbb{E}(V_{A,a}^{\top} B_A {\xi}_A)=0
\quad \text{for all } a\in\{1,2\}, \; A\in\{D,S,R\}.
\]
\end{proof}

\begin{lemma}
\label{mean_zero_CSE}
 For estimators of CSE, we have $\mathbb{E}(V^T_{A,a}B_A{\xi}_A)=0$ and ${\xi}_A$ is defined in Theorem \ref{CLT_est} for $a\in \{1,\cdots, 4\}$ and $A \in \{D,R,S\}$.   
\end{lemma}
\begin{proof}
 For $A \in \{D,S\}$, 
\begin{small}
\begin{equation}
\label{proof_proof_CLT_est_int3}
    \begin{split}
     & \mathbb{E}\large( V^T_{A,1} B_A {\xi}_A)=\mathbb{E} (\sum_{k=1}^K \sum_{ik=1}^{n_k} \sum_{jk \neq ik} B_{ik,jk} \xi_{ik} \large) \\
     & =_{(1)} \sum_{k=1}^K \sum_{\substack{ik=1\\jk \neq ik}}^{n_k} \mathbb{E} \left[B_{ik,jk} \left( Y_{ik}-\beta^r_{A,1}(\alpha)- \beta^r_{A,2}(\alpha) \tilde{X}_{jk}- (\beta^r_{A,3}(\alpha)+ \beta^r_{A,4}(\alpha)\tilde{X}_{jk} )\frac{1}{2} (Z_{jk}-(1-Z_{jk}))    \right)\right]\\
     & =_{(2)} \sum_{k=1}^{K}\sum_{\substack{ik=1\\jk \neq i_k}}^{n_k} S^r_{ik,jk} \left[ \bar{Y}_{ik}(Z_{jk}=1,\alpha)+\bar{Y}_{ik}(Z_{jk}=0,\alpha)-\beta_{A,1}^r(\alpha) \right]-0-0 =_{(3)}0 \\
    \end{split}
\end{equation}
\end{small}
(1) is by plugging in the formula of $\xi_{A}$ in Theorem \ref{CLT_est} and the definitions of $Z^*_{jk}$ and $(Z_{jk}\tilde{X}_{jk})^*$ in \eqref{transformed_causal_component}. (2) is by 
\begin{equation}
\label{sum_demean_0}
    \begin{split}
    \sum_{k=1}^K \sum_{ik=1}^{n_k} \sum_{jk \neq ik} S^r_{ik,jk}\tilde{X}_{jk}=0.   
    \end{split}
\end{equation}
and 
\begin{equation}
\label{Z_star_0}
\mathbb{E}(B_{ik,jk}Z^*_{jk})=\mathbb{E}\left[B_{ik,jk} \frac{1}{2}(Z_{jk}-(1-Z_{jk}))\right]= 0
\end{equation}
$(3)$ is by equation \eqref{proof_Const_ratio_estimable_quantity_int1}.
\begin{small}
\begin{equation*}
    \begin{split}
     & \mathbb{E}\large( V^T_{A,2} B_A {\xi}_A)=\mathbb{E} (\sum_{k=1}^K \sum_{ik=1}^{n_k} \sum_{jk \neq ik} \tilde{X}_{jk} B_{ik,jk} \xi_{ik} \large) \\
     & = \sum_{k=1}^K \sum_{\substack{ik=1\\jk \neq ik}}^{n_k} \mathbb{E} \left[B_{ik,jk} \tilde{X}_{jk}\left( Y_{ik}-\beta^r_{A,1}(\alpha)- \beta^r_{A,2}(\alpha) \tilde{X}_{jk}- (\beta^r_{A,3}(\alpha)+ \beta^r_{A,4}(\alpha)\tilde{X}_{jk} )\frac{1}{2} (Z_{jk}-(1-Z_{jk}))    \right)\right]\\
     & =_{(1)} \sum_{k=1}^{K}\sum_{\substack{ik=1\\jk \neq i_k}}^{n_k} S^r_{ik,jk} \left[\tilde{X}_{jk} \bar{Y}_{ik}(Z_{jk}=1,\alpha)+\tilde{X}_{jk}\bar{Y}_{ik}(Z_{jk}=0,\alpha)\right]- \sum_{k=1}^{K}\sum_{\substack{ik=1\\jk \neq i_k}}^{n_k} S^r_{ik,jk}\tilde{X}^{2}_{jk}\beta^r_{A,2}(\alpha)=_{(2)}0 \\
    \end{split}
\end{equation*}
\end{small}
$(1)$ is by \eqref{sum_demean_0} and \eqref{Z_star_0}. 
$(2)$ is by plugging in the formula of $\beta^r_{A,2}(\alpha)$ in \eqref{proof_Const_ratio_estimable_quantity_int1}.
\begin{small}
\begin{equation*}
    \begin{split}
     & \mathbb{E}\large( V^T_{A,3} B_A {\xi}_A)=\mathbb{E} (\sum_{k=1}^K \sum_{ik=1}^{n_k} \sum_{jk \neq ik} {Z}^*_{jk} B_{ik,jk} \xi_{ik} \large) \\
     & = \sum_{k=1}^K \sum_{\substack{ik=1\\jk \neq ik}}^{n_k} \mathbb{E} \left[B_{ik,jk} {Z}^*_{jk}\left( Y_{ik}-\beta^r_{A,1}(\alpha)- \beta^r_{A,2}(\alpha) \tilde{X}_{jk}- (\beta^r_{A,3}(\alpha)+ \beta^r_{A,4}(\alpha)\tilde{X}_{jk} )Z^*_{jk}   \right)\right]\\
     & =_{(1)} \sum_{k=1}^{K}\sum_{\substack{ik=1\\jk \neq i_k}}^{n_k} \frac{1}{2}S^r_{ik,jk} \left[ \bar{Y}_{ik}(Z_{jk}=1,\alpha)-\bar{Y}_{ik}(Z_{jk}=0,\alpha)\right]- 0-\sum_{k=1}^{K}\sum_{\substack{ik=1\\jk \neq i_k}}^{n_k} \frac{1}{2} S^r_{ik,jk}(\beta^r_{A,3}(\alpha)+ \beta^r_{A,4}(\alpha)\tilde{X}_{jk} )   \\
     & =_{(2)} \sum_{k=1}^{K}\sum_{\substack{ik=1\\jk \neq i_k}}^{n_k} \frac{1}{2}S^r_{ik,jk} \left[ \bar{Y}_{ik}(Z_{jk}=1,\alpha)-\bar{Y}_{ik}(Z_{jk}=0,\alpha)\right]-  \frac{1}{2} \beta^r_{A,3}(\alpha)-0=0
    \end{split}
\end{equation*}
\end{small}
$(1)$ is by $\mathbb{E}(B_{ik,jk}Z^*_{jk})=0$ and $Z^{*2}_{jk}=\frac{1}{4}$. (2) is by \eqref{sum_demean_0} and plugging in the formula of $\beta^{r}_{A,3}(\alpha)$ from equation in \eqref{proof_Const_ratio_estimable_quantity_int1}.
\begin{small}
\begin{equation*}
    \begin{split}
     & \mathbb{E}\large( V^T_{A,4} B_A {\xi}_A)=\mathbb{E} (\sum_{k=1}^K \sum_{ik=1}^{n_k} \sum_{jk \neq ik} ({Z}_{jk} \tilde{X}_{jk})^* B_{ik,jk} \xi_{ik} \large) \\
     & = \sum_{k=1}^K \sum_{\substack{ik=1\\jk \neq ik}}^{n_k} \mathbb{E} \left[B_{ik,jk} {Z}^*_{jk} \tilde{X}_{jk} \left( Y_{ik}-\beta^r_{A,1}(\alpha)- \beta^r_{A,2}(\alpha) \tilde{X}_{jk}- (\beta^r_{A,3}(\alpha)+ \beta^r_{A,4}(\alpha)\tilde{X}_{jk} )Z^*_{jk}   \right)\right]\\
     & =_{(1)} \sum_{k=1}^{K}\sum_{\substack{ik=1\\jk \neq ik}}^{n_k} \frac{1}{2}S^r_{ik,jk}\tilde{X}_{jk} \left[ \bar{Y}_{ik}(Z_{jk}=1,\alpha)-\bar{Y}_{ik}(Z_{jk}=0,\alpha)\right]- 0-\sum_{k=1}^{K}\sum_{\substack{ik=1\\jk \neq ik}}^{n_k} \frac{1}{2} (\beta^r_{A,3}(\alpha)\tilde{X}_{jk} + \beta^r_{A,4}(\alpha)\tilde{X}^2_{jk} )   \\
     & =_{(2)} \sum_{k=1}^{K}\sum_{\substack{ik=1\\jk \neq ik}}^{n_k} \frac{1}{2}S^r_{ik,jk} \left[ \bar{Y}_{ik}(Z_{jk}=1,\alpha)-\bar{Y}_{ik}(Z_{jk}=0,\alpha)\right]-0-  \sum_{k=1}^{K}\sum_{\substack{ik=1\\jk \neq ik}}^{n_k} \frac{1}{2} \tilde{X}^2_{jk} \beta^r_{A,4}(\alpha)=0
    \end{split}
\end{equation*}
\end{small}
$(1)$ is by \eqref{Z_star_0}. $(2)$ is by plugging in the formula in \eqref{proof_Const_ratio_estimable_quantity_int1}. 

For $A\in \{R\}$, based on \eqref{proof_relation_est_CSE_int3} and replacing $Y_R$ by $\xi_R$, we have 
\begin{small}
\begin{equation}
\label{proof_proof_CLT_est_int4}
    \begin{split}
    & \mathbb{E}(V^T_{R,1}B_R\xi_R)= \sum_{k=1}^K \sum_{ik=1}^{n_k} \mathbb{E} \left[ B^1_{ik} ({Y}_{ik}-\beta^{r}_{R,1}(\alpha)- \beta^{r}_{R,3}(\alpha))+ B^0_{ik} (Y_{ik}-\beta^{r}_{R,1}(\alpha)) \right] \\
    & =_{(1)} \sum_{k=1}^{K}\sum_{\substack{ik=1\\jk \neq ik}}^{n_k} S^r_{ik,jk}\left\lbrace \left[ \bar{Y}_{ik}(Z_{jk}=1,\alpha)-\bar{Y}_{ik}(Z_{jk}=0,\alpha) -(\bar{Y}_{ik}(Z_{jk}=1,\alpha)-\bar{Y}_{ik}(Z_{jk}=0,\alpha)) \right] \right.\\
    & \left.+ \bar{Y}_{ik}(Z_{jk}=0,\alpha)-\bar{Y}_{ik}(Z_{jk}=0,\alpha)] \right\rbrace=0 \\
    \end{split}
\end{equation}
\end{small}
$(1)$ is based on equation \eqref{proof_Const_ratio_estimable_quantity_int2} and the fact that the first $N$ entries of $\xi_R$ is different from the second $N$ entries of $\xi_R$ due to the different entries in $V_R$ and $B_R$, although first $N$ entries of $Y_R$ is the same as the second $N$ entries of $Y_R$. 
\begin{small}
\begin{equation*}
    \begin{split}
    & \mathbb{E}(V^T_{R,2}B_R\xi_R)= \sum_{k=1}^K \sum_{ik=1}^{n_k} \mathbb{E} \left[ B^1_{ik}  (Y_{ik}-\beta^{r}_{R,1}(\alpha)- \beta^{r}_{R,3}(\alpha))\right]=_{(1)} 0\\
    \end{split}
\end{equation*}
\end{small}
$(1)$ is based on the formulas for $\beta^{r}_{R,1}(\alpha)$ and $\beta^{r}_{R,3}(\alpha)$ in equation 
\eqref{proof_Const_ratio_estimable_quantity_int2}.
\begin{small}
\begin{equation}
\label{proof_proof_CLT_est_int5}
    \begin{split}
    & \mathbb{E}(V^T_{R,3}B_R\xi_R)= \sum_{k=1}^K \sum_{ik=1}^{n_k} \mathbb{E} \left[B^1_{ik} X^\dagger_{ik} (Y_{ik}-\beta^{r}_{R,2}(\alpha){X}^\dagger_{ik} - \beta^{r}_{R,4}(\alpha) X^{\dagger}_{ik} )+ B^0_{ik} X^\dagger_{ik} (Y_{ik}-\beta^{r}_{R,2}(\alpha) X^{\dagger}_{ik} ) \right] \\
    & =  \sum_{k=1}^K \sum_{ik=1}^{n_k}\sum_{jk \neq ik} S^r_{ik,jk} \left[\bar{Y}_{ik}(Z_{jk}=1,\alpha)X^\dagger_{ik}- \beta^r_{R,2}(\alpha) X^{\dagger 2}_{ik} - \beta^r_{R,4}(\alpha) X^{\dagger 2}_{ik} \right.\\
    & \left.+ \bar{Y}_{ik}(Z_{jk}=0,\alpha)X^{\dagger}_{ik} - \beta^r_{R,2}(\alpha) X^{\dagger 2}_{ik} \right]    \\
    & =_{(1)} \sum_{k=1}^K \sum_{ik=1}^{n_k}\sum_{jk \neq ik} S^r_{ik,jk} \left[ \bar{Y}_{ik}(Z_{jk}=1,\alpha)X^\dagger_{ik} - \bar{Y}_{ik}(Z_{jk}=0,\alpha) X^\dagger_{ik} \right.\\
    & \left.- (\bar{Y}_{ik}(Z_{jk}=1,\alpha)-\bar{Y}_{ik}(0,\alpha))X^\dagger_{ik} + \bar{Y}_{ik}(Z_{jk}=0,\alpha) X^\dagger_{ik}-\bar{Y}_{ik}(Z_{jk}=0,\alpha) X^\dagger_{ik}  \right]=0
    \end{split}
\end{equation}
\end{small}
(1) is based on plugging the formula of $\beta^r_{R,2}$ and $\beta^r_{R,4}$ respectively. 
\begin{small}
\begin{equation*}
    \begin{split}
    & \mathbb{E}(V^T_{R,4}B_R\xi_R)= \sum_{k=1}^K \sum_{ik=1}^{n_k}\mathbb{E} \left[  B^1_{ik}X^\dagger_{ik} (Y_{ik}-\beta^{r}_{R,1}(\alpha)X^\dagger_{ik}- \beta^{r}_{R,4}(\alpha)X^\dagger_{ik}) \right]=_{(1)} 0\\
    \end{split}
\end{equation*}
\end{small}
(1) is by the third row of \eqref{proof_proof_CLT_est_int5}.    
\end{proof}

\end{appendix}

\bibliographystyle{plainnat}
\bibliography{reference_regression}

\end{document}